\documentclass[onefignum,onetabnum]{siamart190516}

\usepackage{lipsum}
\usepackage{xspace}
\usepackage{bold-extra}
\usepackage[most]{tcolorbox}
\usepackage[utf8]{inputenc} 
\usepackage[T1]{fontenc}    
\usepackage{hyperref}       
\usepackage{cleveref}
\usepackage{url}            
\usepackage{booktabs}       
\usepackage{amsfonts}       
\usepackage{amssymb}
\usepackage{nicefrac}       
\usepackage{microtype}      
\usepackage{xcolor}         
\usepackage{comment}
\usepackage{empheq}
\usepackage{enumitem}
\usepackage{algpseudocode}
\usepackage{placeins}
\usepackage{multirow}
\usepackage{xcolor,colortbl}
\usepackage{stmaryrd}
\usepackage{graphicx,epstopdf} 
\usepackage[caption=false]{subfig}

\crefname{opC}{Line}{Lines}
\Crefname{opC}{Line}{Lines}

\setlist[itemize]{leftmargin=20pt}
\newsiamremark{remark}{Remark}
\newsiamthm{assumption}{Assumption}
\newsiamremark{example}{Example}
\crefname{assumption}{Assumption}{Assumption}
\crefname{example}{Example}{Example}
\crefname{definition}{Definition}{Definition}
\newsiamthm{claim}{Claim}

\newcommand{\mc}[1]{\mathcal{#1}}
\newcommand{\mb}[1]{\mathbb{#1}}
\newcommand{\mbf}[1]{\mathbf{#1}}
\newcommand{\te}[1]{\tilde{#1}}

\newcommand{\vertiii}[1]{{\left\vert\kern-0.25ex\left\vert\kern-0.25ex\left\vert #1 
		\right\vert\kern-0.25ex\right\vert\kern-0.25ex\right\vert}}

\usepackage{amsmath}

\DeclareMathOperator*{\argmin}{arg\,min}
\newcommand{\supp}[0]{\mathtt{supp}}

\newcommand\intset[1]{\llbracket #1 \rrbracket}
\newcommand\col[1]{{\bullet, {#1}}}
\newcommand\row[1]{{{#1}, \bullet}}
\newcommand\RG[1]{\textcolor{black}{#1}}

\newcommand\mS{\mc{S}}

\title{Spurious Valleys, NP-hardness, and Tractability \\
	of Sparse Matrix Factorization With Fixed Support}

\author{Quoc-Tung Le\thanks{Univ Lyon, ENS de Lyon, UCBL, CNRS, Inria, LIP, F-69342, LYON Cedex 07, France (\email{quoc-tung.le@ens-lyon.fr, elisa.riccietti@ens-lyon.fr, remi.gribonval@inria.fr}).}
	\and Elisa Riccietti\footnotemark[1]
	\and Remi Gribonval\footnotemark[1]}

\headers{Fixed support matrix factorization}{QUOC T. LE, ELISA RICCIETTI, REMI GRIBONVAL}

\usepackage{amsopn}

\ifpdf
\hypersetup{
	pdftitle={An Example Article},
	pdfauthor={D. Doe, P. T. Frank, and J. E. Smith}
}
\fi

\ifpdf
\hypersetup{ pdftitle={Fixed Support Matrix Factorization} }
\fi

\begin{document}
\maketitle

\begin{tcbverbatimwrite}{tmp_\jobname_abstract.tex}
\begin{abstract}
  The problem of approximating a dense matrix by a product of sparse factors is a fundamental problem for many signal processing and machine learning tasks. 
  It can be decomposed into two subproblems: finding the position of the non-zero coefficients in the sparse factors, and determining their values. While the first step is usually seen as the most challenging one due to its combinatorial nature, this paper focuses on the second step, referred to as sparse matrix approximation with fixed support. 
  First, we show its NP-hardness, while also presenting a nontrivial family of supports making the problem practically tractable with a dedicated algorithm. Then, we investigate the landscape of its natural optimization formulation, proving the absence of spurious local valleys and spurious local minima, whose presence could prevent local optimization methods to achieve global optimality. The advantages of the proposed algorithm over state-of-the-art first-order optimization methods are discussed.
\end{abstract}


\end{tcbverbatimwrite}
\input{tmp_\jobname_abstract.tex}

\section{Introduction}

Matrix factorization with sparsity constraints is the problem of approximating a (possibly dense) matrix as the product of two or more sparse factors. It is playing an important role in many domains and applications such as dictionary learning and signal processing \cite{DictionaryLearning,DoubleSparseFactorization,SparseFactorization}, linear operator acceleration \cite{lemagoarou:hal-01158057,ChasingButterfly,candesFourier}, deep learning \cite{DBLP:journals/corr/abs-1903-05895,Dao2020Kaleidoscope:,chen2022pixelated}, to mention only a few. 
{
Given a matrix $Z$, sparse matrix factorization can be expressed as the optimization problem:
\begin{equation}
	\centering
	\label{eq:matrix_factorization}
	\begin{aligned}
		\underset{X^1, \ldots, X^N}{\text{Minimize}} \;& \|Z - X^1\ldots X^N\|_F^2\\
		\text{subject to: }& \text{constraints on } \supp(X_i), \;\forall 1 \leq i \leq N  
	\end{aligned}
\end{equation}
where $\supp(X) := \{(i,j) \mid X_{i,j} \neq 0\}$ is the set of indices whose entries are nonzero. }

{For example, one can employ generic sparsity constraints such as $|\supp(X_i)| \leq k_i, 1 \leq i \leq N$ where $k_i$ controls the sparsity of each factor. More structured types of sparsity (for example, sparse rows/ columns) can also be easily encoded since the notion of support $\supp(X)$ captures completely the sparsity structure of a factor.}

{In general}, Problem \eqref{eq:matrix_factorization} is challenging due to its non-convexity as well as the discrete nature of $\supp(X_i)$ (which can lead to an exponential number of supports to consider). Existing algorithms to tackle it directly comprise heuristics such as Proximal Alternating Linearization Minimization (PALM) \cite{bolte:hal-00916090,lemagoarou:hal-01158057} and its variants \cite{papierICASSP21}.

In this work, we consider a {restricted class of instances of Problem \eqref{eq:matrix_factorization}}, in which just two factors are considered {($N = 2$)} and with {\emph{prescribed supports}}. We call this problem {\em fixed support (sparse) matrix factorization} (FSMF).  In details, given a matrix $A \in \mb{R}^{m \times n}$, we look for two sparse factors $X,Y$ that solve the following problem: 
\begin{equation}
	\label{eq: fixed_supp_prob}
	\tag{FSMF}
	\begin{aligned}
		& \underset{X \in \mb{R}^{m \times r}, Y \in \mb{R}^{n \times r}}{\text{Minimize}} 
		& &L(X, Y) = \|A - XY^\top\|^2&\\
		& \text{Subject to: }
		& &\texttt{supp}(X) \subseteq I\ \text{and}\ 
		\texttt{supp}(Y) \subseteq J
	\end{aligned}
\end{equation}
where $\|\cdot\|$ is the Frobenius norm, $I\subseteq \intset{m}  \times\intset{r} $, $J\subseteq \intset{n}  \times \intset{r} $\footnote{$\forall m\in\mathbb{N}$, $\intset{m}:=\{1,\dots,m\}$} are given support constraints, i.e., $\texttt{supp}(X) \subseteq I$ implies that $\forall (i,j) \notin I, X_{ij} = 0$. 

The main aim of this work is to investigate the theoretical properties of \eqref{eq: fixed_supp_prob}. To the best of our knowledge the analysis of matrix factorization problems with fixed supports has never been addressed in the literature. This analysis is however interesting, for the following reasons:
{
\begin{enumerate}[leftmargin=*]
	\item The asymptotic behaviour of heuristics such as PALM \cite{bolte:hal-00916090,lemagoarou:hal-01158057} when applied to Problem \eqref{eq:matrix_factorization} can be characterized by studying the behaviour of the method on an instance of \eqref{eq: fixed_supp_prob}. Indeed, PALM updates the factors alternatively by a projected gradient step onto the set of the constraints. It is experimentally observed that for many instances of the problem, the support becomes constant after a certain number of iterations. Let us illustrate this on an instance of Problem \eqref{eq:matrix_factorization} with $N = 2, X^i \in \mb{R}^{100 \times 100}, i = 1, 2$ and the constraints $|\supp(X^i)| \leq 1000, i = 1, 2$. In this setting, running PALM is equivalent to an iterative method in which we consecutively perform one step of gradient descent for \emph{each} factor, while keeping the other fixed, and project that factor  onto $\{X \mid X \in \mb{R}^{100 \times 100}, |\supp(X)| \leq 1000\}$ by simple hard-thresholding\footnote{Code for this experiment can be found in \cite{codeSLV}}. \Cref{fig: PALMbehavior} illustrates the evolution of the difference between the support of each factor before and after each iteration of PALM through {$5000$} iterations (the difference between two sets $B_1$ and $B_2$ is measured by $|(B_1 \setminus B_2) \cup (B_2 \setminus B_1)|$).  We observe that when the iteration counter is large enough, the factor supports do not change (or equivalently they become \emph{fixed}): further iterations of the algorithm simply optimize an instance of \eqref{eq: fixed_supp_prob}. Therefore, to develop a more precise understanding of the possible convergence of PALM in such a context, it is necessary to understand properties of  \eqref{eq: fixed_supp_prob}. {For instance, in this example, once the supports stop to change, the factors $(X^1_n, X^2_n)$ converge inside this fixed support (\Cref{fig: PALMbehavior}c)}. {However, there are cases in which PALM generates iterates $(X^1_n,X^2_n)$  diverging to infinity due to the presence of a \emph{spurious local valley} in the landscape of $L(X,Y)$ (cf \cref{rem:divergencePALM})}. This is not in conflict with the convergence results for PALM in this context \cite{bolte:hal-00916090,lemagoarou:hal-01158057} since these are established \emph{under the assumption of bounded iterates}. 
	\begin{figure}[h]
		\centering
		\includegraphics[scale=0.22]{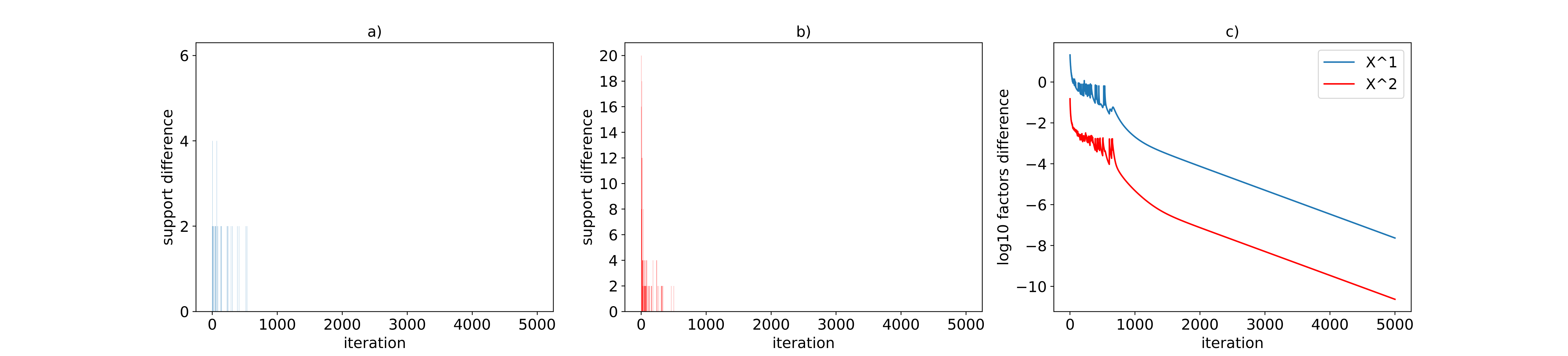}
		\caption{{Support change for the first (a) and the second (b) factor in PALM. The norm of the difference between two consecutive factors updates is depicted in (c) (logarithmic scale).}}
		\label{fig: PALMbehavior}
	\end{figure}
	\item While \eqref{eq: fixed_supp_prob} is just a class of the general problem \eqref{eq:matrix_factorization}, its coverage includes many other interesting problems:
	\begin{itemize}[leftmargin=*, label = $\bullet$]
		\item \textbf{Low rank matrix approximation (LRMA) \cite{eckart1936approximation}: } By taking $I = \intset{m} \times \intset{r}$, $J = \intset{n} \times \intset{r}$, addressing \eqref{eq: fixed_supp_prob} is equivalent to looking for the best rank $r$ matrix approximating $A$, cf. \Cref{fig:lmra_lu}(a). We will refer to this instance in the following as the full support case. This problem is known to be polynomially tractable, cf.  \Cref{sec:easyinstance}. This work enlarges the family of supports for which \eqref{eq: fixed_supp_prob} remains tractable.
		\item $\mbf{LU}$ \textbf{decomposition \cite[Chapter 3.2]{matrixcomputation}:} Considering $m = n= r$ and $I = J = \{(i,j) \mid 1 \leq j \leq i \leq n\}$, it is easy to check that \eqref{eq: fixed_supp_prob} is equivalent to factorizing $A$ into a lower and an upper triangular matrix ($X$ and $Y$ respectively, cf. \Cref{fig:lmra_lu}(b)), and in this case, the \emph{infimum} of \eqref{eq: fixed_supp_prob} is always zero. It is worth noticing that there exists a non-empty set of matrices for which this infimum is not attained (or equivalently matrices which do not admit the $\mbf{LU}$ decomposition \cite{matrixcomputation}). This behaviour will be further discussed in \Cref{sec: NPhard} and \Cref{sec: spuriouslocal}. More importantly, our analysis of \eqref{eq: fixed_supp_prob} will cover the non-zero infimum case as well.
		\begin{figure}[h]
			\centering
			\includegraphics[scale=0.32]{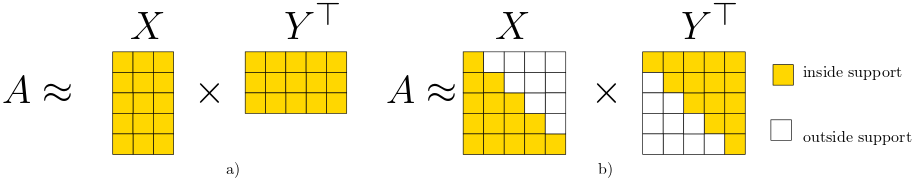}
			\caption{Illustrations for(a) LRMA and (b) $\mbf{LU}$ decomposition as instances of \eqref{eq: fixed_supp_prob}.}
			\label{fig:lmra_lu}
		\end{figure}
		\item \textbf{Butterfly structure and fast transforms \cite{DBLP:journals/corr/abs-1903-05895,chen2022pixelated,Dao2020Kaleidoscope:,ChasingButterfly,candesFourier}:} Many linear operators admit fast algorithms since their associated matrices can be written as a product of sparse factors whose supports are known to possess the \emph{butterfly structure} (and they are \emph{known} in advance). This is the case for instance of the Discrete Fourier Transform (DFT) or the Hadamard transform (HT). For example, a Hadamard transform of size $2^N \times 2^N$ can be written as the product of $N$ factors of size $2^N \times 2^N$ whose factors have two non-zero coefficients per row and per column. \Cref{fig: hadamard} illustrates such a factorization for $N = 3$.
		\begin{figure}[h]
			\centering
			\includegraphics[scale=0.45]{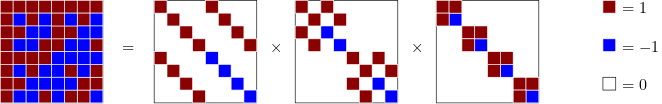}
			\caption{The factorization of the Hadamard transform of size $8 \times 8\; (N = 3)$.}
			\label{fig: hadamard}
		\end{figure}
		Although our analysis of \eqref{eq: fixed_supp_prob} only deals with $N = 2$, the butterfly structure allows one to reduce to the case $N = 2$ in a recursive\footnote{{{While revising this manuscript we heard about the work of Dao et al \cite{monarch} introducing the ``Monarch'' class of structured matrices, essentially corresponding to the first stage of the recursion from \cite{papierICASSP,papierLeon}.}}} manner \cite{papierICASSP,papierLeon}. 
		\item \textbf{Hierarchical $\mc{H}$-matrices \cite{Hackbusch1999ASM,Hackbusch2000ASH}:} We prove in \Cref{app:hierarchical} that the class of hierarchically off-diagonal low-rank (HODLR) matrices (defined in \cite[Section 3.1]{HODLR}, \cite[Section 2.3]{Hackbusch1999ASM}), a subclass of hierarchical  $\mc{H}$-matrices, can be expressed as the product of two factors with fixed supports, that are illustrated on Figure~\ref{fig:hodlr}. Therefore, the task of finding the best $\mc{H}$-matrix from this class to approximate a given matrix is reduced to \eqref{eq: fixed_supp_prob}.
		\begin{figure}[h]
			\centering
			\includegraphics[scale=0.35]{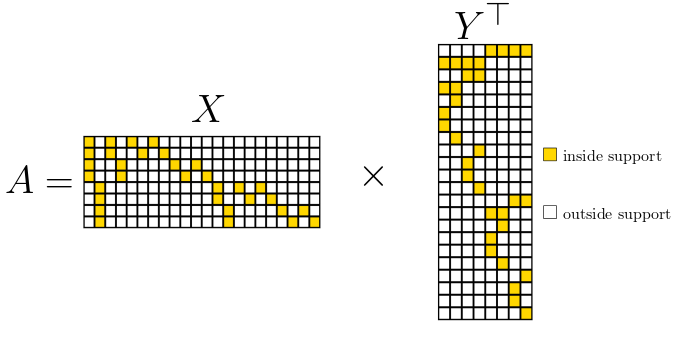}
			\caption{Two fixed supports for factors of a HODLR matrix of size $8 \times 8$ illustration based on analysis of \Cref{app:hierarchical}.}
			\label{fig:hodlr}
		\end{figure}
		\item \textbf{Matrix completion:} We show that matrix completion can be reduced to \eqref{eq: fixed_supp_prob}, which is the main result of \Cref{sec: NPhard}.
	\end{itemize}
\end{enumerate}
}
Our aim is to then study the theoretical properties of \eqref{eq: fixed_supp_prob} and in particular to assess its difficulty.  
{This leads us to consider four
complementary aspects}. 

First, we show the NP-hardness of \eqref{eq: fixed_supp_prob}. While this result contrasts with the theory established for coefficient recovery with a fixed support in the classical sparse recovery problem (that can be trivially addressed by least squares), it is in line with the known hardness of related matrix factorization with additional constraints or different losses. Indeed, famous variants of matrix factorization such as non-negative matrix factorization (NMF) \cite{NMFExactNPHard,NMFRankNPhard}, weighted low rank \cite{WLRANPHard} and matrix completion \cite{WLRANPHard} were all proved to be NP-hard. We prove the NP-hardness by reduction from the Matrix Completion problem with noise. To our knowledge this proof is new and cannot be trivially deduced from any existing result on the more
classical full support case.

{
Second, we show that besides its NP-hardness, problem \eqref{eq: fixed_supp_prob} also shares some properties with another hard problem: low-rank tensor approximation \cite{TensorRank}. Similarly to the classical example of \cite{TensorRank}, which shows that the set of rank-two tensors is not closed, we show that there are support constraints $I,J$ such that the set of matrix products $XY^\top$ with ``feasible'' $(X,Y)$ (i.e.,  $\{XY^\top \mid \supp(X) \subseteq I, \supp(Y) \subseteq J\}$),  is not a closed set. Important examples are the supports $(I,J)$ for which \eqref{eq: fixed_supp_prob} corresponds to $\mathbf{LU}$ matrix factorization. For such support constraints, there exists a matrix $A$ such that the infimum of $L(X,Y)$ is zero and can only be approached if either $X$ or $Y$ have at least an arbitrarily large coefficient. This is precisely one of the settings leading to a diverging behavior of PALM (cf \Cref{rem:divergencePALM}). }

{Third,} we show that despite the hardness of \eqref{eq: fixed_supp_prob} in the general case, many pairs of support constraints $(I,J)$ make the problem solvable by an effective direct algorithm based on the block singular value decomposition (SVD). The investigation of those supports is also covered in this work and  a dedicated polynomial algorithm is proposed to deal with this family of supports. This includes for example the full support case. Our analysis of tractable instances of \eqref{eq: fixed_supp_prob} actually includes and substantially generalizes the analysis of the instances that can be classically handled with the SVD decomposition. In fact, the presence of the constraints on the support makes it impossible to directly use the SVD to solve the problem, because coefficients outside the support have to be zero. However, the presented family of support constraints allows for an iterative decomposition of the problem into "blocks" that can be exploited to build up an optimal solution using blockwise SVDs. This technique  can be seen in many sparse representations of matrices (for example, hierarchical {$\mc{H}$-}matrices \cite{Hackbusch1999ASM, Hackbusch2000ASH}) to allow fast matrix-vector and matrix-matrix multiplication. 


The {fourth} contribution of this paper is the study of the landscape of the objective function $L$ of \eqref{eq: fixed_supp_prob}. Notably, we investigate the existence of \emph{spurious local minima} and \emph{spurious local valleys}, which will be collectively referred to as \emph{spurious objects}. They will be formally introduced in \Cref{sec: spuriouslocal}, but intuitively these objects may represent a challenge for the convergence of local optimization methods.

The {global} landscape of the loss functions for {matrix decomposition related problems (matrix sensing \cite{LRMR, LRMO}, phase retrieval \cite{phaseretrieval}, matrix completion \cite{MatrixCompletionNoSLM, NoLocalMinimaGeometryAnalysis,localminimaanalysiskernelPCA}) and neural network training (either with linear \cite{DBLP:journals/corr/abs-1805-04938, NIPS2016_6112, venturi2020spurious} or non-linear activation functions \cite{nguyen2019connected, nguyen2017loss})} 
has been a popular subject of study recently. 
These works have direct link to ours since matrix factorization {\em without any support constraint} can be seen {either as a matrix decomposition problem or as} a specific case of neural network (with two layers, no bias and linear activation function). 
 Notably it has been proved \cite{DBLP:journals/corr/abs-1805-04938} that for linear neural networks, every local minimum is a global minimum and if the network is shallow (i.e., there is only one hidden layer), critical points are either global minima or strict saddle points (i.e., their Hessian have at least one --\emph{strictly}-- negative eigenvalue). However, there is still a \textit{tricky} type of landscape that could represent a challenge for local optimization methods and has not been covered until recently: spurious local valleys \cite{nguyen2019connected,venturi2020spurious}.  {In particular, the combination of these results shows the benign landscape for LMRA, a particular instance of \eqref{eq: fixed_supp_prob}.}

{However}, to the best of our knowledge, existing analyses of {landscape} are only proposed for {neural network training in general and matrix factorization problem in particular} \emph{without support constraints}, cf.  \cite{DBLP:journals/corr/abs-1805-04938,venturi2020spurious,NIPS2016_6112}, while the study of the landscape of \eqref{eq: fixed_supp_prob} remains untouched in the literature and our work can be considered as a generalization of such previous results.  {Moreover, unlike many existing results of matrix decomposition problems that are proved to hold with high probability under certain random models \cite{LRMR,LRMO,phaseretrieval,MatrixCompletionNoSLM,NoLocalMinimaGeometryAnalysis,localminimaanalysiskernelPCA, reviewLRMF}), our result deterministically ensures the benign landscape for \emph{each} matrix $A$, under certain conditions on the support constraints $(I,J)$}.

To summarize, our main contributions in this paper are:
\begin{itemize} 
	\item[1)] We prove that \eqref{eq: fixed_supp_prob} is NP-hard in \cref{theorem: NPhardfixedsupport}. {In addition, in light of classical results on the $\mbf{LU}$ decomposition, we highlight  in \Cref{sec: NPhard}
	a challenge related to the possible non-existence of an optimal solution of \eqref{eq: fixed_supp_prob} .} 
	\item[2)] We introduce  families of support constraints $(I,J)$ making \eqref{eq: fixed_supp_prob} tractable  (\cref{theorem:disjoint_totally_overlapping} and \cref{theorem: reduction_disjoint_overlapping}) and provide dedicated polynomial algorithms for those families.
	
	\item[3)]  We show that the landscape of \eqref{eq: fixed_supp_prob} corresponding to the support pairs $(I,J)$ in these families are free of spurious local valleys, regardless of the factorized matrix $A$ (\autoref{theorem:noSpuriousSimple}, \autoref{th:MainNoSpuriousComplex}). We also investigate the presence of spurious local minima for such families (\autoref{theorem:noSpuriousSimple}, \autoref{theorem:nospuriousminimaComplex}). 
	
	\item[4)] These results might suggest a conjecture that holds true for the full support case: an instance of \eqref{eq: fixed_supp_prob} is tractable if and only if {its} corresponding landscape is benign, i.e. free of spurious objects. We give a counter-example to this conjecture (\cref{ex:spuriousinstances})
	and  illustrate numerically that even with support constraints ensuring a benign landscape, state-of-the-art gradient descent methods can be significantly slower than the proposed dedicated algorithm.
\end{itemize}

\subsection{Notations}
\label{sec:notation}
For $n \in \mb{N}$, define $\llbracket n\rrbracket:= \{1, \ldots, n\}$. The notation $\mbf{0}$ (resp. $\mbf{1}$) stands for a  matrix with all zeros (resp. all  ones) coefficients. The identity matrix of size $n \times n$ is denoted by $\mbf{I}_n$. Given a matrix $A \in \mb{R}^{m \times n}$ and $T \subseteq \llbracket n\rrbracket$, $A_{\bullet, T} \in \mb{R}^{m \times |T|}$  is the submatrix of $A$ {restricted} to the columns indexed in $T$ while $A_T \in \mb{R}^{m \times n}$ is the matrix that has the same columns as $A$ for {indices} in $T$ and is zero elsewhere. If $T = \{k\}$ is a singleton, $A_{\bullet, T}$ is simplified as $A_{\bullet, k}$ (the $k^{th}$ column of $A$). For $(i,j) \in \intset{m} \times \intset{n}, A_{i,j}$ is the coefficient of $A$ at index $(i,j)$.
If $S \subseteq \llbracket m\rrbracket, T \subseteq \llbracket n\rrbracket$, then $A_{S, T} \in \mb{R}^{|S| \times |T|}$ is the submatrix of $A$ {restricted} to rows and columns indexed in $S$ and $T$ respectively.

A support constraint $I$ on a matrix $X \in \mb{R}^{m \times r}$ can be interpreted either as a subset $I \subseteq \llbracket m\rrbracket \times \llbracket r \rrbracket$ or as its indicator matrix $1_I \in \{0,1\}^{m \times r}$ defined as: $(1_I)_{i,j} = 1$ if $(i,j) \in I$ and  $0$ otherwise. Both representations will be used interchangeably and the meaning should be clear from the context. For $T \subseteq \llbracket r\rrbracket$, we use the notation $I_{T} := I \cap \left( \llbracket m \rrbracket \times T\right)$ (this is {consistent} with the notation $A_T$ introduced earlier).

The notation $\supp(A)$ is used for both vectors and matrices: if $A \in \mb{R}^m$ is a vector, then $\supp(A) = \{i \mid A_i \neq 0\} \subseteq \intset{m}$; if $A \in \mb{R}^{m \times n}$ is a matrix, then $\supp(A) = \{(i,j) \mid A_{i,j} \neq 0\} \subseteq \intset{m} \times \intset{n}$.
Given two matrices $A, B \in \mb{R}^{m\times n}$, the Hadamard product $A \odot B$ between $A$ and $B$ is defined as $(A \odot B)_{i,j} = A_{i,j}B_{i,j}, \forall (i,j) \in \llbracket m\rrbracket \times \llbracket n \rrbracket$. Since a support constraint $I$ of a matrix $X$ can be thought of as a binary matrix of the same size, we define $X \odot I := X \odot 1_I$ analogously (it is a matrix whose coefficients in $I$ are unchanged while the others are set to zero).

\section{Matrix factorization with fixed support is NP-hard}
\label{sec: NPhard}

To show that \eqref{eq: fixed_supp_prob} is NP-hard we use the classical technique to prove NP-hardness: reduction. Our choice of reducible problem is matrix completion with noise \cite{WLRANPHard}.

\begin{definition}[Matrix completion with noise  \cite{WLRANPHard}]
	\label{def: matrixcompletion}
	Let $W \in \{0,1\}^{m \times n}$ be a binary matrix. Given $A \in \mb{R}^{m \times n}, s \in \mathbb{N}$, the matrix completion problem (MCP) is:
	\begin{equation}\label{eq:MCP}
		\tag{MCP}
		\underset{X \in \mb{R}^{m \times s}, Y \in \mb{R}^{n \times s}}{\emph{\text{Minimize }}} \|A - XY^\top\|_W^2 = \|(A - XY^\top) \odot W\|^2.
	\end{equation}
\end{definition}
This problem is NP-hard even when $s = 1$ \cite{WLRANPHard} by its reducibility from Maximum-Edge Biclique Problem, which is NP-complete \cite{MaximumBicliqueNPComplete}. This is given in the following theorem:

\begin{theorem}[NP-hardness of matrix completion with noise \cite{WLRANPHard}]
	\label{theorem:MCPhard}
	Given a binary weighting matrix $W \in \{0,1\}^{m \times n}$ and $A \in [0,1]^{m \times n}$, the optimization problem
	\begin{equation}
		\label{eq: WLRA1}
		\tag{MCPO}
		\underset{{x \in \mb{R}^m, y \in \mb{R}^n}}{\emph{\text{Minimize }}} \|A - xy^\top\|_W^2.
	\end{equation}
	is called rank-one matrix completion problem (MCPO). Denote $p^*$ the infimum of \eqref{eq: WLRA1} and let $\epsilon = 2^{-12}(mn)^{-7}$. It is NP-hard to find an approximate solution with objective function accuracy less than $\epsilon$, i.e. with objective value $p \leq p^* + \epsilon$.
\end{theorem}


The following lemma gives a reduction from  \eqref{eq: WLRA1} to \eqref{eq: fixed_supp_prob}.

\begin{lemma}
	\label{lem: NPhardness}
	For any binary matrix $W \in \{0,1\}^{m \times n}$, there exist an integer $r$ and two sets $I$ and $J$ such that for all $A \in \mb{R}^{m \times n}$, \eqref{eq: WLRA1} and \eqref{eq: fixed_supp_prob} share the same infimum. $I$ and $J$ can be constructed in polynomial time. Moreover, if one of the  problems has a known solution that provides objective function accuracy $\epsilon$, we can find a solution with the same accuracy for the other one in polynomial time.
\end{lemma}

\begin{proof}[Proof sketch]
	Up to a transposition, we can assume without loss of generality that $m \geq n$. Let $r = n + 1 = \min(m,n) + 1$.
	We define $I \in \{0,1\}^{m \times (n + 1)}$ and $J \in \{0,1\}^{n \times (n + 1)}$ as follows:
	\begin{equation*}
		\begin{aligned}
			I_{i,j} &= \begin{cases}
				1 - W_{i,j} & \text{if } j \neq n\\
				1 & \text{if } j = n + 1
			\end{cases},
			J_{i,j} &= \begin{cases}
				1 & \text{if } j = i \text{ or } j = n + 1\\
				0 & \text{otherwise}
			\end{cases}
		\end{aligned}
	\end{equation*}
	This construction can clearly be made in polynomial time. 
	We show in the supplementary material (\cref{sm:NPhardness}) that the two problems share the same infimum.
\end{proof}

Using \cref{lem: NPhardness}, we obtain a result of NP-hardness for \eqref{eq: fixed_supp_prob} as follows. 
\begin{theorem}
	\label{theorem: NPhardfixedsupport}
	When $A \in [0,1]^{m \times n}$, it is NP-hard to solve \eqref{eq: fixed_supp_prob} with arbitrary index sets $I,J$ and objective function accuracy less than $\epsilon = 2^{-12}(mn)^{-7}$.
\end{theorem}

\begin{proof}
	Given any instance of \eqref{eq: WLRA1} (i.e., two matrices $A \in [0,1]^{m \times n}$ and $W \in \{0,1\}^{m \times n}$), we can produce an instance of \eqref{eq: fixed_supp_prob} (the same  matrix $A$ and $I\in \{0,1\}^{m \times r}, J \in \{0,1\}^{n \times r}$) such that both have the same infimum (\cref{lem: NPhardness}). Moreover, for any given objective function accuracy, we can use the procedure of \cref{lem: NPhardness} to make sure the solutions of both problems share the same accuracy. 
	
	Since all procedures are polynomial, this defines a polynomial reduction from \eqref{eq: WLRA1} to \eqref{eq: fixed_supp_prob}. Because \eqref{eq: WLRA1} is NP-hard to obtain a solution with objective function accuracy less than $\epsilon$ (\cref{theorem:MCPhard}), so is \eqref{eq: fixed_supp_prob}.
\end{proof}

We point out that, while the result is interesting on its own, for some applications, such as those arising in machine learning, the accuracy bound $O((mn)^{-7})$ may not be really appealing.  We thus keep as an interesting open research direction to determine if some precision threshold exists that make the general problem easy.

\Cref{lem: NPhardness} constructs a hard instance where $(I,J) \in \{0,1\}^{m \times r} \times \{0,1\}^{n \times r}$ and $r = \min(m,n) + 1$. It is also interesting to investigate the hardness of \eqref{eq: fixed_supp_prob} given a fixed $r$. When $r = 1$, the problem is polynomially tractable since this case is  covered by \Cref{theorem:disjoint_totally_overlapping} below. On the other hand, when $r \geq 2$, the question becomes complicated due to the fact that the set $\{XY^\top \mid \supp(X) \subseteq I, \supp(Y) \subseteq J\}$ is not always closed. In \Cref{rem: no_global_min}, we show an instance of \eqref{eq: fixed_supp_prob} where the infimum is zero but cannot be attained. Interestingly enough, this is exactly the example for the non-existence of an exact $\mbf{LU}$ decomposition of a matrix in $\mb{R}^{2 \times 2}$ presented in \cite[Chapter 3.2.12]{matrixcomputation}. We emphasize that this is not a mere consequence of the non-coercivity of $L(X,Y)$ -- which follows from rescaling invariance, see e.g. \Cref{rem:nostrict} -- as we will also present support constraints for which the problem always admits a global minimizer and can be solved with an efficient algorithm. More generally, one can even show that the set $\mc{L}$ of square matrices of size $n \times n$ having an exact $\mbf{LU}$ decomposition ({\em i.e.}, $\mc{L}:= \{XY^\top \mid \supp(X) \subseteq I, \supp(J) \subseteq J\}$ where $I = J = \{(i,j) \mid 1 \leq j \leq i \leq n\}$) is {open and} dense in $\mb{R}^{n \times n}$ (since a matrix having all non-zero leading principal minors admits an exact $\mbf{LU}$ factorization \cite[Theorem 3.2.1]{matrixcomputation}) but $\mc{L} \subsetneq \mb{R}^{n \times n}$. Thus, $\mc{L}$ is not closed.

{Furthermore, one might wonder whether the pathological cases consists of a ``zero measure'' set, as many results for deterministic matrix completion problem \cite{algebraiccombLRMC,universalmatrixcompletion} can be established for ``almost all instances''. Our examples on the \textbf{LU} decomposition seem to corroborate this hypothesis as well. Nevertheless, for the problem of tensor decomposition, which is very closely related to ours,  \cite{TensorRank} showed the converse: the pathological cases related to the projection of a \emph{real} tensor of size $2 \times 2 \times 2$ to the set of rank two tensors consists of an open subset of $\mb{R}^{2 \times 2 \times 2}$, thus ``non-negligible''. The answer also changes depending on the underlying field ($\mb{R}$ or $\mb{C}$) of the tensor/matrix \cite{complexbestrankr,algebraiccombLRMC}. Given the richness of this topic, we leave this question open as a future research direction.}

\section{Tractable instances of matrix factorization with fixed support}
\label{sec:easyinstance}
Even though \eqref{eq: fixed_supp_prob} is generally NP-hard, when we consider the full support case $I = \llbracket m\rrbracket \times \llbracket r\rrbracket, J = \llbracket n\rrbracket \times \llbracket r \rrbracket$ 
the problem is equivalent to 
LRMA \cite{eckart1936approximation}, which can be solved using the Singular Value Decomposition (SVD) \cite{SVDcompute}\footnote{SVD can be computed to machine precision in $O(mn^2)$ \cite{kumar2016literature}, see also \cite[Lecture 31, page 236]{trefethen1997numerical}. It is thus convenient to think of LRMA as polynomially solvable.}. This section is devoted to enlarge the family of supports for which \eqref{eq: fixed_supp_prob} can be solved by an effective direct algorithm {based on blockwise SVDs}. We start with an important definition:

\begin{definition}[Support of rank-one contribution]
	\label{def: support_rank_one_synthesis}
	Given two support constraints  $I\in\{0,1\}^{m\times r}$ and $J\in\{0,1\}^{n\times r}$ of \eqref{eq: fixed_supp_prob} and $k \in \llbracket r \rrbracket$, we define the $k^{th}$ rank-one contribution support $\mc{S}_k(I,J)$ (or in short, $\mS_k$) as: $\mc{S}_k(I,J) = I_{\col{k}}J_{\col{k}}^{\top}$. This can be seen either as:	a tensor product: $\mS_k \in \{0,1\}^{m \times n}$ is a binary matrix or a Cartesian product: $\mS_k$ is a set of matrix indices defined as $ \supp(I_\col{k}) \times  \supp(J_\col{k})$.
\end{definition}

Given a pair of support constraints $I, J$, if $\texttt{supp}(X) \subseteq I, \texttt{supp}(Y) \subseteq J$, we have: $\texttt{supp}(X_\col{k}Y_\col{k}^\top) \subseteq \mS_k,\; \forall k \in \llbracket r\rrbracket$. Since $XY^\top = \sum_{k = 1}^r X_\col{k} Y_\col{k}^\top$
the notion of contribution support $\mS_k$ captures the constraint on the support of the $k^{th}$ \emph{rank-one contribution}, $X_\col{k} Y_\col{k}^\top$, of the matrix product $XY^\top$ (illustrated in \cref{fig:rank1supp}).

{In the case of full supports ($\mS_k = \mbf{1}_{m \times n}$ for each $k \in \intset{r}$), the optimal solution can be obtained in a greedy manner: indeed, it is well known that 
\Cref{algorithm0} computes factors achieving the best rank-$r$ approximation to $A$ (notice that here the algorithm also works for complex-valued matrices):
\begin{algorithm}[H]
	\centering
	\caption{Generic Greedy Algorithm} 
	\label{algorithm0}
	\begin{algorithmic}[1]
		\Require $A \in \mathbb{R}^{m \times n}$ or $\mathbb{C}^{m \times n}$; $\{\mS_{k}\}_{k \in \intset{r}}$ rank-one supports 
		\For {$i \in \intset{r}$}
		\State $(X_\col{i}, Y_\col{i}) = (u,v)$  where $uv^{\top}$ is any best rank-one approximation to $A \odot \mS_{i}$ \label{algo0:line2}			
		\State $A = A - X_\col{i}, Y_\col{i}^\top$	
		\EndFor
		\State \Return $(X,Y)$
	\end{algorithmic}
\end{algorithm}
Even beyond the full support case, the output of \Cref{algorithm0} always satisfies the support constraints due to line~\ref{algo0:line2}, however it may not always be the optimal solution of \eqref{eq: fixed_supp_prob}. 
Our analysis of the polynomial tractability conducted below will 
allow us to show that, under appropriate assumptions on $I,J$, one can compute in polynomial time an  optimal solution of \eqref{eq: fixed_supp_prob} using variants of \Cref{algorithm0}. The definition of these variants will involve a partition of $\llbracket r \rrbracket$ in terms of equivalence classes of rank-one supports:}


\begin{figure}[h]
	\centering
	\includegraphics[scale = 0.3]{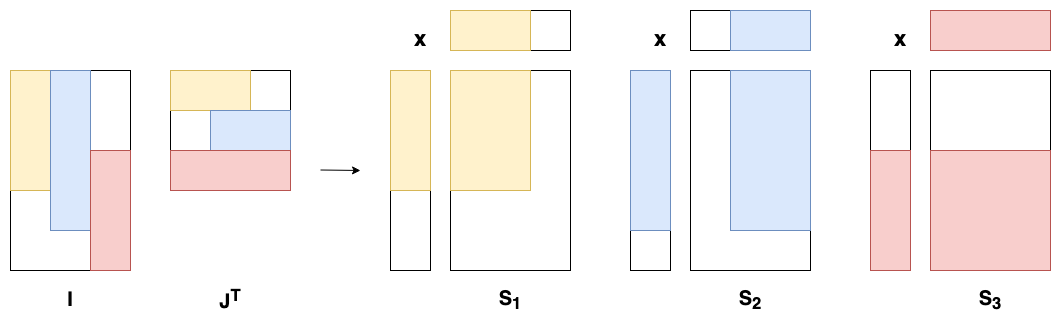}
	\caption{Illustration the idea of support of rank-one contribution. Colored rectangles indicate the support constraints $(I,J)$ and the support constraints $\mS_k$ on each component matrix $X_\col{k}Y_\col{k}^\top$.}
	\label{fig:rank1supp}
\end{figure}
\begin{figure}
	\centering
	\includegraphics[scale=0.4]{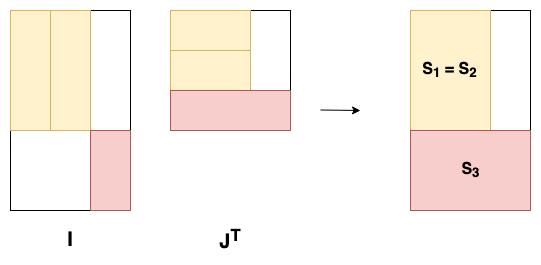}
	\caption{An instance of support constraints $(I,J)$ satisfying \cref{theorem:disjoint_totally_overlapping}. We use colored rectangles to indicate the support constraints $(I,J)$. The indices belonging to the same equivalence class share the same color.}
	\label{fig:tractable1}
\end{figure}
\begin{definition}[Equivalence classes of rank-one supports, representative rank-one supports] 
	\label{def:completeeqclass}
	Given $I \in \{0,1\}^{m \times r}$, $J \in \{0,1\}^{n \times r}$, define an equivalence relation on $\llbracket r \rrbracket$ as: $i\sim j$ if and only if $\mS_i =  \mS_j$ (or equivalently $(I_\col{i},J_\col{i})=(I_\col{j},J_\col{j})$). This yields a partition of $\llbracket r \rrbracket$ into equivalence classes.

	Denote $\mathcal{P}$ the collection of equivalence classes.
	For each class $P \in \mathcal{P}$ denote $\mS_{P}$ a representative rank-one support, $R_{P} \subseteq \intset{m}$ and $C_{P} \subseteq \intset{n}$ the supports of rows and columns in $\mS_{P}$, respectively. For every $k \in P$ we have $ \mS_k = \mS_{P}$ and $\supp(I_\col{k}) = R_{P}$, $\supp(J_\col{k})=C_{P}$. 
	
	For every $\mathcal{P}' \subseteq \mathcal{P}$ denote $\mS_{\mathcal{P}'} = \cup_{P \in \mathcal{P}'} \mS_{P} \subseteq \intset{m} \times \intset{n}$ and $\bar{\mS}_{\mathcal{P}'} = (\intset{m} \times \intset{n}) \backslash \mS_{\mathcal{P}'}$.
\end{definition}
For instance, in the example in \cref{fig:rank1supp} we have three distinct {equivalence} classes. {With the introduction of equivalence classes, one can modify \Cref{algorithm0} to make it more efficient, as in \Cref{algorithm0.5}: Instead of computing the SVD $r$ times, one can simply compute it only $|\mc{P}|$ times. For the full support case, we have $\mc{P} = \{\intset{r}\}$, thus \Cref{algorithm0.5} is identical to the classical SVD.
\begin{algorithm}[H]
	\centering
	\caption{Alternative Generic Greedy Algorithm} 
	\label{algorithm0.5}
	\begin{algorithmic}[1]
		\Require $A \in \mathbb{R}^{m \times n}$ or $\mathbb{C}^{m \times n}$; $\{\mS_{P}\}_{P \in \mc{P}}$ representative rank-one supports
		\For {$P \in \mc{P}$}
		\State $(X_\col{P}, Y_\col{P}) = (U,V)$  where $UV^{\top}$ is any best rank-$|P|$ approximation to 
		$A \odot \mS_{P}$
		\State $A = A - X_\col{P}, Y_\col{P}^\top$	
		\EndFor
		\State \Return $(X,Y)$
	\end{algorithmic}
\end{algorithm}
}
A first simple sufficient condition ensuring the tractability of an instance of \eqref{eq: fixed_supp_prob} is {stated in the following theorem. }

\begin{theorem}
	\label{theorem:disjoint_totally_overlapping}
	Consider $I \in \{0,1\}^{m \times r}$, $J \in \{0,1\}^{n \times r}$, and $\mc{P}$ the collection of equivalence classes of \cref{def:completeeqclass}. If the representative rank-one supports are pairwise disjoint, i.e., $\mS_P \cap \mS_{P'} = \emptyset$ for each distinct $P, P' \in \mc{P}$, then matrix factorization with fixed support is tractable for any $A\in \mathbb{R}^{m \times n}$.
\end{theorem}

\begin{proof}
In this proof, for each equivalent class $P \in \mc{P}$ (\cref{def:completeeqclass}) we use the notations $X_P \in \mb{R}^{m \times r}, Y_P \in \mb{R}^{n \times r}$ (introduced in \Cref{sec:notation}). We also use the notations $R_P,C_P$ (\cref{def:completeeqclass}). For each equivalent class $P$, we have:
\begin{equation}
	\label{eq:submatrixfac}
	(X_PY_P^\top)_{R_P,C_P} = X_{R_P,P} Y_{C_P,P}^\top
\end{equation}
and the product $XY^\top$ can be decomposed as: $XY^\top = \sum_{P \in \mc{P}} X_PY_P^\top$. Due to the hypothesis of this theorem, with $P, P' \in \mc{P}, P' \neq P$, we further have: 
\begin{equation}
	\label{eq:Hadamardproduct}
	X_{P'}Y_{P'}^\top \odot \mS_{P} = \mathbf{0}
\end{equation}

\begin{algorithm}[H]
	\centering
	\caption{Fixed support matrix factorization (under \cref{theorem:disjoint_totally_overlapping} assumptions)} \label{algorithm1}
	\begin{algorithmic}[1]
		\Procedure{SVD\_FSMF}{$A \in \mb{R}^{m \times n}, I \in \{0,1\}^{m \times r}, J \in \{0,1\}^{n \times r}$}
		\State Partition $\intset{r}$ into $\mc{P}$ (\cref{def:completeeqclass}) {to get $\{\mS_{P}\}_{P \in \mc{P}}$}
		\State \Return $(X,Y)$ {using \Cref{algorithm0.5} with input $A$, $\{\mS_{P}\}_{P \in \mc{P}}$}
		\EndProcedure
	\end{algorithmic}
\end{algorithm}

The objective function $L(X,Y)$ is:
\begin{equation}
	\label{eq:decomposedisjointoverlapping}
	\begin{split}
		\|A - XY^\top\|^2 & = \left(\sum_{P \in \mc{P}}\|(A -XY^\top) \odot \mS_{P}\|^2\right) + \|(A -XY^\top) \odot \bar{\mS}_\mc{P}\|^2\\
		&= \left(\sum_{P \in \mc{P}}\|(A - \sum_{P' \in \mc{P}} X_{P'}Y_{P'}^\top) \odot \mS_{P}\|^2\right) + \|(A - \sum_{P' \in \mc{P}} X_{P'}Y_{P'}^\top) \odot \bar{\mS}_\mc{P}\|^2\\
		&\overset{\eqref{eq:Hadamardproduct}}{=} \left(\sum_{P \in \mc{P}}\|(A - X_PY_P^\top) \odot \mS_{P}\|^2\right) + \|A \odot \bar{\mS}_\mc{P}\|^2\\
		&= \left(\sum_{P \in \mc{P}}\|A_{R_P,C_P} - (X_PY_P^\top)_{R_P,C_P}\|^2\right) + \|A \odot \bar{\mS}_\mc{P}\|^2\\
		&\overset{\eqref{eq:submatrixfac}}{=} \left(\sum_{P \in \mc{P}}\|A_{R_P,C_P} - X_{R_P,P}Y_{C_P,P}^\top\|^2\right) + \|A \odot \bar{\mS}_\mc{P}\|^2\\
	\end{split}
\end{equation}
Therefore, if we ignore the constant term $\|A \odot \bar{\mS}_\mc{P}\|^2$,  the function $L(X,Y)$ is decomposed into a sum of functions $\|A_{R_P,C_P} - X_{R_P,P}Y_{C_P,P}^\top\|^2$, which are LRMA instances. Since all the optimized parameters are $\{(X_{R_P,P}, Y_{C_P,P})\}_{P \in \mc{P}}$, an optimal solution of $L$ is $\{(X^\star_{R_P,P}, Y^\star_{C_P,P})\}_{P \in \mc{P}}$, where $(X^\star_{R_P,P}, Y^\star_{C_P,P})$ is a  minimizer of $\|A_{R_P,C_P} - X_{R_P,P}Y_{C_P,P}^\top\|^2$ {which is computed efficiently using a truncated SVD. Since the blocks associated to distinct $P$ are disjoint, these 
SVDs can be performed blockwise, in any order, and even in parallel.}
\end{proof}

For these easy instances, we can therefore recover the factors in polynomial time with the procedure described in \cref{algorithm1}. Given a target matrix $A \in \mb{R}^{m \times n}$ and support constraints $I \in \{0,1\}^{m \times r}$, $J \in \{0,1\}^{n \times r}$ satisfying the condition in \cref{theorem:disjoint_totally_overlapping}, \cref{algorithm1} returns two factors $(X,Y)$ solution of \eqref{eq: fixed_supp_prob}.  

{As simple as this condition is, it is satisfied in some important cases, for instance for a class of Hierarchical matrices (HODLR, cf.  \Cref{app:hierarchical}), or for the so-called \emph{butterfly supports}: in the latter case, the condition is used in \cite{papierICASSP,papierLeon} to design an efficient hierarchical factorization method, which is shown to outperform first-order optimization approaches commonly used in this context, in terms both of computational time and accuracy.
}


In the next result, we explore the tractability of \eqref{eq: fixed_supp_prob} while allowing partial intersection between two representative rank-one contribution supports.
\begin{definition}[Complete equivalence classes of rank-one supports - CEC] 
	\label{def:completeeqclass2}	
	$P \in \mc{P}$ is a \emph{complete equivalence class} (or \emph{CEC}) if $|P| \geq \min\{|C_{P}|,|R_{P}|\}$ with $C_{P},R_{P}$ as  in \cref{def:completeeqclass}. Denote $\mathcal{P}^{\star} \subseteq \mathcal{P}$ the family of all complete equivalence classes, $T = \cup_{P \in \mathcal{P}^{\star}} P \subseteq \llbracket r\rrbracket$, $\bar{T} = \llbracket r \rrbracket \backslash T$, and the shorthand $\mS_{T} = \mS_{\mathcal{P}^{\star}}$.
\end{definition}

The interest of complete equivalence classes is that their expressivity is powerful enough to represent any matrix whose support is included in $\mS_T$, as illustrated by the following lemma. 

\begin{lemma}
	\label{lem: expressibility_two}
	Given $I \in \{0,1\}^{m \times r}$, $J \in \{0,1\}^{n \times r}$, consider $T$, $\mS_T$ as in \cref{def:completeeqclass2}.
	For any matrix $A \in \mb{R}^{m \times n}$ such that $\supp(A) \subseteq \mS_{T}$, there exist $X \in \mb{R}^{m \times r}, Y \in \mb{R}^{n \times r}$ such that $A = XY^{\top}$ and
	$\supp(X) \subseteq I_T$, $\supp(Y) \subseteq J_T$. Such a pair can be computed using 
	{\Cref{algorithm1}} 
	{$(X,Y) = \text{SVD\_FSMF}(A,I_{T},J_{T})$.}	
\end{lemma}

The proof of \cref{lem: expressibility_two} is deferred to the supplementary material (\cref{subapp:expressibilitytwo}).
%
%
The next definition introduces the key properties that the indices $k \in \llbracket r \rrbracket$ which are not in any CEC need to satisfy in order to make \eqref{eq: fixed_supp_prob} overall tractable. 

\begin{definition}[Rectangular support outside CECs of rank-one supports] 
	\label{def: supp_outsideCEC}
	Given $I \in \{0,1\}^{m \times r}$, $J \in \{0,1\}^{n \times r}$, consider $T$ and $\mS_T$ as in \cref{def:completeeqclass2} and $\bar{T} = \llbracket r \rrbracket \setminus T$. For $k \in\bar T$ define the support outside CECs of the $k^{th}$ rank-one support.
	as: $\mS'_k = \mS_k \setminus \mS_T$. If $\mc{S}'_k = R_k \times C_k$ for some $R_k \subseteq \llbracket m\rrbracket, C_k \subseteq \llbracket n\rrbracket$, (or equivalently $\mS_k'$ is of rank at most one), we say the support outside CECs of the $k^{th}$ rank-one support $\mc{S}'_k$ is \emph{rectangular}. 
\end{definition}
To state our tractability result, we further categorize the indices in $I$ and $J$ as follows:
\begin{definition}[Taxonomy of indices of $I$ and $J$]
	\label{def:taxonomy}
	With the notations of \cref{def: supp_outsideCEC}, assume that $\mS'_k$ is rectangular for all $k\in \bar{T}$. We decompose the indices of $I$ (resp $J$) into three sets as follows:
	\begin{table}[H]
		\centering
		\begin{tabular}{ccc}
			\toprule
			& Classification for $I$ & Classification for $J$\\
			\midrule
			$1$ & $I_T = \{(i, k) \mid k \in T, i \in \llbracket m\rrbracket \} \cap I$ & $J_T = \{(j, k) \mid k \in T, j \in \llbracket n\rrbracket \} \cap J$\\
			\midrule
			$2$ & $I_{\bar{T}}^1 = \{(i, k) \mid k \notin T, i \in R_k\} \cap I$ & $J_{\bar{T}}^1 = \{(j, k) \mid k \notin T, j \in C_k\} \cap J$\\
			\midrule
			$3$ & $I_{\bar{T}}^2 = \{(i, k) \mid k \notin T, i \notin R_k\} \cap I$  & $J_{\bar{T}}^2 = \{(j, k) \mid k \notin T, j\notin C_k\} \cap J$\\
			\bottomrule
		\end{tabular}
	\end{table}
\end{definition}
The following theorem generalizes \cref{theorem:disjoint_totally_overlapping}.

\begin{theorem}
	\label{theorem: reduction_disjoint_overlapping}
	Consider $I \in \{0,1\}^{m \times r}$, $J \in \{0,1\}^{n \times r}$. Assume that for all $k \in \bar{T}$, $\mS'_k$ is rectangular and that for all $k, l \in \bar{T}$ 
	we have $\mS'_k = \mS'_l$ or $\mS'_k \cap \mS'_l = \emptyset$. Then, $(I_{\bar{T}}^1, J_{\bar{T}}^1)$ satisfy the assumptions of \cref{theorem:disjoint_totally_overlapping}. Moreover, for any matrix $A\in \mathbb{R}^{m \times n}$, two instances of \eqref{eq: fixed_supp_prob} with data $(A, I, J)$ and $(A \odot {\bar{\mS}_T}, I_{\bar{T}}^1, J_{\bar{T}}^1)$ respectively, share the same infimum. Given an optimal solution of one instance, we can construct the optimal solution of the other in polynomial time. {In other word, \eqref{eq: fixed_supp_prob} with $(A,I,J)$ is polynomially tractable.}
\end{theorem}

\cref{theorem: reduction_disjoint_overlapping} is proved in the supplementary material (\cref{subapp: reduction_disjoint_overlapping}). It implies that solving the problem with support constraints $(I,J)$ can be achieved by reducing to another problem, with  support constraints satisfying the assumptions of \cref{theorem:disjoint_totally_overlapping}. The latter problem can thus be efficiently solved by \cref{algorithm1}. In particular, \cref{theorem:disjoint_totally_overlapping} is a special case of \cref{theorem: reduction_disjoint_overlapping} when all the equivalent classes (including CECs) have disjoint representative rank-one supports. 

\cref{fig:tractable2} shows an instance of $(I,J)$ satisfying the assumptions of \cref{theorem: reduction_disjoint_overlapping}. {The extension in   \cref{theorem: reduction_disjoint_overlapping} is not directly motivated by concrete examples, but it is rather introduced as a first step to show that the family of polynomially tractable supports $(I,J)$ can be enlarged, as it is not restricted to just the family introduced in \cref{theorem:disjoint_totally_overlapping}.}
\begin{figure}
	\centering
	\includegraphics[scale = 0.4]{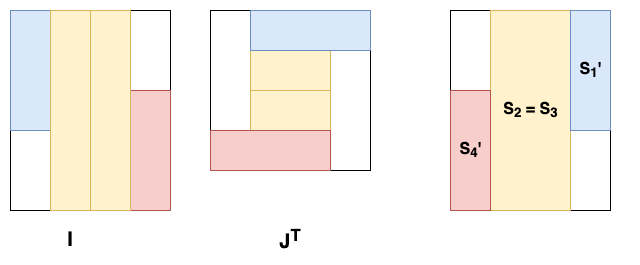}
	\caption{An instance of support constraints $(I,J)$ satisfying the assumptions of \cref{theorem: reduction_disjoint_overlapping}. We have $T = \{2,3\}$. The supports outside CEC $\mS_1'$ and $\mS_4'$ are disjoint.}
	\label{fig:tractable2}
\end{figure}
An algorithm for instances satisfying the assumptions of \cref{theorem: reduction_disjoint_overlapping} is given in \cref{algorithm3} (more details can be found in \cref{cor:optimalsol} and \cref{rem:optimalsol} in \cref{app:resultsec3} in the supplementary material). {In \cref{algorithm3}, two calls to 
 \cref{algorithm1} are made, they can be done in any order (Line~\ref{line:algo1} and Line~\ref{line:algo2} can be switched without changing the result).}
\begin{algorithm}[H]
	\centering
	\caption{Fixed support matrix factorization (under \cref{theorem: reduction_disjoint_overlapping}'s assumptions)}\label{algorithm3}
	\begin{algorithmic}[1]
		\Procedure{SVD\_FSMF2}{$A \in \mb{R}^{m \times n}, I \in \{0,1\}^{m \times r}, J \in \{0,1\}^{n \times r}$}
		\State Partition the indices of $I, J$ into $I_{T}, I_{\bar{T}}^1, I_{\bar{T}}^2$ (and $J_{T}, J_{\bar{T}}^1, J_{\bar{T}}^2$) (\cref{def: supp_outsideCEC}).
		\State {$(X_T, Y_T) = \text{SVD\_FSMF}(A \odot {\mS_T},I_{T},J_{T})$ ($T, \mS_T$ as in \cref{def:completeeqclass2}).}{\label{line:algo1}}
		\State{\label{line:algo2}}{$(X_{\bar{T}}^1, Y_{\bar{T}}^1) = \text{SVD\_FSMF}(A \odot {\bar{\mS_T}}, I_{\bar{T}}^1, J_{\bar{T}}^1)$} 
		\State \Return $(X_T + X_{\bar{T}}^1, Y_T + Y_{\bar{T}}^1)$
		\EndProcedure
	\end{algorithmic}
\end{algorithm}
 
\section{Landscape of matrix factorization with fixed support}
\label{sec: spuriouslocal}

In this section, we first recall the definition of \emph{spurious local valleys} and \emph{spurious local minima}, which are undesirable objects in the landscape of a function, as they may prevent local optimization methods to converge to globally optimal solutions. Previous works \cite{venturi2020spurious,DBLP:journals/corr/abs-1805-04938,NIPS2016_6112} showed that the landscape of the optimization problem associated to low rank approximation is free of such \emph{spurious objects}, which potentially gives the intuition for its tractability. 

We prove that similar results hold for the much richer family of tractable support constraints for \eqref{eq: fixed_supp_prob} that we introduced in \cref{theorem:disjoint_totally_overlapping}. The landscape with the assumptions of \cref{theorem: reduction_disjoint_overlapping} is also analyzed. These results might  suggest a natural conjecture: an instance of \eqref{eq: fixed_supp_prob} is tractable if and only if the landscape is benign. However, this is not true. We show an example that contradicts this conjecture: we show an instance of \eqref{eq: fixed_supp_prob} that can be solved efficiently, despite the fact that  its corresponding landscape contains spurious objects.  

\subsection{Spurious local minima and spurious local valleys}
\label{subsec: background}

We start by recalling the classical definitions of global and local minima of a real-valued function.

\begin{definition}[Spurious local minimum \cite{DBLP:journals/corr/abs-1805-04938,NumericalOptimization}]
	Consider $L: \mb{R}^{d} \to \mb{R}$. A vector $x^* \in \mb{R}^{d}$ is a:
	\begin{itemize}[leftmargin=*]
	\item \textbf{global minimum} (of $L$) if $L(x^*) \leq L(x), \forall x$.
	
	\item \textbf{local minimum} if there is a neighborhood $\mc{N}$ of $x^*$ such that $L(x^*) \leq L(x), \forall x \in \mc{N}$.
	
	\item \textbf{strict local minimum} if there is a neighborhood $\mc{N}$ of $x^*$ such that $L(x^*) < L(x), \forall x \in \mc{N}, x \neq x^*$.
	
	\item \textbf{(strict) spurious local minimum} if $x^*$ is a (strict) local minimum but it is not a global minimum. 
	\end{itemize}
\end{definition}

The presence of spurious local minima is undesirable because local optimization methods can get stuck  in one of them and never reach the global optimum. 
\begin{remark}\label{rem:nostrict}
With the loss functions $L(X,Y)$ considered in this paper, strict local minima do not exist since for every invertible diagonal matrix $D$, possibly arbitrarily close to the identity, we have $L(XD,YD^{-1}) = L(X,Y)$.
\end{remark}

However, this is not the only undesirable landscape in an optimization problem: spurious local valleys, as defined next, are also challenging.

\begin{definition}[Sublevel Set \cite{ConvexOptimization}]
	Consider $L: \mathbb{R}^{d} \to \mathbb{R}$. For every $\alpha \in \mathbb{R}$, the $\alpha$-level set of $L$ is the set $E_{\alpha} = \{x \in \mb{R}^{d} \mid L(x) \leq \alpha\}$.
\end{definition}

\begin{definition}[Path-Connected Set and Path-Connected Component]
	A subset $S \subseteq \mb{R}^{d}$ is path-connected if for every $x,y \in S$, there is a continuous function $r: [0,1] \to S$ such that $r(0) = x, r(1) = y$. A path-connected component of $E \subseteq \mb{R}^{d}$ is a maximal path-connected subset: $S \subseteq E $ is  path-connected, and if $S' \subseteq E$ is path-connected with $S \subseteq S'$ then $S=S'$.
\end{definition}
\begin{definition}[Spurious Local Valley \cite{venturi2020spurious,nguyen2019connected}]
	Consider $L: \mathbb{R}^{d} \to \mathbb{R}$ and a set $S \subset \mb{R}^{d}$.
	\begin{itemize}
	\item $S$ is a \textbf{local valley} of $L$ if it is a non-empty path-connected component of some sublevel set. 
	
	\item $S$ is a \textbf{spurious  local valley} of $L$ if it is a local valley of $L$ and does not contain a global minimum.
	\end{itemize}
\end{definition}

The notion of spurious local valley is inspired by the definition of a \emph{strict} spurious local minimum.
If $x^*$ is a strict spurious local minimum, then $\{x^*\}$ is a spurious local valley. However, the notion of spurious local valley has a wider meaning than just a neighborhood of a strict spurious local minimum. \cref{fig: spuriousfig} illustrates some other scenarios:
\begin{figure}[tbhp]
	\centering
	\subfloat{\label{fig:spuriouslocalvalley}\includegraphics[width=0.33\textwidth]{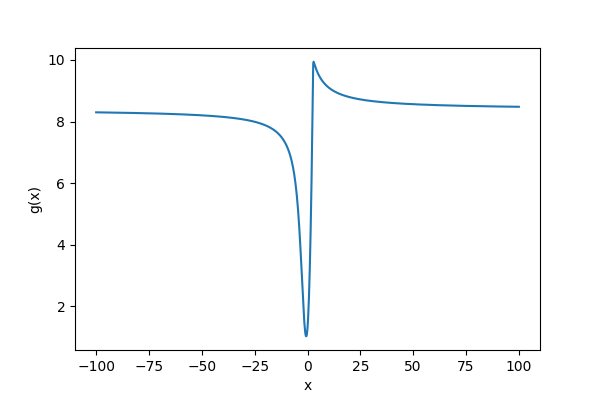}}
	\subfloat{\label{fig:spuriouslocalvalley2}\includegraphics[width=0.33\textwidth]{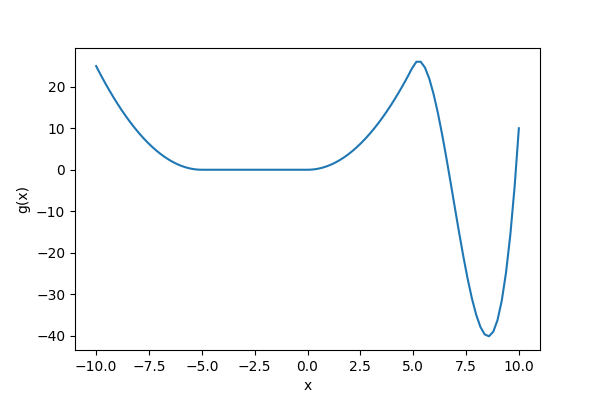}}
	\subfloat{\label{fig:spuriousfig2}\includegraphics[width=0.33\textwidth]{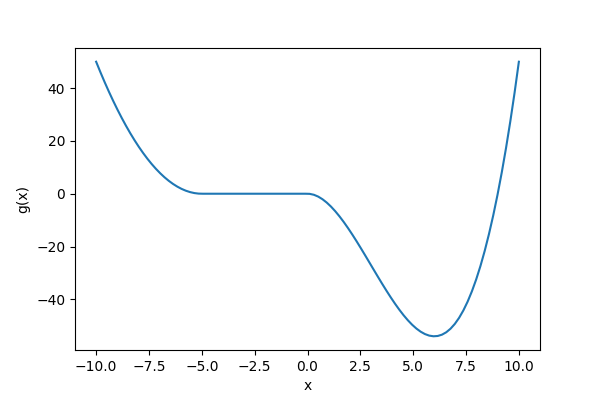}}
	\caption{Examples of functions with spurious objects.}
	\label{fig: spuriousfig}
\end{figure}
as shown on \cref{fig:spuriouslocalvalley}, the segment (approximately) $[10,+\infty)$ creates a spurious local valley, and this function has only one local (and global) minimizer, at zero; in \cref{fig:spuriouslocalvalley2}, there are spurious local minima that are not strict, but form a spurious local valley anyway. It is worth noticing that the concept of a spurious local valley does \emph{not} cover that of a spurious local minimum. Functions can have spurious (non-strict) local minima even if they do not possess any spurious local valley (\cref{fig:spuriousfig2}). Therefore, in this paper, we treat the existence of spurious local valleys and spurious local minima independently. The common point is that if the landscape possesses either of them, local optimization methods need to have proper initialization to have guarantees of convergence to a global minimum.

\subsection{Previous results on the landscape}
Previous works \cite{NIPS2016_6112, DBLP:journals/corr/abs-1805-04938} studied the non-existence of spurious local minima of \eqref{eq: fixed_supp_prob} in the classical case of ``low rank matrix approximation'' (or \emph{full support matrix factorization})\footnote{Since previous works also considered the case $r \geq m,n$, low rank approximation might be misleading sometimes. That is why we occasionally use the name full support matrix factorization to emphasize this fact., where no support constraints are imposed ($I = \llbracket m\rrbracket \times \llbracket r\rrbracket, J = \llbracket n\rrbracket \times \llbracket r\rrbracket$)}. 
To prove that a critical point 
is never a spurious local minimum, previous work used the notion of {\em strict saddle point} (i.e a point where the Hessian is not positive semi-definite, or equivalently has at least one --{\em strictly} -- negative eigenvalue), see \Cref{def:strictsaddle} below.
To prove the non-existence of spurious local valleys, the following lemma was employed in previous works \cite{venturi2020spurious,nguyen2019connected}:

\begin{lemma}[Sufficient condition for the non-existence of any spurious local valley {\cite[Lemma 2]{venturi2020spurious}}]
	\label{lemma: non_spurious_local_valley}
	Consider a continuous function $L:\mb{R}^d \to \mb{R}$. Assume that, for any initial parameter $\tilde{x} \in \mb{R}^d$, there exists a continuous path $f: t \in [0,1] \to \mb{R}^d$ such that:
	\begin{enumerate}
		\item[a)] $f(0) = \tilde{x}$.
		
		\item[b)] $f(1) \in \argmin_{x \in \mb{R}^d} L(x)$. 
		
		\item[c)] The function $L \circ f: t \in [0,1] \to \mb{R}$ is non-increasing.
	\end{enumerate}
	Then there is no spurious local valley in the landscape of function $L$.
\end{lemma}

The result is intuitive and a formal proof can be found in \cite{venturi2020spurious}. The theorem claims that given any initial point, if one can find a continuous path connecting the initial point to a global minimizer and the loss function is non-increasing on the path, then there does not exist any spurious local valley. We remark that although \eqref{eq: fixed_supp_prob} is a constrained optimization problem, \cref{lemma: non_spurious_local_valley} is still applicable because one can think of the objective function as defined on a subspace: $L: \mb{R}^{|I| + |J|} \to \mb{R}$. In this work, to apply \cref{lemma: non_spurious_local_valley}, the constructed function $f$ has to be a \emph{feasible path}, defined as:
\begin{definition}[Feasible path]
	\label{def:feasiblepath}
	A feasible path w.r.t the support constraints $(I,J)$ (or simply a feasible path) is a continuous function $f(t) = (X_f(t), Y_f(t)): [0,1] \to \mb{R}^{m \times r} \times \mb{R}^{n \times r}$ satisfying $\supp(X_f(t)) \subseteq I, \supp(Y_f(t)) \subseteq J, \forall t \in [0,1]$. 
\end{definition}
Conversely, we generalize and formalize an idea from \cite{venturi2020spurious} into the following lemma, which gives a sufficient condition for the existence of a spurious local valley:

\begin{lemma}[Sufficient condition for the existence of a spurious local valley]
	\label{lemma: spurious_local_valley}
	Consider a continuous function $L: \mb{R}^d \to \mb{R}$ whose global minimum is attained. Assume we know three subsets $S_1, S_2, S_3 \subset \mb{R}^d$ such that:
	\begin{enumerate}
		\item[1)] The global minima of $L$ are in $S_1$.
		
		\item[2)] Every continuous path from $S_3$ to $S_1$ passes through $S_2$.
		
		\item[3)] $\displaystyle\inf_{x \in S_2} L(x) > \inf_{x \in S_3} L(x) > \inf_{x \in S_1} L(x)$.
	\end{enumerate}
	Then $L$ has a spurious local valley. {Moreover, any $x \in S_3$ such that 
	$\displaystyle L(x) < \inf_{x \in S_2} L(x)$  is a point inside a spurious local valley.}
\end{lemma}
\begin{proof}
	Denote $\Sigma = \{x \mid L(x) = \inf_{x \in \mb{R}^d} L(\theta)\}$ the set of global minimizers of $L$. $\Sigma$ is not empty due to the assumption that the global minimum is attained, and $\Sigma \subseteq S_{1}$ by the first assumption.
	
	Since $\inf_{x \in S_2} L(x) > \inf_{x \in S_3} L(x) $, there exists $\tau \in S_3, L(\tau) < \inf_{x \in S_2} L(x)$. 
	Consider $\Phi$ the path-connected component of the sublevel set $\{x \mid L(x) \leq L(\tau)\}$ that contains $\tau$. Since $\Phi$ is a non-empty path-connected component of a level set, it is a local valley. It is thus sufficient to prove that $\Phi \cap \Sigma = \emptyset$ to obtain that it matches the very definition of a spurious local valley.
	
	Indeed, by contradiction, let's assume that there exists $\tau' \in \Phi \cap \Sigma$. Since $\tau,\tau' \in \Phi$ and $\Phi$ is path-connected, by definition of path-connectedness there exists a continuous function $f: [0,1] \to \Phi$ such that $f(0) = \tau \in S_3, f(1) = \tau' \in \Sigma \subseteq S_1$. Due to the assumption that every continuous path from $S_3$ to $S_1$ has to pass through a point in $S_2$, there must exist $t \in (0,1)$ such that $f(t) \in S_2 \cap \Phi$. Therefore, $L(f(t)) \leq L(\tau)$ (since $f(t) \in \Phi$) and $L(f(t)) > L(\tau)$ (since $f(t) \in S_2$), which is a contradiction.
\end{proof}

To finish this section, we formally recall previous results which are related to \eqref{eq: fixed_supp_prob} and will be used in our subsequent proofs. The questions of the existence of spurious local valleys and spurious local minima were addressed in previous works for full support matrix factorization and deep linear neural networks \cite{venturi2020spurious,nguyen2019connected,DBLP:journals/corr/abs-1805-04938,NIPS2016_6112}. We present only results related to our problem of interest.

\begin{theorem}[No spurious local valleys in linear networks {\cite[Theorem 11]{venturi2020spurious}}]
	\label{theorem: Nospuriousvalley}
	Consider linear neural networks of any depth $K \geq 1$ and of any layer widths $p_k \geq 1$ and any input - output dimension $n, m \geq 1$ with the following form: $\Phi(b, \theta) = W_K \ldots W_1 b$ where $\theta = (W_i)_{i=1}^K$, and $b \in \mb{R}^n$ is a training input sample. With the squared loss function, there is no spurious local valley. More specifically, the function 
$L(\theta) = \|A-\Phi(B, \theta)\|^2$
	satisfies the condition of \cref{lemma: non_spurious_local_valley} for any matrices $A \in \mb{R}^{m \times N}$ and $B \in \mb{R}^{n \times N}$ ($A$ and $B$ are the whole sets of training output and input respectively). 
\end{theorem} 
\begin{definition}[Strict saddle property {\cite[Definition $3$]{DBLP:journals/corr/abs-1805-04938}}]
\label{def:strictsaddle}
Consider a twice differentiable function $f: \mathbb{R}^{d} \to \mathbb{R}$. If each critical point of $f$ is either a global minimum or a \emph{strict} saddle point then $f$ is said to have the strict saddle property. When this property holds, $f$ has no spurious local minimum.
\end{definition}
Even if $f$ has the strict saddle property, it may have no global minimum, consider e.g. the function $f(x) = -\|x\|_{2}^{2}$.

\begin{theorem}[No spurious local minima in shallow linear networks {\cite[Theorem 3]{DBLP:journals/corr/abs-1805-04938}}]
	\label{theorem: nospuriousminima}
	Let $B \in \mb{R}^{d_0 \times N}, A \in \mb{R}^{d_2 \times N}$ be input and output training examples. Consider the problem:
	\begin{equation*}
		\begin{aligned}
			\underset{X \in \mb{R}^{d_0 \times d_1}, Y \in \mb{R}^{d_1 \times d_2}}{\emph{\text{Minimize}}} \quad 
			L(X, Y) = \|A-XYB\|^2
		\end{aligned}
	\end{equation*}
	If $B$ is full row rank, $f$ has the strict saddle property (see \Cref{def:strictsaddle}) hence $f$ has no spurious local minimum.
\end{theorem}

Both theorems are valid for a particular case of matrix factorization with fixed support: full support matrix factorization. Indeed, given a factorized matrix $A \in \mb{R}^{m \times n}$, in \cref{theorem: Nospuriousvalley}, if $K = 2, B = \mbf{I}_n$ $(n = N)$, then the considered function is $L = \|A - W_2W_1\|^2$. This is \eqref{eq: fixed_supp_prob} without support constraints $I$ and $J$ (and without a transpose on $W_1$, which does not change the nature of the problem). \cref{theorem: Nospuriousvalley} guarantees that $L$ satisfies the conditions of \cref{lemma: non_spurious_local_valley}, thus has no spurious local valley.

Similarly, in \cref{theorem: nospuriousminima}, if $B = \mbf{I}_{d_0}$ ($d_0 = N$, therefore $B$ is full row rank), we return to the same situation of \cref{theorem: Nospuriousvalley}. In general, \cref{theorem: nospuriousminima} claims that the landscape of the full support matrix factorization problem has the strict saddle property and thus, does not have spurious local minima. 

However, once we turn to \eqref{eq: fixed_supp_prob} with \emph{arbitrary} $I$ and $J$, such benign landscape is not guaranteed anymore, as we will show in \cref{ex:spuriousinstances}. Our work in the next subsections studies conditions on the support constraints $I$ and $J$ ensuring the absence / allowing the presence of spurious objects, and can be considered as a generalization of previous results with full supports. \cite{DBLP:journals/corr/abs-1805-04938,venturi2020spurious,NIPS2016_6112}.

\subsection{Landscape of matrix factorization with fixed support constraints}
\label{sec: landscape}
We start with the first result on the landscape in the simple setting of \cref{theorem:disjoint_totally_overlapping}.

\begin{theorem}
	\label{theorem:noSpuriousSimple}
	Under the assumption of \cref{theorem:disjoint_totally_overlapping}, the function $L(X,Y)$ in \eqref{eq: fixed_supp_prob} does not admit any spurious local valley for any matrix $A$. In addition, $L$ has the strict saddle property.
\end{theorem}

\begin{proof}
Recall that under the assumption of \cref{theorem:disjoint_totally_overlapping}, all the variables to be optimized are decoupled into ``blocks'' $\{ (X_{R_P,P}, Y_{C_P,P})\}_{P \in \mc{P}}$ ($P, \mc{P}$ are defined in \cref{def:completeeqclass}). We denote $\mc{P} = \{P_1, P_2, \ldots, P_{\ell}\}$, $P_i \subseteq \intset{r}$, $1\leq i \leq \ell$. From \Cref{eq:decomposedisjointoverlapping}, we have:
\begin{equation}
	\label{eq:decomposition}
	\|A - XY^\top\|^2 = \left(\sum_{P \in \mc{P}}\|A_{R_P,C_P} - X_{R_P,P}Y_{C_P,P}^\top\|^2\right) + \|A \odot \bar{\mS_\mc{P}}\|^2\\
\end{equation}
Therefore, the function $L(X,Y)$ is a sum of functions $L_P(X_{R_P,P},Y_{C_P,P}) := \|A_{R_P,C_P} - X_{R_P,P}Y_{C_P,P}^\top\|^2$, which do \emph{not} share parameters and are instances of the full support matrix factorization problem restricted to the corresponding blocks in $A$. The global minimizers of $L$ are $\{(X^\star_{R_P,P}, Y^\star_{C_P,P})\}_{P \in \mc{P}}$, where for each $P \in \mc{P}$ the pair $(X^\star_{R_P,P}, Y^\star_{C_P,P})$ is any global minimizer of $\|A_{R_P,C_P} - X_{R_P,P}Y_{C_P,P}^\top\|^2$.

\begin{itemize}[leftmargin = 14pt]
	\item[1)] \textbf{Non-existence of any spurious local valley}: By \cref{theorem: Nospuriousvalley}, from any initial point $(X_{R_P,P}^0, Y_{C_P,P}^0)$, there exists a continuous function $f_P(t) = (\te{X}_P(t),\te{Y}_P(t)): [0,1] \mapsto \mb{R}^{|R_P| \times |P|} \times \mb{R}^{|C_P| \times |P|}$ satisfying the conditions in \cref{lemma: non_spurious_local_valley}, which are:
	\begin{itemize}
		\item[i)] $f_P(0) = (X_{R_P,P}^0, Y_{C_P,P}^0)$.
		
		\item[ii)] $f_P(1) = (X^\star_{R_P,P}, Y^\star_{C_P,P})$.
		
		\item[iii)] $L_P \circ f_P: [0,1] \to \mb{R}$ is non-increasing.
	\end{itemize}

	Consider a feasible path (\cref{def:feasiblepath}) $f(t) = (\te{X}(t), \te{Y}(t)):[0,1] \mapsto \mathbb{R}^{m\times r}\times  \mathbb{R}^{r\times n}$ defined in such a way that $\te{X}(t)_{R_P,P} = \te{X}_P(t)$ for each $P \in \mc{P}$ and similarly for $\te{Y}(t)$. Since $L \circ f = \sum_{P \in \mc{P}} L_P \circ f_P + \|A \odot \bar{\mS_\mc{P}}\|^2$, $f$ satisfies the assumptions of \cref{lemma: non_spurious_local_valley}, which shows the non-existence of any spurious local valley.
	
	\item[2)] \textbf{Non-existence of any spurious local minimum}: Due to the decomposition in \Cref{eq:decomposition}, the gradient and Hessian of $L(X,Y)$ have the following form:
	\begin{equation*}
		\frac{\partial L}{\partial X_{R_P, P}} = \frac{\partial L_{P}}{\partial X_{R_P, P}}, \qquad \frac{\partial L}{\partial Y_{C_P,P}} = \frac{\partial L_P}{\partial Y_{C_P,P}}, \;\forall P \in \mathcal{P}
	\end{equation*}
	\begin{equation*}
		H(L)_{\mid (X,Y)}
		 \begin{pmatrix}
			H(L_{P_1})_{\mid (X_{R_{P_1}, P_1},Y_{C_{P_1}, P_1}))}
			& \ldots & \mathbf{0}\\
			\vdots & \ddots & \vdots \\
			\mathbf{0} & \ldots & 
			H(L_{P_\ell})_{\mid (X_{R_{P_\ell}, P_\ell},Y_{C_{P_\ell}, P_\ell}))}
		\end{pmatrix}
	\end{equation*}
	Consider a critical point $(X,Y)$ of $L(X,Y)$ that is not a global minimizer. Since $(X,Y)$ is a critical point of $L(X,Y)$, $(X_{R_P, P}, Y_{C_PP})$ is a critical point of the function $L_P$ for all $P \in \mc{P}$. Since $(X,Y)$ is not a global minimizer of $L(X,Y)$, there exists $P \in \mc{P}$ such that $(X_{R_P,P}, Y_{C_P,P})$ is not a global minimizer of $L_P$. By \cref{theorem: nospuriousminima}, $H(L_P)_{\mid (X_{R_P,P},Y_{C_P,P})}$ is not positive semi-definite. Hence, $H(L)_{\mid (X,Y)}$ is not positive semi-definite either (since $H(L)_{\mid (X,Y)}$ has block diagonal form). This implies that $(X,Y)$ it is a strict saddle point as well (hence, not a spurious local minimum).
\end{itemize}
\end{proof}
For spurious local valleys, we have the same results for the setting in \cref{theorem: reduction_disjoint_overlapping}. The proof is, however, less straightforward.

\begin{theorem}\label{th:MainNoSpuriousComplex}
	If $I$, $J$ satisfy the assumptions of \cref{theorem: reduction_disjoint_overlapping}, then for each matrix $A$ the landscape of $L(X,Y)$ in \eqref{eq: fixed_supp_prob} has no spurious local valley.
\end{theorem}

The following is a concept which will be convenient for the proof of \cref{th:MainNoSpuriousComplex}. 
\begin{definition}[CEC-full-rank]
	\label{def:fullrankCEC}
	A feasible point $(X,Y)$ is said {to be} \emph{CEC-full-rank} if $\forall P \in \mc{P}^\star$, either $X_{R_P,P}$ or $Y_{C_P,P}$ is full row rank. 
\end{definition}

We need three following lemmas to prove \cref{th:MainNoSpuriousComplex}:

\begin{lemma}
	\label{lem:fullrankhypothesis}
	Given $I \in \{0,1\}^{m \times r}$, $J \in \{0,1\}^{n \times r}$, consider $T$ and $\mS_T$ as in \cref{def:completeeqclass} and a feasible point $(X,Y)$. There exists a feasible path $f: [0,1] \to \mb{R}^{m \times r} \times \mb{R}^{n \times r}: f(t) = (X_f(t), Y_f(t))$ such that: 
	\begin{itemize}
		\item[1)] $f$ connects $(X,Y)$ with a CEC-full-rank point: $f(0) = (X,Y)$, and $f(1)$ is CEC-full-rank.
		\item[2)] $X_f(t)(Y_f(t))^{\top} = XY^\top, \forall t \in [0,1]$.
	\end{itemize}
\end{lemma}

\begin{lemma}
	\label{lem:connecttozeroCEC}
	Under the assumption of \cref{theorem: reduction_disjoint_overlapping}, for any CEC-full-rank feasible point $(X,Y)$, there exists feasible path $f: [0,1] \to \mb{R}^{m \times r} \times \mb{R}^{n \times r}: f(t) = (X_f(t), Y_f(t))$ such that:
	\begin{itemize}
		\item[1)] $f(0) = (X,Y)$.
		
		\item[2)] $L \circ f$ is non-increasing.
		
		\item[3)] $(A -X_f(1)(Y_f(1))^\top) \odot {\mS_T} = \mathbf{0}$.
	\end{itemize} 
\end{lemma}

\begin{lemma}
	\label{lem:connecttooptimal}
	Under the assumption of \cref{theorem: reduction_disjoint_overlapping}, for any CEC-full-rank feasible point $(X,Y)$ {satisfying}: $(A - XY^\top) \odot {\mS_T} = \mathbf{0}$, there exists a feasible path $f: [0,1] \to \mb{R}^{m \times r} \times \mb{R}^{n \times r}: f(t) = (X_f(t), Y_f(t))$ such that:
	\begin{itemize}
		\item[1)] $f(0) = (X,Y)$.
		
		\item[2)] $L \circ f$ is non-increasing.
		
		\item[3)] $f(1)$ is an optimal solution of $L$.
	\end{itemize} 
\end{lemma}
The proofs of \cref{lem:fullrankhypothesis}, \cref{lem:connecttozeroCEC} and \cref{lem:connecttooptimal} can be found in \cref{subapp:fullrankhypothesis}, \cref{subapp:connecttozeroCEC} and \cref{subapp:connecttooptimal} of the supplementary material. 
\begin{proof}[Proof of \cref{th:MainNoSpuriousComplex}]
	Given any initial point $(X^{0},Y^{0})$,  \cref{lem:fullrankhypothesis} shows the existence of a continuous path along which the product of $XY^\top = X^0(Y^0)^\top$ does not change (thus, $L(X,Y)$ is constant) and ending at a CEC-full-rank point. Therefore it is sufficient to prove the theorem under the additional assumption that $(X^0,Y^0)$ is CEC-full-rank. 	
	With this additional assumption, one can employ \cref{lem:connecttozeroCEC} to build a continuous path $f_1(t) = (X_1(t), Y_1(t))$, such that $t \mapsto L(X_1(t), Y_1(t))$ is non-increasing, that connects $(X^0, Y^0)$ to a point $(X^1, Y^1)$ satisfying: 
	\begin{equation*}
		(A - X^1(Y^1)^\top) \odot {\mS_T} = \mathbf{0}.
	\end{equation*}	
	Again, one can assume that $(X^1, Y^1)$ is CEC-full-rank (one can invoke \cref{lem:fullrankhypothesis} one more time). Therefore, $(X^1, Y^1)$ satisfies the conditions of \cref{lem:connecttooptimal} . Hence, there exists a continuous path $f_2(t) = (X_2(t), Y_2(t))$ that makes $L(X_2(t), Y_2(t))$ non-increasing and that connects $(X^1, Y^1)$ to $(X^*,Y^*)$, a global minimizer.
	
	Finally, since the concatenation of $f_1$ and $f_2$ satisfies the assumptions of \cref{lemma: non_spurious_local_valley}, we can conclude that there is no spurious local valley in the landscape of $\|A - XY^\top\|^2$.
\end{proof}
The next natural question is whether spurious local minima exist in the setting of \cref{theorem: reduction_disjoint_overlapping}. While in the setting of \cref{theorem:disjoint_totally_overlapping}, all critical points which are not global minima are saddle points, the setting of \cref{theorem: reduction_disjoint_overlapping} allows second order critical points (point whose gradient is zero and Hessian is positive semi-definite), which are not global minima.

\begin{example}
	\label{ex:ex1}
	Consider the following pair of support contraints $I, J$ and factorized matrix $I = \bigl[\begin{smallmatrix}	1 & 1 \\ 0 & 1 \end{smallmatrix}\bigr], \; J = \bigl[\begin{smallmatrix}	1 & 1 \\ 1 & 1 \end{smallmatrix}\bigr], \; A = \bigl[\begin{smallmatrix}	10 & 0 \\ 0 & 1 \end{smallmatrix}\bigr]$.	
	With the notations of \cref{def:completeeqclass2} we have $T = \{1\}$ and one can check that this choice of $I$ and $J$ satisfies the assumptions of \cref{theorem: reduction_disjoint_overlapping}.
	The infimum of $L(X,Y) = \|A - XY^\top\|^2$ is zero, and attained, for example at $X^* = I_2, Y^* = A$. Consider the following feasible point $(X_0,Y_0)$: $X_0 = \bigl[\begin{smallmatrix}	0 & 1 \\ 0 & 0 \end{smallmatrix}\bigr], \; Y_0 = \bigl[\begin{smallmatrix}	0 & 10 \\ 0 & 0 \end{smallmatrix}\bigr]$. Since $X_0Y_0^\top = \bigl[\begin{smallmatrix}	10 & 0 \\ 0 & 0 \end{smallmatrix}\bigr] \neq A$, $(X_0,Y_0)$ is not a global optimal solution. Calculating the gradient of $L$ verifes that $(X_0, Y_0)$ is a critical point:
	 \begin{equation*}
	 	\nabla L(X_0,Y_0) = ((A - X_0Y_0^\top)Y_0, (A^\top - Y_0X_0^\top)X_0) = (\mbf{0}, \mbf{0})
	 \end{equation*}
	 Nevertheless, the Hessian of the function $L$ at $(X_0,Y_0)$ is positive semi-definite. Direct calculation can be found in \cref{sec:proof_ex} of the supplementary material.
\end{example}
This example shows that if we want to prove the non-existence of spurious local minima in the new setting, one cannot rely on the Hessian. This is challenging since the second order derivatives computation is already tedious. Nevertheless, with \cref{def:fullrankCEC}, we can still say something about spurious local minima in the new setting. 
\begin{theorem}
	\label{theorem:nospuriousminimaComplex}
	Under the assumptions of \cref{theorem: reduction_disjoint_overlapping}, if a feasible point $(X,Y)$ is CEC-full-rank, then $(X,Y)$ is not a spurious local minimum of \eqref{eq: fixed_supp_prob}.
	Otherwise there is a feasible path, along which  $L(\cdot,\cdot)$ is constant, that joins $(X,Y)$ to some $(\te{X},\te{Y})$ which is not a spurious local minimum.
\end{theorem}

When $(X,Y)$ is \emph{not} CEC-full-rank, the theorem guarantees that it is not a strict local minimum, since there is path starting from $(X,Y)$ with constant loss. This should however not be a surprise in light of \Cref{rem:nostrict}: indeed, the considered loss function admits no strict local minimum at all. Yet, the path with ``flat'' loss constructed in the theorem is fundamentally different from the ones naturally due to scale invariances of the problem and captured by  \Cref{rem:nostrict}. Further work would be needed to investigate whether this can be used to get a stronger result.
\begin{proof}[Proof sketch]
	To prove this theorem, we proceed through two main steps:
	\begin{itemize}
		\item[1)] First, we show that any local minimum satisfies:
		\begin{equation}
			\label{eq: special_critical_point}
			(A - XY^\top) \odot {\mS_T} = \mathbf{0}
		\end{equation}
		
		\item[2)] Second, we show that if a point $(X,Y)$ is CEC-full-rank and satisfies \Cref{eq: special_critical_point}, it cannot be a spurious local minimum. 
			\end{itemize}
	Combining the above to steps, we obtain as claimed that if a feasible pair $(X,Y)$ is CEC-full-rank, then it is not a spurious local minimum. Finally, if a feasible pair $(X,Y)$ is not CEC-full-rank, \cref{lem:fullrankhypothesis} yields a feasible path along which $L$ is constant that joins $(X,Y)$ to some feasible $(\te{X},\te{Y})$ which is CEC-full-rank, hence (as we have just shown) not a spurious local mimimum. 
	
	A complete proof is presented in \cref{sec:proof_4} of the supplementary material.
\end{proof}

Although \cref{theorem:nospuriousminimaComplex} does not exclude completely the existence of spurious local minima, together with \cref{theorem:noSpuriousSimple}, we eliminate a large number of such points.

\subsection{Absence of correlation between tractability and benign landscape}

So far, we have witnessed that the instances of \eqref{eq: fixed_supp_prob} satisfying the assumptions of \cref{theorem: reduction_disjoint_overlapping} are not only efficiently solvable using \cref{algorithm3}: they also have a landscape with no spurious local valleys and  favorable in terms of spurious local minima \cref{theorem:nospuriousminimaComplex}.
The question of interest is: Is there a link between such benign landscape and the tractability of the problem? Even if the natural answer could intuitively seem to be positive, as it is the case for the full support case, we prove that this conjecture is not true. We provide a counter example showing that tractability does not imply a benign landscape. First, we establish a sufficient condition for the \emph{existence} of a spurious local valley in \eqref{eq: fixed_supp_prob}.
\begin{theorem}
	\label{theorem: sufficient_existence_slv}
	Consider function $L(X,Y) = \|A - XY^\top\|^2$ in \eqref{eq: fixed_supp_prob}. Given two support constraints  
	 $I \in \{0,1\}^{m \times r}$, $J \in \{0,1\}^{n \times r}$,
	if there exist 
	$i_1 \neq i_2 \in \llbracket m \rrbracket, j_1 \neq j_2 \in \llbracket n \rrbracket$ and $k \in \llbracket r \rrbracket$ such that
	{$(i_2, j_2)$} belongs to at least $2$ rank-one supports, one of which is $\mc{S}_k$, and if {$(i_1, j_1), (i_2,j_1), (i_1,j_2)$} belong only to $\mc{S}_k$, then: 
	\begin{itemize}
		\item[1)] There exists $A$ such that: $L(X,Y)$ has a spurious local valley.	
		\item[2)] There exists $A$ such that: $L(X,Y)$ has a spurious local minimum.
	\end{itemize}
	{In both cases, $A$ can be chosen so that the global minimum of $L(X,Y)$ under the considered support constraints is achieved and is zero.}
\end{theorem}
\begin{remark}
Note that the conditions of \Cref{theorem: sufficient_existence_slv} exclude these of \Cref{theorem:disjoint_totally_overlapping} and \Cref{theorem: reduction_disjoint_overlapping} (which is reasonable since the assumptions of \Cref{theorem:disjoint_totally_overlapping} and \Cref{theorem: reduction_disjoint_overlapping} rule out the possibility of spurious local valleys for any matrix $A$.).
\end{remark}


\begin{proof}
	Let $l \neq k$ be another rank-one contribution support $\mc{S}_l$ that contains $(i_1, j_1)$. Without loss of generality, we can assume $i_1 = j_1 = 1, i_2 = j_2 = 2$ and $k = 1, l = 2$. In particular, let $I' = J':= \{(1,1), (2,1), (2,2)\}$, then $I' \subseteq I, J' \subseteq J$. 	{When $m=n=2$, these are the support constraints for the $\mathbf{LU}$ decomposition.}
	\begin{itemize}
		\item[1)] We define the matrix $A$ by block matrices as:
		{\begin{equation}
			\label{eq:choiceA'slv}
			A = \begin{pmatrix}
				A' & \mathbf{0}\\
				\mathbf{0} & \mathbf{0}
			\end{pmatrix}, \text{ where }
			A' = \begin{pmatrix}
				1 & 1\\
				1 & 0
			\end{pmatrix} \in \mb{R}^{2 \times 2}.
		\end{equation}}
		The minimum of $L(X,Y) := \|A-XY^{\top}\|^{2}$ over feasible pairs is zero and it is attained at {$X = \bigl[\begin{smallmatrix} X' & \mbf{0} \\ \mbf{0} & \mbf{0} \end{smallmatrix}\bigr], Y = \bigl[\begin{smallmatrix} Y' & \mbf{0} \\ \mbf{0} & \mbf{0} \end{smallmatrix}\bigr]$ where $X' = \bigl[\begin{smallmatrix} 1 & 0 \\ 1 & 1 \end{smallmatrix}\bigr], Y' = \bigl[\begin{smallmatrix} 1 & 0 \\ 1 & -1 \end{smallmatrix}\bigr]$}. $(X,Y)$ is feasible since $\supp(X) = \supp(X') = I' \subseteq I, \supp(Y) = \supp(Y') = J' \subseteq J$. Moreover,
		\begin{equation}
			\label{eq:infattained}
			XY^\top = \begin{pmatrix}
				X'Y'^\top & \mbf{0}\\
				\mbf{0} & \mbf{0}
			\end{pmatrix} = \begin{pmatrix}
			A' & \mbf{0}\\
			\mbf{0} & \mbf{0}
		\end{pmatrix} = A.
		\end{equation}
		Using \cref{lemma: spurious_local_valley} we now prove that this matrix $A$ produces {a} spurious local 
		{valley} for $L(X,Y)$ with 
		{the considered support constraints $(I,J)$.}
		In fact, since {$(1, 1), (1, 2), (2,1)$} are only in $\mS_1$ and in no other support $\mS_{\ell}$, $\ell \neq 1$, one can easily check that for every feasible pair $(X,Y)$ we have: 
		{\begin{equation}\label{eq:SpecialAssumption}
		(XY^{\top})_{i,j} = X_{i,1}Y_{j,1},\quad \forall (i,j) \in \{(1,1),(1,2),(2,1)\}.
		\end{equation}}
		Thus, every feasible pair $(X^{\star},Y^{\star})$ reaching the global optimum $ \|A - X^{\star}(Y^{\star})^\top\| = 0$ must satisfy {$ X^{\star}_{1,1}Y^{\star}_{1,1} = X^{\star}_{2,1}Y^{\star}_{1,1} = X^{\star}_{1,1}Y^{\star}_{2,1} = 1$}. This implies {$X^{\star}_{2,1}Y^{\star}_{2,1} = (X^{\star}_{2,1}Y^{\star}_{1,1})(X^{\star}_{1,1}Y^{\star}_{2,1}) / (X^{\star}_{1,1}Y^{\star}_{1,1}) = 1$}. Moreover, such an optimum feasible pair {also satisfies $0 = A_{2,2}=(X^{\star}(Y^{\star})^{\top})_{2,2} = \sum_{p} X^{\star}_{2,p}Y^{\star}_{2,p}$}, hence
{$\sum_{p \neq 1} X^{\star}_{2,p}Y^{\star}_{2,p} = - X^{\star}_{2,1}Y^{\star}_{2,1} = -1$}. 
		
		To show the existence of a spurious local valley we use \cref{lemma: spurious_local_valley} and consider the set {$\te{S}_\sigma = \{(X,Y) \mid \supp(X) \subseteq I, \supp(Y) \subseteq J, \sum_{p \neq 1} X_{2,p}Y_{2,p} = \sigma\}$}. We will show that $S_1 := \te{S}_{-1}, S_2 := \te{S}_1, S_3 := \te{S}_5$ satisfy the assumptions of \cref{lemma: spurious_local_valley}.

		To compute $\inf_{(X,Y) \in S_i} L(X,Y)$, we study $g(\sigma) := \inf_{(X,Y) \in \te{S}_\sigma} L(X,Y)$. Denoting $Z = \bigl[\begin{smallmatrix} \mbf{1}_{2 \times 2} & \mbf{0} \\ \mbf{0} & \mbf{0} \end{smallmatrix}\bigr] \in \{0,1\}^{m \times n}$ we have:
		{\begin{equation*}
			\begin{aligned}
				g(\sigma) &= \inf_{(X,Y) \in \te{S}_\sigma} \|A - XY^\top\|^2\\
				&\geq \inf_{(X,Y) \in \te{S}_\sigma} \|(A - XY^\top) \odot Z\|^2\\
				&\stackrel{\eqref{eq:SpecialAssumption}}{=} \inf_{(X,Y) \in \te{S}_\sigma} \left\|\begin{pmatrix} A_{1,1} - X_{1,1}Y_{1,1} & A_{1,2} - X_{11}Y_{21}\\ A_{2,1} -X_{2,1}Y_{1,1} & A_{2,2} - \sigma - X_{2,1}Y_{2,1}\end{pmatrix}\right\|^2\\
				&= \inf_{X_{1,1},X_{2,1},Y_{1,1}, Y_{2,1}} \left\|\begin{pmatrix} 1 - X_{1,1}Y_{1,1} & 1 - X_{11}Y_{21}\\ 1 -X_{2,1}Y_{1,1} & -\sigma -X_{2,1}Y_{2,1}\end{pmatrix}\right\|^2\\
			\end{aligned}
		\end{equation*}}
		Besides~\Cref{eq:SpecialAssumption}, the third equality exploits the fact that {$(XY^{\top})_{2,2} = \sum_{p} X_{2,p}Y_{2,p} = X_{2,1}Y_{2,1}+\sigma$}. The last quantity is the loss of the best rank-one approximation of {$\te{A} = \bigl[\begin{smallmatrix} 1 & 1 \\ 1 & -\sigma \end{smallmatrix}\bigr] \in \mb{R}^{2 \times 2}$}. Since this is a $2\times 2$ symmetric matrix, its eigenvalues can be computed as the solutions of a second degree polynomial, leading to an analytic expression of this last quantity as: $\frac{2(\sigma+1)^2}{(\sigma^2+3) + \sqrt{(\sigma^2 + 3)^2 - 4(\sigma+1)^2}}$.
	
		Moreover, this infimum can be attained if $\bigl[X_{1,1},X_{2,1}\bigr]=\bigl[Y_{1,1}, Y_{2,1}\bigr]$ is the first eigenvector of $\te{A}$ and the other coefficients of $X,Y$ are set to zero. Therefore,
		\begin{equation}
			\label{eq: g(sigma)}
			g(\sigma) =  \frac{2(\sigma+1)^2}{(\sigma^2+3) + \sqrt{(\sigma^2 + 3)^2 - 4(\sigma+1)^2}}.
		\end{equation}	
		We can now verify that $S_{1},S_{2},S_{3}$ satisfy all the conditions of \cref{lemma: spurious_local_valley}.
		\begin{itemize}
			\item[1)] The minimum value of $L$ is zero. As shown above, it is only attained with {$\sum_{p \neq 1} X^{\star}_{2,p}Y^{\star}_{2,p} = -1$} as shown. Thus, the global minima belong to $S_{1} = \te{S}_{-1}$.
			
			\item[2)] For any feasible path {$r: [0,1] \to \mb{R}^{m \times r} \times \mb{R}^{n \times r}: t \to (X(t), Y(t))$ we have $\sigma_r(t) = \sum_{p \neq 1} X(t)_{2,p}Y(t)_{2,p}$} is also continuous. If $(X(0),Y(0)) \in S_3 = \te{S}_{5}$ and $(X(1), Y(1)) \in S_1 = \te{S}_{-1}$ then $\sigma_r(0) = 5$ and $\sigma_r(1) = -1$), hence by the Mean Value Theorem, there must exist $t \in (0,1)$ such that $\sigma_r(t) = 1$, which means $(X(t), Y(t)) \in S_2 = \te{S}_{1}$.
			
			\item[3)] Since one can check numerically that $g(1) > g(5) > g(-1)$, we have 
			\[
			\displaystyle\inf_{(X,Y) \in S_2} L{(X,Y)} > \inf_{(X,Y) \in S_3} L {(X,Y)} > \inf_{(X,Y) \in S_1} L{(X,Y)} .
			\]
		\end{itemize}
		The proof is concluded with the application of \cref{lemma: spurious_local_valley}. {In addition, any point $(X,Y)$ satisfying $\sigma = 5$ and $L(X,Y) < g(1) = 2$ is inside a spurious local valley. For example, one of such a point is $X = \bigl[\begin{smallmatrix} X' & \mbf{0} \\ \mbf{0} & \mbf{0} \end{smallmatrix}\bigr], Y = \bigl[\begin{smallmatrix} Y' & \mbf{0} \\ \mbf{0} & \mbf{0} \end{smallmatrix}\bigr]$ where $X' = \bigl[\begin{smallmatrix} 1 & 0 \\ -5 & 1 \end{smallmatrix}\bigr], Y' = \bigl[\begin{smallmatrix} -1/5 & 0 \\1 & 5 \end{smallmatrix}\bigr]$.}
		
		\item[2)] We define the matrix $A$ by block matrices as:
		{\begin{equation}
			\label{eq:choiceA'slm}
			A = \begin{pmatrix}
				A' & \mathbf{0}\\
				\mathbf{0} & \mathbf{0}
			\end{pmatrix}, \text{ where }
			A' = \begin{pmatrix}
				b & 0\\
				0 & a
			\end{pmatrix} \in \mb{R}^{2 \times 2}.
		\end{equation}}
		where $a > b > 0$. It is again evident that 
		{The infimum of $\|A - XY^\top\|^2$ under the considered support constraints is zero, and is achieved}
		 (taking $X = \bigl[\begin{smallmatrix} X' & \mbf{0} \\ \mbf{0} & \mbf{0} \end{smallmatrix}\bigr], Y = \bigl[\begin{smallmatrix} Y' & \mbf{0} \\ \mbf{0} & \mbf{0} \end{smallmatrix}\bigr]$ where {$X' = \bigl[\begin{smallmatrix} b & 0 \\ 0 & a \end{smallmatrix}\bigr], Y' = \bigl[\begin{smallmatrix} 1 & 0 \\0 & 1 \end{smallmatrix}\bigr]$} and with the same proof as in \Cref{eq:infattained}, we have $XY^\top = A$.)
		
		Now, we will consider {$\te{X} = \bigl[\begin{smallmatrix} X' & \mbf{0} \\ \mbf{0} & \mbf{0} \end{smallmatrix}\bigr],\te{Y} = \bigl[\begin{smallmatrix} Y' & \mbf{0} \\ \mbf{0} & \mbf{0} \end{smallmatrix}\bigr]$ where $X' = \bigl[\begin{smallmatrix} 0 & 0 \\ a & 0 \end{smallmatrix}\bigr], Y' = \bigl[\begin{smallmatrix} 0 & 0 \\1 & 0 \end{smallmatrix}\bigr]$}. Since $L(\te{X},\te{Y}) = b^2 > 0$ it cannot be a global minimum.  We will show that $(\te{X},\te{Y})$ is indeed a local minimum, which will thus imply that $(\te{X},\te{Y})$ is a spurious local minimum. For each feasible pair $(X,Y)$ we have:
		{\begin{equation*}
			\begin{split}
				&\|A - XY^\top\|^2
				= \sum_{i,j} (A_{i,j} - (XY^\top)_{i,j})^2\\
				&\geq(A_{1,1} -(XY^\top)_{1,1})^2 + (A_{2,1} -(XY^\top)_{2,1})^2 + (A_{1,2} -(XY^\top)_{1,2})^2\\
				&\stackrel{\eqref{eq:SpecialAssumption}}{=} (b - X_{1,1}Y_{1,1})^2 + (X_{2,1}Y_{1,1})^2 + (X_{1,1}Y_{2,1})^2\\
				&\geq (X_{1,1}Y_{1,1})^2 - 2bX_{1,1}Y_{1,1} + b^2  + 2(X_{2,1}Y_{2,1})|X_{1,1}Y_{1,1}|\\
				&\geq 2(X_{2,1}Y_{2,1} - b)|X_{1,1}Y_{1,1}| + b^2.
			\end{split}
		\end{equation*}}
	where in the third line we used that for $u = |X_{2,1}|Y_{11}$, $v = X_{11} |Y_{2,1}|$, since $(u-v)^{2} \geq 0$ we have $u^{2}+v^{2} \geq 2uv$. Since {$\te{X}_{2,1}\te{Y}_{2,1} = a > b$}, there exists a neighborhood of $(\te{X},\te{Y})$ such that {${X}_{2,1}{Y}_{2,1} - b > 0$ }for all $(X,Y)$ in that neighbourhood. Since {$|X_{1,1}Y_{1,1}| \geq 0$} in this neighborhood it follows that $\|A - XY^\top\|^2 \geq b^2 = L(\te{X},\te{Y}) > 0$ in that neighborhood. This concludes the proof.
	\end{itemize}		
\end{proof}

\begin{remark}
	\label{rem:divergencePALM}
	{\Cref{theorem: sufficient_existence_slv} is constructed based on the $\mbf{LU}$ structure. We elaborate our intuition on the technical proof of \Cref{theorem: sufficient_existence_slv} as follows: Consider the $\mbf{LU}$ decomposition problem of size $2 \times 2$ (i.e., $I = J = \{(1,1), (2,1), (2,2)\}$). It is obvious that such $(I,J)$ satisfies the assumptions of \Cref{theorem: sufficient_existence_slv} (for $i_1 = j_1 = 1, i_2 = j_2 = 2$). We consider three matrices of size $2 \times 2$:
	\begin{equation*}
		A_1 = \begin{pmatrix}
			1 & 1\\
			1 & 0
		\end{pmatrix}, \;A_2 = \begin{pmatrix}
			1 & 0\\
			0 & 2
		\end{pmatrix}, \;A_3 = \begin{pmatrix}
			0 & 1\\
			1 & 0
		\end{pmatrix}.
	\end{equation*}
	$A_1$ (resp. $A_2$) is simply the matrix $A'$ in \eqref{eq:choiceA'slv} (resp. in \eqref{eq:choiceA'slm}, with $a = 2, b = 1$) in the proof of \Cref{theorem: sufficient_existence_slv}. $A_3$ is a matrix which does not admit an $\mbf{LU}$ decomposition. We plot the graphs of $\displaystyle g_i(\sigma) = \inf_{X_{2,2}Y_{2,2} = \sigma} \|A_i - XY^\top\|$ (this is exactly $g(\sigma)$ introduced in the proof of \Cref{theorem: sufficient_existence_slv}) in \Cref{fig:gsigma}.
	\begin{figure}[bhtp]
		\centering
		\includegraphics[scale=0.35]{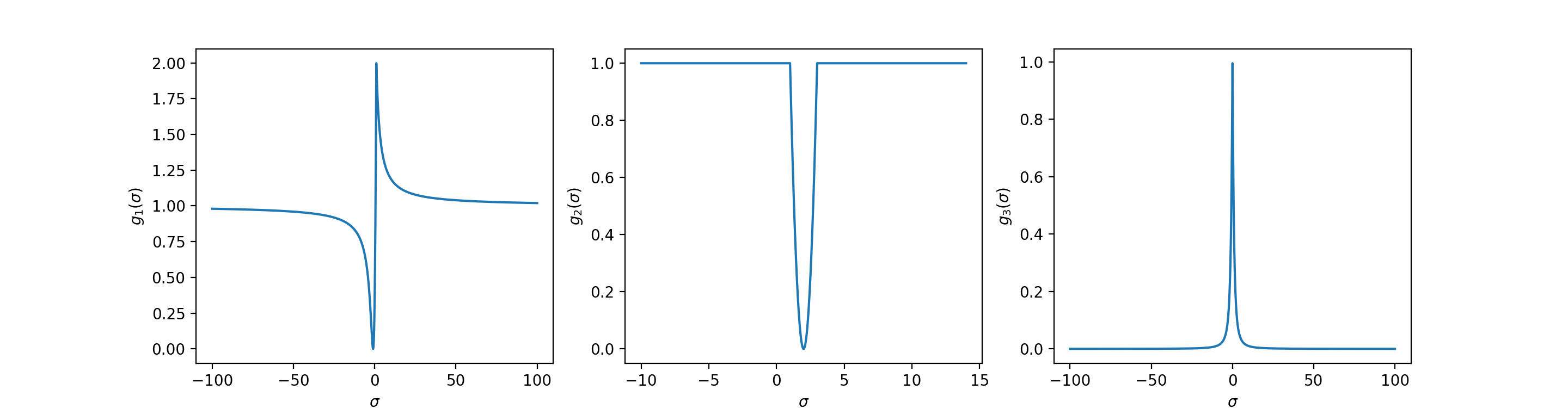}
		\caption{Illustration of the functions $g_i(\sigma), i = 1, 2, 3$ from left to right.}
		\label{fig:gsigma}
	\end{figure}}
	
	{
	In particular, the spurious local valley constructed in the proof of \Cref{theorem: sufficient_existence_slv} with $A_1$ is a spurious local valley extending to infinity. With $A_2$, one can see that $g_2(\sigma)$ has a plateau with value $1 = b^2$. The local minimum that we consider in the proof of \Cref{theorem: sufficient_existence_slv} is simply a point in this plateau (where $\sigma = 0$). Lastly, since the matrix $A_3$ does not admit an $\mbf{LU}$ decomposition, there is no optimal solution. Nevertheless, the infimum zero can be approximated with arbitrary precision when $\sigma$ tends to infinity (two valleys extending to $\pm \infty$).}
	
	{
	For the cases with the matrices $A_1$ and $A_3$, once initialized inside the valleys of their landscapes, any sequence $(X_k,Y_k)$ with sufficiently small steps associated to a decreasing loss $L(X_k, Y_k)$ will have the corresponding parameter $\sigma$ converging to infinity. As a consequence, at least one parameter of either $X_k$ or $Y_k$ has to diverge. This is thus a setting in which PALM (and other optimization algorithms which seek to locally decrease their objective function in a monotone way) can diverge. }
\end{remark}

We can now exhibit the announced counter-example to the mentioned conjecture:
{
\begin{remark}
	\label{ex:spuriousinstances}
	Consider the $\mbf{LU}$ decomposition as an instance of \eqref{eq: fixed_supp_prob} with $m = n =r$, $I = J = \{(i,j) \mid 1 \leq j \leq i \leq n\}$, taking $i_1 = j_1 = 1, i_2 = j_2 = 2$ shows that the $\mbf{LU}$ decomposition satisfies the condition of \Cref{theorem: sufficient_existence_slv}. Consequently, there exists a matrix $A$ such that the global optimum of $L(X,Y)$ is achieved (and is zero), yet the landscape of $L(X,Y)$ will have spurious objects. Nevertheless, a polynomial algorithm to compute the $\mbf{LU}$ decomposition exists \cite{LUexistence}. This example is in the same spirit of a recent result presented in  \cite{exponentialspurious}, where a polynomially solvable instance of Matrix Completion is constructed, whose landscape can have an exponential number of spurious local minima.
\end{remark}
}
The existence of spurious local valleys shown in \cref{theorem: sufficient_existence_slv} highlights the importance of initialization: if an initial point is already inside a spurious valley, first-order methods cannot escape this suboptimal area. An optimist may wonder if there nevertheless exist a smart initialization that avoids all spurious local valleys initially. The answer is positive, as shown in the following theorem.
\begin{theorem}
	\label{theorem:smartinit}
	Given any $I, J, A$ such that the infimum of \eqref{eq: fixed_supp_prob} is attained, 
	every initialization $(X, \mbf{0}), \supp(X) \subseteq I$ (or symmetrically $(\mathbf{0}, Y), \supp(Y) \subseteq J$) is not in any spurious local valley. In particular, $(\mbf{0},\mbf{0})$ is never in any spurious local valley.

\end{theorem}
\begin{proof}
	Let $(X^*,Y^*)$ be a minimizer of \eqref{eq: fixed_supp_prob}, which exists due to our assumptions. We only prove the result for the initialization $(X, \mathbf{0}), \supp(X) \subseteq I$. The case of the initialization $(\mathbf{0},Y)$, $\supp(Y) \subseteq J$ can be dealt with similarly. 
	
	To prove the theorem, it is sufficient to construct $f(t) = (X_f(t), Y_f(t) ): [0,1] \to \mb{R}^{m \times r} \times \mb{R}^{n \times r}$ as a feasible path such that:
	\begin{itemize}
		\item[1)] $f(0) = (X, \mathbf{0})$.
		\item[2)] $f(1) = (X^*, Y^*)$.
		\item[3)] $L \circ f$ is non-increasing w.r.t $t$.
	\end{itemize}
	Indeed, if such $f$ exists, the sublevel set corresponding to $L(X, \mathbf{0})$ has both $(X, \mathbf{0})$ and $(X^*, Y^*)$ in the same path-connected components (since $L \circ f$ is non-increasing).
	
	We will construct such a function feasible path $f$ as a concatenation of two functions feasible paths $f_1: [0, 1/2] \to \mb{R}^{m \times r} \times \mb{R}^{n \times r}, f_2: [1/2, 1] \to \mb{R}^{m \times r} \times \mb{R}^{n \times r}$, defined as follows:
	\begin{itemize}
		\item[1)] $f_1(t) = ((1 - 2t)X + 2tX^*, \mathbf{0})$.
		\item[2)] $f_2(t) = (X^*, (2t-1)Y^*)$.
	\end{itemize}
	It is obvious that $f(0) = f_1(0) = (X, \mathbf{0})$ and $f(1) = f_2(1) = (X^*, Y^*)$. Moreover $f$ is continuous since $f_1(1/2) = f_2(1/2) = (X^*, \mathbf{0})$. Also, $L \circ f$ is non-increasing on $[0,1]$ since:
	\begin{itemize}
		\item[1)] $L(f_1(t)) = \|A - ((1 - 2t)X + 2tX^*)\mbf{0}^\top\|^2 = \|A\|^2$ is constant for $t \in [0,1/2]$.
		\item[2)] $L(f_2(t)) = \|A - (2t-1)X^*Y^*\|^2$ is convex w.r.t $t$. Moreover, it attains a global minimum at $t = 1$ (since we assume that $(X^*, Y^*)$ is a global minimizer of \eqref{eq: fixed_supp_prob}). As a result, $t \mapsto L(f_2(t))$ is non-increasing on $[1/2,1]$.
	\end{itemize}
	
\end{proof}
Yet, such an initialization does not guarantee that first-order methods converge to a global minimum. Indeed, while in the proof of this result we do show that there exists a feasible path joining this ``smart'' initialization to an optimal solution without increasing the loss function, the value of the objective function is ``flat'' in the first part of this feasible path. Thus, even if such initialization is completely outside any spurious local valley, it is not clear whether local information at the initialization allows to ``guide'' optimization algorithms towards the global optimum to blindly find such a path. {In fact}, first-order methods are not bound to follow our constructive continuous path.

\section{Numerical illustration: landscape and behaviour of gradient descent}
\label{sec:experiment}
{
As a numerical illustration of the practical impact of our results, we compare the performance of \Cref{algorithm3} to other popular first-order methods on problem \eqref{eq: fixed_supp_prob}. 
}

{We consider two types of instances of $\eqref{eq: fixed_supp_prob}$: $I_1 = \mbf{1}_{2^a \times 2^a} \otimes \mbf{I}_{2^b\times 2^b}, J_1 = \mbf{I}_{2^a \times 2^a} \otimes \mbf{1}_{2^b \times 2^b}$ where $\otimes$ denotes the Kronecker product, $a = \lceil N/2 \rceil, b = \lfloor N/2 \rfloor$ (hence $a+b = N$) and $I_2 = \mbf{1}_{2 \times 2} \otimes \mbf{I}_{2^{N - 1}}, J_2 = \mbf{I}_2 \otimes \mbf{1}_{2^{N-1} \times 2^{N-1}}$. These supports are interesting because they are those taken at the first two steps of the hierarchical algorithm in \cite{papierICASSP,papierLeon} for approximating a matrix by a product of $N$ butterfly factors \cite{papierICASSP}. The first pair of support constraints $(I_1, J_1)$ is also equivalent to the recently proposed Monarch parameterization \cite{monarch}. Both pairs $(I_{1},J_{1})$ and $(I_{2},J_{2})$ are proved to satisfy \Cref{theorem:disjoint_totally_overlapping} \cite[Lemma 3.15]{papierLeon}.
}

{
\begin{figure}[bhtp]
	\label{fig:experment1}
	\includegraphics[scale=0.35]{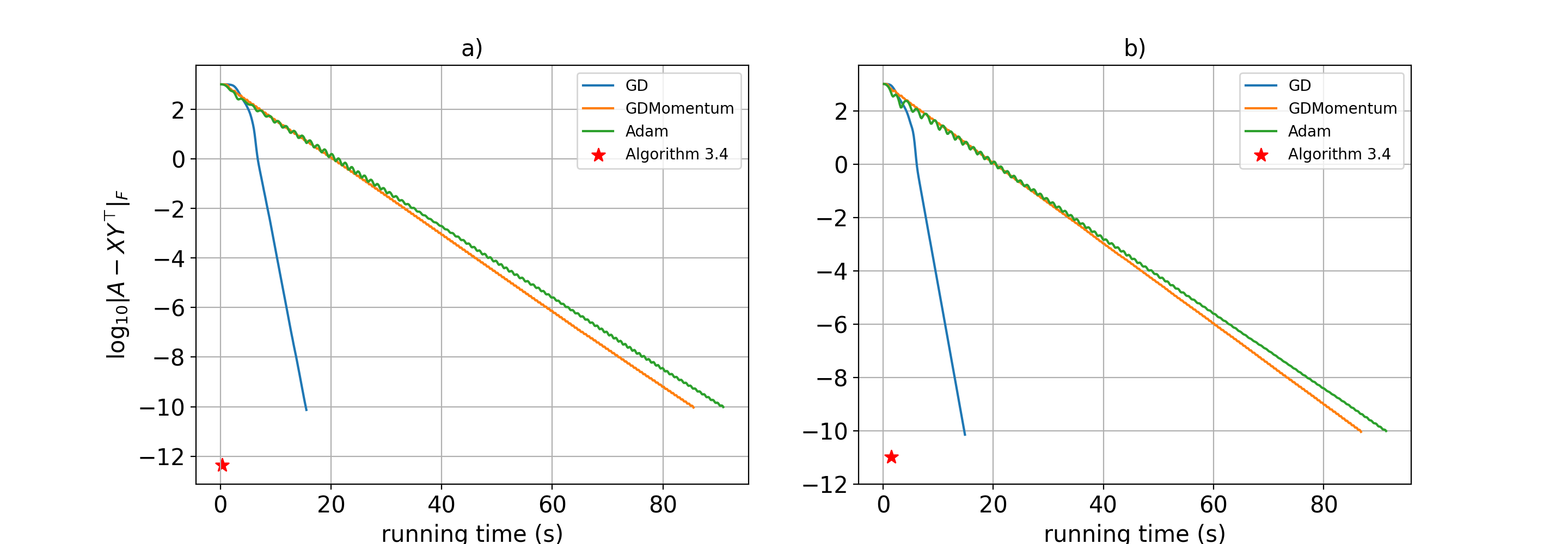}
	\caption{Evolution of $\log_{10} \|A - XY^\top\|_{F}$ for three variants of gradient descent and \Cref{algorithm3} with support constraints $(I_1,J_1)$ (left) and $(I_2, J_2)$ (right) for $N =10$.}
\end{figure}	
We consider $A$ as the Hadamard matrix of size $2^N \times 2^N$, which is known to admit an exact factorization with each of the considered support constraints, and we employ \Cref{algorithm3} to factorize $A$ in these two settings. We compare \Cref{algorithm3}  to three variants of gradient descent: vanilla gradient descent (GD), gradient descent with momentum (GDMomentum) and ADAM \cite[Chapter 8]{DeepLearning}. We use the efficient implementation of these iterative algorithms available in Pytorch 1.11. 
For each matrix size $2^{N}$, learning rates for iterative methods are tuned by grid search: we run all the factorizations with all learning rates in $\{5 \times 10^{-k}, 10^{-k} \mid k = 1, \ldots, 4\}$. Matrix $X$ (resp. $Y$) is initialized with i.i.d. random coefficients inside its support $I$ (resp. $J$) drawn according to the law $\mc{N}(0, 1/R_{I})$ (resp. $\mc{N}(0,1/R_{J})$) where $R_{I}, R_{J}$ are respectively the number of elements in each column of $I$ and of $J$. All these experiments are run on an Intel Core i7 CPU 2,3 GHz. In the interest of reproducible research, our implementation is available in open source \cite{codeSLV}. 
Since $A$ admits an exact factorization with both the supports $(I_1, J_1)$ and $(I_2, J_2)$, we set a threshold $\epsilon = 10^{-10}$ for these iterative algorithms (i.e if $\log_{10} (\|A - XY^\top\|_F) \leq -10$, the algorithm is terminated and considered to have found an optimal solution). This determines the running time for a given iterative algorithm for a given dimension $2^{N}$ and a given learning rate.  For each dimension $2^{N}$ we report the best running time over all learning rates. The reported running times do not include the time required for hyperparameters tuning.
\begin{figure}[bhtp]
	\label{fig:experment2}
	\includegraphics[scale=0.35]{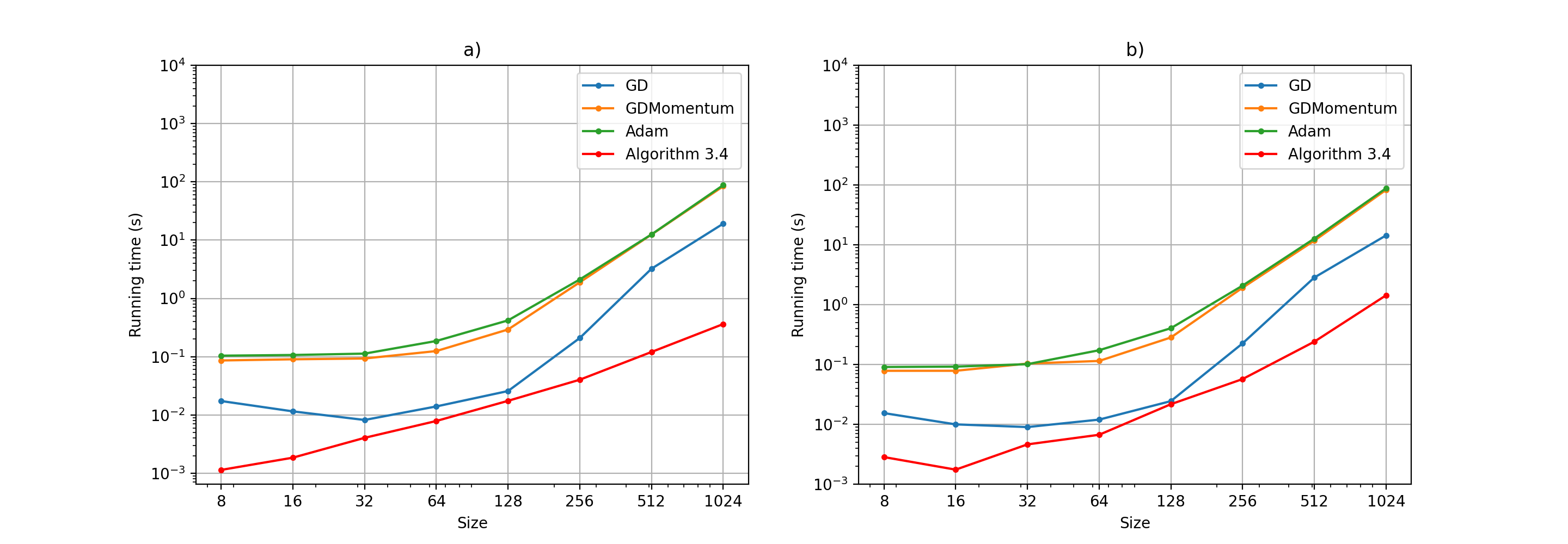}
	\caption{Running time (in logarithmic scale, contrary to Figure~\ref{fig:experment1}) of three variants of gradient descent and \Cref{algorithm3} to reach a precision $\log_{10}(\|A-XY^{\top}\|_{F}) \leq -10$; $N \in \{3, \ldots, 10\}$ with support constraints $(I_1, J_1)$ (left) and $(I_2, J_2)$ (right).}
\end{figure}
The experiments illustrated in Figure \ref{fig:experment1} for $N=10$ confirm our results on the landscape presented in \Cref{sec: landscape}: the assumptions of theorem \cref{theorem:disjoint_totally_overlapping} are satisfied so the landscape is benign and all variants of gradient descent are able to find a good factorization for $A$ from a random initialization}. 

{Figure \ref{fig:experment1} also shows that \Cref{algorithm3} is consistently better than the considered iterative methods 
in terms of running time, regardless of the size of $A$, cf. Figure \ref{fig:experment2}. A crucial advantage of \Cref{algorithm3} over gradient methods is also that it is free of hyperparameter tuning, which is critical for iterative methods to perform well, and may be quite time consuming (we recall that the time required for hyperparameters tuning of these iterative methods is \emph{not} considered in Figure  \ref{fig:experment2}). In addition, \Cref{algorithm3} can be further accelerated since its main steps (cf \Cref{algorithm0.5}) rely on block SVDs that can be computed in parallel (in these experiments, our implementation of \Cref{algorithm3} is not parallelized yet). Interested readers can find more applications of \Cref{algorithm3} on the problem of fixed-support multilayer sparse factorization in \cite{papierICASSP}.
}

\section{Conclusion}
In this paper, we studied the problem of two-layer matrix factorization with fixed support. We showed that 
this problem is NP-hard in general. Nevertheless, certain structured supports allow for an efficient solution algorithm. Furthermore, we also showed the non-existence of spurious objects in the landscape of function $L(X,Y)$ of \eqref{eq: fixed_supp_prob} with these support constraints. Although it would have seemed natural to assume an equivalence between tractability and benign landscape of \eqref{eq: fixed_supp_prob}, we also show a counter-example that contradicts this conjecture. That shows that there is still  room for improvement of the current tools (spurious objects) to characterize the tractability of an instance. We have also shown numerically {the advantages of the proposed algorithm over state-of-the-art first order optimization methods usually employed in this context.}
We refer the reader to \cite{papierICASSP} where we propose an extension of \Cref{algorithm1} to fixed-support multilayer sparse factorization and show the superiority of the resulting method in terms of both accuracy and speed compared to the state of the art \cite{DBLP:journals/corr/abs-1903-05895}.


\label{sec:conclusion}

\bibliographystyle{siamplain}
\bibliography{references}

\newpage
\appendix
\section{Proof of \cref{lem: NPhardness}}
\label{sm:NPhardness}
Up to a transposition, we can assume WLOG that $m \geq n$. We will show that with $r = n + 1 = \min(m,n) + 1$, we can find two supports $I$ and $J$ satisfying the conclusion of \cref{lem: NPhardness}.

To create an instance of \eqref{eq: fixed_supp_prob} (i.e., two supports $I,J$) that is \textit{equivalent} to \eqref{eq: WLRA1}, we define $I \in \{0,1\}^{m \times (n + 1)}$ and $J \in \{0,1\}^{n \times (n + 1)}$ as follows:
\begin{equation}
	\label{eq: support_hard}
	\begin{aligned}
		I_{i,j} &= \begin{cases}
			1 - W_{i,j} & \text{if } j \neq n\\
			1 & \text{if } j = n + 1
		\end{cases}, \;
		J_{i,j} &= \begin{cases}
			1 & \text{if } j = i \text{ or } j = n + 1\\
			0 & \text{otherwise}
		\end{cases}
	\end{aligned}
\end{equation}

\Cref{fig: support} illustrates an example of support constraints built from $W$.
\begin{figure}[h]
	\centering
	\includegraphics[scale=0.5]{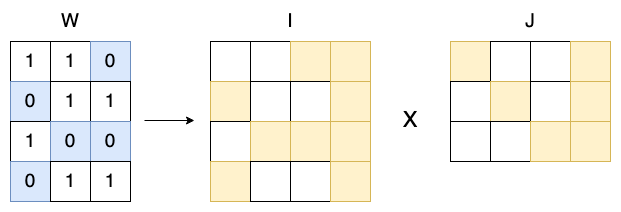}
	\caption{Factor supports $I$ and $J$ constructed from the weighted matrix $W \in \{0,1\}^{4 \times 3}$. Colored squares in $I$ and $J$ are positions in the supports.}
	\label{fig: support}
\end{figure}

We consider the \eqref{eq: fixed_supp_prob} with the same matrix $A$ and $I,J$ defined as in Equation \eqref{eq: support_hard}. This construction (of $I$ and $J$) can clearly be made in polynomial time. Consider the coefficients $(XY^\top)_{i,j}$:
\begin{itemize}
	\item[1)] If $W_{i,j} = 0$: $(XY^\top)_{i,j} = \sum_{k = 1}^{n+1} X_{i,k}Y_{j,k} = X_{i,j}Y_{j,j} + X_{i,n+1}Y_{j, n + 1}$ (except for $k = n + 1$, only $Y_{j,j}$ can be different from zero due to our choice of $J$).
	
	\item[2)] If $W_{i,j} = 1$: $(XY^\top)_{i,j} = \sum_{k = 1}^{n+1} X_{i,k}Y_{j,k} = X_{i,n + 1}Y_{j,n + 1}$ (same reason as in the previous case, in addition to the fact that $I_{i, j} = 1 - W_{i,j} = 0$). 	
\end{itemize}

Therefore, the following equation holds:
\begin{equation}
	\label{eq: xydotw}
	(XY^\top) \odot W = (X_\col{n+1}Y_\row{n + 1}^\top) \odot W
\end{equation}

We will prove that \eqref{eq: fixed_supp_prob} and \eqref{eq: WLRA1} share the same infimum\footnote{
	We focus on the infimum instead of minimum since there are cases where the infimum is not attained, as shown in \cref{rem: no_global_min}}. Let $\mu_1 = \inf_{x \in \mb{R}^m, y \in \mb{R}^n} \|A - xy^\top\|_W^2$ and $\mu_2 = \inf_{\supp(X) \subseteq I, \supp(Y) \subseteq J} \|A - XY^\top\|^2$. It is clear that $\mu_i \geq 0 > -\infty, i = 1,2$. Our objective is to prove $\mu_1 \leq \mu_2$ and $\mu_2 \leq \mu_1$.

\begin{itemize}
	\item[1)] Proof of $\mu_1 \leq \mu_2$: By definition of an infimum, for all $\mu > \mu_1$, there exist $x,y$ such that $\|A - xy^\top\|_W^2 \leq \mu$. We can choose $X$ and $Y$ (with $\supp(X)  \subseteq I, \supp(Y) \subseteq J$) as follows: we take the last columns of $X$ and $Y$ equal to $x$ and $y$ ($X_{\bullet, n+1} = x, Y_{\bullet, n + 1} = y$). For the \textit{remaining} columns of $X$ and $Y$, we choose:
	\begin{equation*}
		\begin{aligned}
			X_{i, j} &= A_{i, j} - x_iy_j &\text{ if } I_{i,j} = 1, j \leq n \\
			Y_{i, j} &= 1 & \text{ if } J_{i,j} = 1, j \leq n\\
		\end{aligned}
	\end{equation*}
	This choice of $X$ and $Y$ will make $\|A - XY^\top\|^2 = \|A - xy^\top\|_W^2 \leq \mu$. Indeed, for all $(i,j)$ such that $W_{i,j} = 0$, we have:
	\begin{equation*}
		\begin{split}
			(A - XY^\top)_{i,j} &= A_{i,j} - X_{i,j}Y_{j,j} - X_{i, n + 1}Y_{j,n + 1} = A_{i,j} - A_{i,j} + x_iy_j - x_iy_j = 0
		\end{split}
	\end{equation*}
	Therefore, it is clear that: $(A - XY^\top) \odot (\mathbf{1} - W) = \mathbf{0}$.
	\begin{equation*}
		\begin{split}
			\|A - XY^\top\|^2 &= \|(A - XY^\top) \odot W\|^2 + \|(A - XY^\top) \odot (\mbf{1} - W)\|^2\\
			&= \|(A - XY^\top) \odot W\|^2\\
			&\overset{\eqref{eq: xydotw}}{=} \|(A - X_\col{n+1}Y_\col{n + 1}^\top) \odot W\|^2 \\
			&= \|(A - xy^\top) \odot W\|^2\\ 
			&= \|A - xy^\top\|_W^2\\ 
		\end{split}
	\end{equation*}
	Therefore, $\mu_2 \leq \mu_1$.
	
	\item[2)] Proof of $\mu_1 \leq \mu_2:$ Inversely, for all $\mu > \mu_2$, there exists $X, Y$ satisfying $\supp(X)  \subseteq I, \supp(Y) \subseteq J$ such that $ \|A - XY^\top\|^2 \leq \mu$. We choose $x = X_\col{n+1}, y = Y_\col{n+1}$. It is immediate that: 
	\begin{equation*}
		\begin{split}
			\|A - xy^\top\|_W^2 &= \|(A - xy^\top) \odot W\|^2\\
			&= \|(A - X_\col{n+1}Y_\col{n+1}^\top) \odot W\|^2\\
			&\overset{\eqref{eq: xydotw}}{=} \|(A - XY^\top) \odot W\|^2 \\
			&\leq \|(A - XY^\top) \odot W\|^2 + \|(A - XY^\top) \odot (\mbf{1} - W)\|^2\\
			&= \|A - XY^\top\|^2
		\end{split}
	\end{equation*}
	Thus, $\|A - xy^\top\|_W^2 \leq \|A - XY^\top\|^2 \leq \mu$. We have $\mu_1 \leq \mu_2$.
\end{itemize}
This shows that $\mu_1 = \mu_2$. Moreover, the proofs of $\mu_1 \leq \mu_2$ and $\mu_2 \leq \mu_1$ also show the procedures to obtain an optimal solution of one problem with a given accuracy $\epsilon$ provided that we know an optimal solution of the other with the same accuracy.

\begin{remark}
	\label{rem: no_global_min}
	In the proof of \cref{lem: NPhardness}, we focus on the infimum instead of minimum since there are cases where the infimum is not attained. Indeed, consider the following instance of \eqref{eq: fixed_supp_prob} with: $A = \bigl[\begin{smallmatrix}	0 & 1 \\ 1 & 0 \end{smallmatrix}\bigr], \;I = \bigl[\begin{smallmatrix}	1 & 1 \\ 0 & 1 \end{smallmatrix}\bigr], \; J = \bigl[\begin{smallmatrix}	1 & 1 \\ 0 & 1 \end{smallmatrix}\bigr]$. The infimum of this problem is zero, which can be shown by choosing: $X_k = \bigl[\begin{smallmatrix} -k & k\\ 0 & \frac{1}{k} \end{smallmatrix}\bigr], \; Y_k = \bigl[\begin{smallmatrix}	k & k \\ 0 &  \frac{1}{k} \end{smallmatrix}\bigr]$.
	In the limit, when $k$ goes to infinity, we have:
	\begin{equation*}
		\lim_{k \to \infty} \|A - X_kY_k^\top\|^2 = \lim_{k \to \infty} \frac{1}{k^2} = 0.
	\end{equation*}
	Yet, there does not exist any couple $(X, Y)$ such that $\|A - XY^\top\|^2 = 0$. Indeed, any such couple would need to satisfy: $X_{1,2}Y_{2,2} = 1, X_{2,2}Y_{1,2} = 1, X_{2,2}Y_{2,2} = 0$. However, the third equation implies that either $X_{2,2} = 0$ or $Y_{2,2} = 0$, which makes either $X_{2,2}Y_{1,2} = 0$ or $X_{1,2}Y_{2,2} = 0$. This leads to a contradiction.
	
	In fact, $I$ and $J$ are constructed from the weight binary matrix $W = \bigl[\begin{smallmatrix}	0 & 1 \\ 1 & 1 \end{smallmatrix}\bigr]$ (the construction is similar to one in the proof of \cref{lem: NPhardness}). Problem \eqref{eq: WLRA1} with $(A, W)$ has unattainable infimum as well. {Note that this choice of $(I,J)$ also makes this instance of \eqref{eq: fixed_supp_prob} equivalent to the problem of $\mbf{LU}$ decomposition of matrix $A$.}
\end{remark}

\section{Proofs for \cref{sec:easyinstance}}
\label{app:resultsec3}
\subsection{Proof of \cref{lem: expressibility_two}}
\label{subapp:expressibilitytwo}
{Denote $\mc{P}$ the partition of $\intset{r}$ into equivalence classes defined by the rank-one supports associated to $(I,J)$, and $\mc{P}^{\star} \subseteq \mc{P}$ the corresponding CECs. Since $T \subseteq \intset{r}$ is precisely the set of indices of CECs, and since $I_{T}$ (resp. $J_{T}$) is the restriction of $I$ (resp. of $J$) to columns indexed by $T$, the partition of 
	$\intset{r}$ into equivalence classes \emph{w.r.t $(I_T,J_T)$} is precisely $\mc{P}^\star$, and for $P \in \mc{P} \backslash \mc{P}^*$, we have $\mS_P = \emptyset$. WLOG, we assume $\mc{P}^\star = \{P_i \mid 1 \leq i \leq \ell\}$}. Denote $\mc{P}_k = \{P_1, \ldots, P_k\}$, $\mS_{\mc{P}_k} = \cup_{1 \leq i \leq k} \mS_{P_i}$ for $1 \leq k \leq \ell$ and $\mS_{P_0} = \emptyset$. 
{We prove below that $(X,Y) = \text{SVD\_FSMF}(A,I_{T},J_{T})$ satisfies:}
\begin{equation}
	\label{eq:lemmaCEC}
	X_{P_k}Y_{P_k}^\top = A \odot (\mS_{\mc{P}_k} \setminus \mS_{\mc{P}_{k-1}}),  {\forall\ 1 \leq k \leq \ell,}
\end{equation} 
{which implies}: $XY^\top = \sum_{P \in \mc{P}^\star} X_PY_P^\top = \sum_{k = 1}^\ell A \odot (\mS_{\mc{P}_k} \setminus \mS_{\mc{P}_{k-1}}) = A \odot \mS_\ell = A \odot \mS_T = A$ (since we assume $\supp(A) = \mS_T$). {This yields the conclusion since $\supp(X) \subseteq I_{T}$ and $\supp(Y) \subseteq J_{T}$ by definition of $\text{SVD\_FSMF}(\cdot)$.}

We prove \Cref{eq:lemmaCEC} by induction {on $\ell$}. To ease the reading, in this proof, we denote $C_{P_k}, R_{P_k}$ (\cref{def:completeeqclass2}) by $C_k, R_k$ respectively.

For {$\ell=1$ it is sufficient to consider} $k = 1$: we have $\mS_{\mc{P}_1} \setminus \mS_{\mc{P}_0} = C_1 \times R_1$. 
{Since $\min(|R_1|, |C_1|) \leq |P_1|$ (\Cref{def:completeeqclass2}), taking the best rank-$|P_1|$ approximation of $A \odot (R_1 \times C_1)$ (whose rank is at most $\min(|R_1|, |C_1|)$) yields
	$X_{P_1}Y_{P_1}^\top = A \odot (R_1 \times C_1) = A \odot (\mS_{\mc{P}_1} \setminus \mS_{\mc{P}_0})$.}

Assume that \Cref{eq:lemmaCEC} holds for $\ell-1$. We prove its correctness {for} $\ell$. {Consider: $A' := A - \sum_{k < \ell}X_{P_k}Y_{P_k}^\top = A - A \odot \mS_{\mc{P}_{\ell-1}} = A \odot \bar{\mS}_{\mc{P}_{\ell-1}}$. Therefore, $A' \odot \mS_{P_{\ell}} = A \odot (\mS_{\mc{P}_\ell} \setminus \mS_{\mc{P}_{\ell-1}})$. Again, since $\min(|R_\ell|, |C_\ell|) \leq |P_\ell|$ (\Cref{def:completeeqclass2}), taking the best rank-$|P_\ell|$ approximation of $A' \odot \mS_{P_{\ell}} = A' \odot (R_\ell \times C_\ell)$ (whose rank is at most $\min(|R_\ell|, |C_\ell|)$)  yields $X_{P_\ell}Y_{P_\ell}^\top = A' \odot (R_\ell \times C_\ell) = A \odot (\mS_{P_\ell} \setminus \mS_{P_{\ell-1}})$. That implies \Cref{eq:lemmaCEC} is correct for all $\ell$.}

\subsection{Proof of \cref{theorem: reduction_disjoint_overlapping}}
\label{subapp: reduction_disjoint_overlapping}
First, we decompose the factors $X$ and $Y$ using the taxonomy of indices from \cref{def:taxonomy}.
\begin{definition}
	\label{def:taxonomy2}
	Given $I_T, J_T$ and $I_{\bar{T}}^i, J_{\bar{T}}^i, i = 1, 2$ as in \cref{def:taxonomy}, consider $(X,Y)$ a feasible point of \eqref{eq: fixed_supp_prob}, we denote:
	\begin{itemize}
		\item[1)] $X_T = X \odot I_T, X_{\bar{T}}^i = X \odot I_{\bar{T}}^i$, for $i=1,2$.
		
		\item[2)] $Y_T = Y \odot I_T, Y_{\bar{T}}^i = Y \odot I_{\bar{T}}^i$, for $i=1,2$.
	\end{itemize} 
	with $\odot$ the Hadamard product between a matrix and a support constraint (introduced in \cref{sec:notation}).
\end{definition}

The following is a technical result.
\begin{lemma}
	\label{lem:taxonomy}
	Given $I, J$ support constraints of \eqref{eq: fixed_supp_prob}, consider $T,\mathcal{S}_{T},\mathcal{S}_{\mathcal{P}}$ as in \cref{def:completeeqclass}, $X_T, X_{\bar{T}}^i, Y_T, Y_{\bar{T}}^i$ as in \cref{def: supp_outsideCEC} and assume that for all $k \in \bar{T}$, $\mS'_k$ is rectangular. It holds:
	\begin{enumerate}[label=\textbf{C\arabic*}]
		\item \label{eq:lemtax1}$\supp(X_TY_T^\top) \subseteq \mS_T$.
		\item \label{eq:lemtax2}$\supp(X_{\bar{T}}^1(Y_{\bar{T}}^1)^\top) \subseteq \mS_\mc{P} \setminus \mS_T$.
		\item \label{eq:lemtax3}$\supp(X_{\bar{T}}^i(Y_{\bar{T}}^j)^\top) \subseteq \mS_T, \forall 1 \leq i, j \leq 2, (i,j) \neq (1,1)$.
	\end{enumerate}
\end{lemma}
\begin{proof}
	We justify \eqref{eq:lemtax1}-\eqref{eq:lemtax3} as follow:
	\begin{itemize}
		\item \ref{eq:lemtax1}: Since $X_TY_T^\top = \sum_{i \in T}X_\col{i}Y_\col{i}^\top$, $\supp(X_TY_T^\top) \subseteq \cup_{i \in T} \mS_k = \mS_T$.
		\item \ref{eq:lemtax2}: Consider the coefficient $(i,j)$ of $(X_{\bar{T}}^1)(Y_{\bar{T}}^1)^\top$
		\begin{equation*}
			((X_{\bar{T}}^1)(Y_{\bar{T}}^1)^\top)_{i,j} = \sum_k (X_{\bar{T}}^1)_{i,k} (Y_{\bar{T}}^1)_{j,k} = \sum_{(i,k) \in I_{\bar{T}}^1, (j,k) \in J_{\bar{T}}^1} X_{i,k} Y_{j,k}
		\end{equation*}
		By the definition of $I_{\bar{T}}^1, J_{\bar{T}}^1$, $(X_{\bar{T}}^1)(Y_{\bar{T}}^1)^\top_{i,j} \neq 0$ iff $(i,j) \in \cup_{\ell \in \bar{T}} R_\ell \times C_\ell = \mS_\mc{P} \setminus \mS_T$. 
		\item \ref{eq:lemtax3}: We prove for the case of $(X_{\bar{T}}^1)(Y_{\bar{T}}^2)^\top$. Others can be proved similarly.
		\begin{equation}
			\label{eq:cannotinPsi2}
			((X_{\bar{T}}^1)(Y_{\bar{T}}^2)^\top)_{i,j} = \sum_k  (X_{\bar{T}}^1)_{i,k} (Y_{\bar{T}}^2)_{j,k} = \sum_{(i,k) \in I_{\bar{T}}^1, (j,k) \in J_{\bar{T}}^2} X_{i,k} Y_{j,k}
		\end{equation}
		Since $\forall \ell \in \bar{T}, \mS'_\ell$ is rectangular, $\mS_\mc{P} \setminus \mS_T = \cup_{\ell \in \bar{T}} \mS'_\ell = \cup_{\ell \in \bar{T}} R_{\ell} \times C_{\ell}$. If $(i,j) \in \mS_\mc{P} \setminus \mS_T$, \Cref{eq:cannotinPsi2} shows that $((X_{\bar{T}}^1)(Y_{\bar{T}}^2)^\top)_{i,j} = 0$ since there is no $k$ such that $(i, k) \in I_{\bar{T}}^2, (j, k) \in J_{\bar{T}}^2$ due to the definition of $I_{\bar{T}}^1, J_{\bar{T}}^2$). Moreover,  $\supp((X_{\bar{T}}^1)(Y_{\bar{T}}^2)^\top) \subseteq \mS_\mc{P}$ (since $\supp(X_{\bar{T}}^1) \subseteq I, \supp(Y_{\bar{T}}^2) \subseteq J$). Thus, it shows that $\supp((X_{\bar{T}}^1)(Y_{\bar{T}}^2)^\top) \subseteq  \mS_\mc{P} \setminus (\mS_\mc{P} \setminus \mS_T) = \mS_T$. 
	\end{itemize}
\end{proof}

Here, we present the proof of \cref{theorem: reduction_disjoint_overlapping}.
\begin{proof}[Proof of \cref{theorem: reduction_disjoint_overlapping}]
	Given $X,Y$ feasible point of the input $(A,I,J)$, consider $X_T, Y_T, X_{\bar{T}}^i, Y_{\bar{T}}^i, i = 1,2$ defined as in \cref{def:taxonomy2}. Let $\mu_1$ and $\mu_2$ be the infimum value of \eqref{eq: fixed_supp_prob} with $(A, I, J)$ and with $(A', I_{\bar{T}}^1, J_{\bar{T}}^1)$ ($A' = A \odot {\bar{\mS}_T}$) respectively. 
	
	First, we remark that $I_{\bar{T}}^1$ and $J_{\bar{T}}^1$ satisfy the assumptions of \cref{theorem:disjoint_totally_overlapping}. Indeed, it holds $\mS_{k}(I_{\bar{T}}^1,J_{\bar{T}}^1) =\mS_k(I,J) \setminus \mS_T= \mS'_{k}$ by construction. For any two indices $k,l \in \bar{T}$, the representative rank-one supports are either equal ($\mS'_k = \mS'_l$) or disjoint ($\mS'_k \cap \mS'_l = \emptyset$) by assumption. That shows why $I_{\bar{T}}^1$ and $J_{\bar{T}}^1$ satisfy the assumptions of \cref{theorem:disjoint_totally_overlapping}.
	
	Next, we prove that $\mu_1 = \mu_2$. Since $ (\mS_{T},\mS_{\mathcal{P}}\setminus\mS_{T},\bar{\mS}_{\mathcal{P}})$ form a partition of $\llbracket m\rrbracket \times \llbracket n\rrbracket$, we have ${C} \odot {D} = \mathbf{0}$, $C \neq D, C,D \in \{\mS_{T},\mS_{\mathcal{P}}\setminus\mS_{T},\bar{\mS}_{\mathcal{P}}\}$. From the definition of $A'$ it holds $A' \odot {\bar{\mS}_\mc{P}} = A \odot {\bar{\mS}_\mc{P}}$ and $A' \odot {\mS_T} =\mathbf{0}$. Moreover, it holds $
	(X_{\bar{T}}^1)(Y_{\bar{T}}^1)^\top \odot {\mS_T \cup \bar{\mS}_\mc{P}} = \mathbf{0}$ due to \ref{eq:lemtax2}. 
	
	Since $\supp(X_T) \subseteq I_T, \supp(X_{\bar{T}}^i) \subseteq I_{\bar{T}}, \supp(Y_T) \subseteq J_T, \supp(Y_{\bar{T}}^i) \subseteq J_{\bar{T}}, i = 1, 2$, the product $XY^\top$ can be decomposed as:
	\begin{equation}
		\label{eq:decomposition_lemma}
		XY^\top = X_TY_T^\top + \sum_{1 \leq i,j \leq 2} (X_{\bar{T}}^i)(Y_{\bar{T}}^j)^\top.
	\end{equation}
	Consider the loss function of \eqref{eq: fixed_supp_prob} with input $(A', I_{\bar{T}}^1, J_{\bar{T}}^1)$ and solution $(X_{\bar{T}}^1, Y_{\bar{T}}^1)$:
	\begin{equation}
		\label{eq:equivalent1}
		\begin{aligned}
			&\|A'-X_{\bar{T}}^1(Y_{\bar{T}}^1)^\top\|^2\\
			&= \|(A'-X_{\bar{T}}^1(Y_{\bar{T}}^1)^\top) \odot {\mS_T}\|^{2}
			+ \|(A'-X_{\bar{T}}^1(Y_{\bar{T}}^1)^\top) \odot {(\mS_\mc{P} \setminus \mS_T)}\|^{2}\\
			&\qquad \qquad + \|(A'-X_{\bar{T}}^1(Y_{\bar{T}}^1)^\top) \odot {\bar{\mS}_\mc{P}}\|^{2}\\
			& \overset{\ref{eq:lemtax2}}{=} \|(A'-(X_{\bar{T}}^1)(Y_{\bar{T}}^1)^\top) \odot { \mS_\mc{P} \setminus \mS_T}\|^{2}+ \|A' \odot {\bar{\mS}_\mc{P}}\|^{2}\\
			&\overset{\ref{eq:lemtax1} + \ref{eq:lemtax3}}{=}  \|(A-X_TY_T^\top - \sum_{1 \leq i,j \leq 2} (X_{\bar{T}}^i)(Y_{\bar{T}}^j)^\top) \odot {(\mS_\mc{P} \setminus \mS_T)}\|^{2}+ \|A \odot {\bar{\mS}_\mc{P}}\|^{2}\\
			&\overset{\eqref{eq:decomposition_lemma}}{=}  \|(A-XY^\top) \odot {(\mS_\mc{P} \setminus \mS_T)}\|^{2}+ \|A \odot {\bar{\mS}_\mc{P}}\|^{2}\\
		\end{aligned}
	\end{equation}
	Perform the same calculation with $(A,I,J)$ and solution $(X,Y)$:
	\begin{equation}
		\label{eq:equivalent2}
		\begin{aligned}
			& \|(A-XY^\top)\|^{2}\\
			& = \|(A-XY^\top) \odot {\mS_T}\|^{2}
			+ \|(A-XY^\top) \odot {(\mS_\mc{P} \setminus \mS_T)}\|^{2} + \|(A-XY^\top) \odot {\bar{\mS}_\mc{P}}\|^{2}\\
			&= \|(A-XY^\top) \odot {\mS_T}\|^{2}
			+ \|(A-XY^\top) \odot {(\mS_\mc{P} \setminus \mS_T)}\|^{2} + \|A \odot {\bar{\mS}_\mc{P}}\|^{2}\\
		\end{aligned}
	\end{equation}
	where the last equality holds since $\supp(XY^\top) \subseteq \mS_\mc{P}$. Therefore, for any feasible point $(X,Y)$ of instance $(A, I, J)$, we can choose $\te{X} = X_{\bar{T}}^1, \te{Y} = Y_{\bar{T}}^1$ feasible point of $(A', I_{\bar{T}}^1, J_{\bar{T}}^1)$ such that $\|A - XY^\top\| \geq \|A' - \te{X}\te{Y}^\top\|$ (\Cref{eq:equivalent1} and \Cref{eq:equivalent2}). This shows $\mu_1 \geq \mu_2$.
	
	On the other hand, given any feasible point $(\te{X}, \te{Y})$ of instance $(A', I_{\bar{T}}^1, J_{\bar{T}}^1)$, we can construct a feasible point $(X, Y)$ for instance $(A, I, J)$ such that  $\|A - XY^\top\|^2 = \|A' - X'Y'^\top\|^2$. We construct $(X,Y) = (X_T + X_{\bar{T}}^1 + X_{\bar{T}}^2, Y_T + Y_{\bar{T}}^1 + Y_{\bar{T}}^2)$ where:
	\begin{itemize}
		\item[1)] $X_{\bar{T}}^1 = \te{X}, Y_{\bar{T}}^1 = \te{Y}$,
		
		\item[2)] $X_{\bar{T}}^2, Y_{\bar{T}}^2$ can be chosen arbitrarily such that $\supp(X_{\bar{T}}^2) \subseteq I_{\bar{T}}^2, \supp(Y_{\bar{T}}^2) \subseteq J_{\bar{T}}^2$
		
		\item[3)]  $X_T$ and $Y_T$ such that $\supp(X_T) \subseteq I_T, \supp(Y_T) \subseteq J_T$ and:
		\begin{equation*}
			X_TY_T^\top = (A - (X_{\bar{T}}^1 + X_{\bar{T}}^2)(Y_{\bar{T}}^1 + Y_{\bar{T}}^2)^\top) \odot {\mS_T}
		\end{equation*}
	\end{itemize}
	$(X_T,Y_T)$ exists due to \cref{lem: expressibility_two}. By \cref{lem:taxonomy}, with this choice we have:
	\begin{equation}
		\label{eq:conditionoptimal}
		\begin{split}
			(A - XY^\top) \odot {\mS_T} &\overset{\eqref{eq:decomposition_lemma}}{=} (A - (X_{\bar{T}}^1 + X_{\bar{T}}^2)(Y_{\bar{T}}^1 + Y_{\bar{T}}^2)^\top - X_TY_T^\top) \odot {\mS_T}\\
			&\overset{\ref{eq:lemtax1}}{=} (A - (X_{\bar{T}}^1 + X_{\bar{T}}^2)(Y_{\bar{T}}^1 + Y_{\bar{T}}^2)^\top) \odot {\mS_T}) - X_TY_T^\top =\mathbf{0}			
		\end{split}
	\end{equation}
	Therefore $\|A - XY^\top\|^2 = \|A' - \te{X}\te{Y}^\top\|^2$ (\Cref{eq:equivalent1} and \Cref{eq:equivalent2}). Thus, $\mu_2 \geq \mu_1$. We obtain $\mu_1 = \mu_2$. In addition, given $(X, Y)$ an optimal solution of \eqref{eq: fixed_supp_prob} with instance $(A, I, J)$, we have shown how to construct an optimal solution $(\te{X}, \te{Y})$ with instance $(A \odot {\bar{\mS}_T}, I_{\bar{T}}^1, J_{\bar{T}}^1)$ and vice versa. That completes our proof.
\end{proof}

The following Corollary is a direct consequence of the proof of \cref{theorem: reduction_disjoint_overlapping}.
\begin{corollary}
	\label{cor:optimalsol}
	With the same assumptions and notations as in \cref{theorem: reduction_disjoint_overlapping}, a feasible point $(X, Y)$ (i.e., such that $\supp(X) \subseteq I, \supp(Y) \subseteq J$) is an optimal solution of \eqref{eq: fixed_supp_prob} if and only if:
	\begin{itemize}
		\item[1)] $(X \odot I_{\bar{T}}^1, Y \odot J_{\bar{T}}^1)$ is an optimal solution of \eqref{eq: fixed_supp_prob} with $(A \odot {\bar{\mS}_T}, I_{\bar{T}}^1, J_{\bar{T}}^1)$. 
		
		\item[2)] The following equation holds:	$(A - XY^\top) \odot {\mS_T} = \mathbf{0}$
	\end{itemize}
\end{corollary}

\begin{remark}
	\label{rem:optimalsol}
	In the proof of \cref{theorem: reduction_disjoint_overlapping}, for an optimal solution, one can choose $X_{\bar{T}}^2, Y_{\bar{T}}^2$ arbitrarily. If we choose $X_{\bar{T}}^2 = \mathbf{0}, Y_{\bar{T}}^2 = \mathbf{0}$, thanks to \cref{eq:conditionoptimal}, $X_T$ and $Y_T$ has to satisfy:
	\begin{equation*}
		X_TY_T^\top = (A - (X_{\bar{T}}^1 + X_{\bar{T}}^2)(Y_{\bar{T}}^1 + Y_{\bar{T}}^2)^\top) \odot {\mS_T} = (A - X_{\bar{T}}^1(Y_{\bar{T}}^1)^\top) \odot {\mS_T} \overset{\ref{eq:lemtax2}}{=} A \odot {\mS_T}
	\end{equation*}
\end{remark}

\section{Proofs for a key lemma}
In this section, we will introduce an important technical lemma. It is used extensively for the proof of the tractability and the landscape of \eqref{eq: fixed_supp_prob} under the assumptions of \cref{theorem: reduction_disjoint_overlapping}, {cf.  \cref{sec:proof_4}}.

\begin{lemma}
	\label{lem: continous_function_of_S2}
	Consider $I ,J$ support constraints of \eqref{eq: fixed_supp_prob} such that $\mc{P}^\star = \mc{P}$. For any CEC-full-rank feasible point $(X,Y)$ and continuous function $g: [0,1] \to \mb{R}^{m \times n}$ satisfying $\supp(g(t)) \subseteq \mS_T$ (\cref{def:completeeqclass2}) and $g(0) = XY^\top$, there exists a feasible continuous function $f: [0,1] \to  \mb{R}^{m \times r} \times \mb{R}^{n \times r}: f(t) = (X_f(t), Y_f(t))$ such that:
	\begin{enumerate}[label=\textbf{A\arabic*}]	
		\item \label{eq:lemS21}$f(0) = (X_T, Y_T)$.
		
		\item \label{eq:lemS22}$g(t) = X_f(t)Y_f(t)^\top, \forall t \in [0,1]$.
		
		\item \label{eq:lemS23}$\|f(z) - f(t)\|^2 \leq \mc{C}\|g(z) - g(t)\|^2,  \forall t,z \in [0,1]$.
	\end{enumerate}
	where $\mc{C} =  \underset{P \in \mc{P}^\star}{\max} \left(\max\left(\vertiii{X_{R_{P}, P}^\dagger}^2, \vertiii{Y_{C_{P}, P}^\dagger}^2\right)\right)$ ($D^\dagger$ and $\vertiii{D}$ denote the pseudo-inverse and operator norm of a matrix $D$ respectively ).
\end{lemma}

\Cref{lem: continous_function_of_S2} consider the case where $\mc{P}$ only contains CECs. Later in other proofs, we will control the factors $(X,Y)$ by decomposing $X = X_T + X_{\bar{T}}$ (and $Y = Y_T + X_{\bar{T}}$) ($T, \bar{T}$ defined in \cref{def:completeeqclass2}) and manipulate $(X_T,Y_T)$ and $(X_{\bar{T}}, Y_{\bar{T}})$ separately. Since the supports of $(X_T,Y_T)$ satisfy \cref{lem: continous_function_of_S2}, it provides us a tool to work with $(X_T,Y_T)$.

The proof of \cref{lem: continous_function_of_S2} is carried out by induction. We firstly introduce and prove two other lemmas: \cref{lemma: continuous function adapt} and \cref{lem: continuous_function_of_S}. While \cref{lemma: continuous function adapt} is \cref{lem: continous_function_of_S2} without support constraints, \cref{lem: continuous_function_of_S} is \cref{lem: continous_function_of_S2} where $|\mc{P}^\star| = 1$.
\begin{lemma}
	\label{lemma: continuous function adapt}
	Let $X \in \mb{R}^{m \times r}, Y \in \mb{R}^{n \times r}, \min(m,n) \leq r$ and assume that $X$ or $Y$ has full row rank. Given any continuous function $g: [0,1] \to \mb{R}^{m \times n}$ in which $g(0) = XY^\top$, there exists a continuous function $f: [0,1] \to  \mb{R}^{m \times r} \times \mb{R}^{n \times r}: f(t) = (X_f(t), Y_f(t))$ such that:
	\begin{itemize}
		\item[1)] $f(0) = (X, Y)$.
		
		\item[2)] $g(t) = X_f(t)Y_f(t)^\top, \forall t \in [0,1]$.
		
		\item[3)] $\|f(z) - f(t)\|^2 \leq \mc{C}\|g(z) - g(t)\|^2, \forall t,z \in [0,1]$. 
	\end{itemize}
	where $\mc{C} =  \max\left(\vertiii{X^\dagger}^2, \vertiii{Y^\dagger}^2\right)$.
\end{lemma}
\begin{proof}
	WLOG, we can assume that $X$ has full row rank. We define $f$ as:
	\begin{equation}
		\label{eq:constructf(t)}
		\begin{aligned}
			X_f(t) &= X\\
			Y_f(t) & = Y + (g(t) - g(0))^\top(XX^\top)^{-1}X = Y + (X^\dagger(g(t) - g(0)))^\top
		\end{aligned}
	\end{equation}
	where $X^\dagger = X^\top(XX^\top)^{-1}$ the pseudo-inverse of $X$. The function $Y_f$ is well-defined due to the assumption of $X$ being full row rank. It is immediate for the first two constraints. Since $\|f(z) - f(t)\|^2 = \|Y_f(z) - Y_f(t)\|^2 = \|X^\dagger(g(z) - g(t))\|^2$, the third one is also satisfied as:
	\begin{equation*}
		\|f(z) - f(t)\|^2 =  \|X^\dagger(g(z) - g(t))\|^2 \leq \vertiii{X^\dagger}^2 \|g(z) - g(t)\|^2 \leq \mc{C}\|g(z) - g(t)\|^2
	\end{equation*}
\end{proof}

\begin{lemma}
	\label{lem: continuous_function_of_S}
	Consider $I, J$ support of \eqref{eq: fixed_supp_prob} where $\mc{P}^\star = \mc{P} = \{P\}$, for any feasible CEC-full-rank point $(X,Y)$ and continuous function $g: [0,1] \to \mb{R}^{m \times n}$ satisfying $\supp(g(t)) \subseteq \mS_P$ (\cref{def:completeeqclass}) and $g(0) = XY^\top$, there exists a feasible continuous function $f: [0,1] \to  \mb{R}^{m \times r} \times \mb{R}^{n \times r}: f(t) = (X_f(t), Y_f(t))$ such that:
	\begin{enumerate}[label=\textbf{B\arabic*}]
		\item \label{eq:lemS1} $f(0) = (X,Y)$.
		
		\item \label{eq:lemS2} $g(t) = X_f(t)Y_f(t)^\top, \forall t \in [0,1]$.
		
		\item \label{eq:lemS3} $\|f(z) - f(t)\|^2 \leq \mc{C}\|g(z) - g(t)\|^2$.
	\end{enumerate}
	where $\mc{C} =  \max\left(\vertiii{X_{R_P,P}^\dagger}^2, \vertiii{Y_{C_P,P}^\dagger}^2\right)$.
\end{lemma}

\begin{proof}
	WLOG, we assume that $P = \intset{|P|}, R_P = \intset{|R_P|}, C_P = \intset{|C_P|}$. Furthermore, we can assume $|P| \geq |R_P|$ and $X_{R_P,P}$ is full row rank (due to the hypothesis and the fact that $P$ is complete).
	
	Since $\mc{P}^\star = \mc{P} = \{P\}$, a continuous feasible function $f(t)$ must have the form: $X_f(t) = \bigl[\begin{smallmatrix} \te{X}_f(t) & \mathbf{0} \\ \mathbf{0} & \mathbf{0} \end{smallmatrix}\bigr]$ and $Y_f(t) = \bigl[\begin{smallmatrix} \te{Y}_f & \mathbf{0} \\ \mathbf{0} & \mathbf{0} \end{smallmatrix}\bigr]$ where $\te{X}_f: [0,1] \to \mb{R}^{|R_P| \times |P|}, \te{Y}_f: [0,1] \to \mb{R}^{|C_P| \times |P|}$ are continuous functions. $f$ is fully determined by $(\te{X}_f(t),\te{Y}_f(t))$.
	
	Moreover, if $g: [0,1] \to \mb{R}^{m \times n}$ satisfying $\supp(g(t)) \subseteq \mS_T$, then $g$ has to have the form: $g(t) = \bigl[\begin{smallmatrix} \te{g} & \mathbf{0} \\ \mathbf{0} & \mathbf{0} \end{smallmatrix}\bigr]$ where $\te{g}: [0,1] \to \mb{R}^{|R_P| \times |C_P|}$ is a continuous function. 
	
	Since $g(0) = XY^\top$, $\te{g}(0) = (X_{R_P,P})(Y_{C_P,P})^\top$. Thus, to satisfy each constraint \ref{eq:lemS1}-\ref{eq:lemS3}, it is sufficient to find $\te{X}_f$ and $\te{Y}_f$ such that:
	\begin{itemize}[label={}]
		\item \ref{eq:lemS1}: $\te{X}_f(0) = X_{R_P, P}, \te{Y}_f(0) = Y_{C_P, P}$.
		
		\item \ref{eq:lemS2}: $\te{g}(t) = \te{X}_f(t)\te{Y}_f(t)^\top, \forall t \in [0,1]$ because:
		\begin{equation*}
			X_f(t)Y_f(t)^\top = \begin{pmatrix}
				\te{X}_f(t)\te{Y}_f(t)^\top & \mathbf{0}\\
				\mathbf{0} & \mathbf{0}
			\end{pmatrix}
			= \begin{pmatrix}
				\te{g}(t) & \mathbf{0}\\
				\mathbf{0} & \mathbf{0}
			\end{pmatrix} = g(t)
		\end{equation*}
		\item \ref{eq:lemS3}: $\|X'(z) - X'(t)\|^2 + \|Y'(z) - Y'(t)\|^2 \leq  \mc{C}\|A'(z) - A'(t)\|^2$ since $\|X'_f(z) - X_f(t)\|^2 + \|Y'_f(z) - Y_f(t)\|^2 = \|f(z) - f(t)\|^2$ and $\|A'(z) - A'(t)\|^ = \|g(z) - g(t)\|^2 $.
	\end{itemize}
	Such function exists thanks \cref{lemma: continuous function adapt} (since we assume $X_{R_P,P}$ has full rank).
\end{proof}

\begin{proof}[Proof of \cref{lem: continous_function_of_S2}]
	We prove by induction on the size $\mc{P}$. By \cref{lem: continuous_function_of_S} the result is true if $|\mathcal{P}|=1$. Assume the result is true if $|\mathcal{P}| \leq p$. We consider the case where $|\mathcal{P}|=p+1$. Let $P \in \mc{P}$ and partition $\mc{P}$ into $\mc{P}' = \mc{P} \setminus \{P\}$ and $\{P\}$. Let $T' = \cup_{P' \in \mc{P}'} P' = T \setminus P$. Since $|\mc{P}'| = p$, we can use induction hypothesis. Define:
	\begin{equation*}
		h_1(t) = (g(t) - X_PY_P^\top) \odot {\mc{S}_{\mc{P}'}}, \qquad h_2(t) = X_PY_P^\top \odot {\mc{S}_{\mc{P}'}} + g(t) \odot {\mc{S}_P \setminus \mc{S}_{\mc{P}'}}
	\end{equation*}
	We verify that the function $h_1(t)$ satisfying the hypotheses to use induction step: $h_1$ continuous, $\supp(h_1(t)) \subseteq \mc{S}_{\mc{P}'}$ and finally $h_1(0) = (g(0) - X_PY_P^\top) \odot {\mc{S}_{\mc{P}'}} = X_{T'}Y_{T'}^\top \odot {\mc{S}_{\mc{P}'}} = X_{T'}Y_{T'}^\top$. Using the induction hypothesis with $\mc{P}'$, there exists a function $f_1:[0,1] \to \mb{R}^{m \times r} \times \mb{R}^{n \times r}: f_1(t) = (X_f^1(t), Y_f^1(t))$ such that:
	\begin{itemize}
		\item[1)] $\supp(X_f^1(t)) \subseteq I_{T'}, \supp(Y_f^1(t)) \subseteq J_{T'}$.
		
		\item[2)] $f_1(0) = (X_{T'}, Y_{T'})$.
		
		\item[3)] $h_1(t) = X_f^1(t)Y_f^1(t)^\top, \forall t \in [0,1]$.
		
		\item[4)] $\|f_1(z) - f_1(t)\|^2 \leq \mc{C}'\|h_1(z) - h_1(t)\|^2$.
	\end{itemize}
	where $\mc{C}' =  \underset{P' \in \mc{P}'}{\max} \left(\max\left(\vertiii{X_{R_{P'}, P'}^\dagger}^2, \vertiii{Y_{C_{P}, P}^\dagger}^2\right)\right)$. 
	
	On the other hand, $h_2(t)$ satisfies the assumptions of \cref{lem: continuous_function_of_S}: $h_2(t)$ is continuous and $\supp(h_2(t)) = \supp(X_PY_P^\top \odot {\mc{S}_\mc{P'}} + g(t) \odot {\mc{S}_P \setminus \mc{S}_\mc{P'}}) \subseteq \supp(X_PY_P^\top) \cup (\mc{S}_P \setminus \mc{S}_\mc{P'}) = \mS_P$. 
	
	In addition, since $g(0) \odot {\mc{S}_P \setminus \mc{S}_\mc{P'}} = (XY^\top) \odot {\mc{S}_P \setminus \mc{S}_\mc{P'}} = (X_{T'}Y_{T'}^\top + X_PY_P^\top) \odot {\mc{S}_P \setminus \mc{S}_\mc{P'}} = X_PY_P^\top \odot {\mc{S}_P \setminus \mc{S}_\mc{P'}}$, we have $h_2(0) = X_PY_P^\top \odot {\mc{S}_\mc{P'}} + g(0) \odot {\mc{S}_P \setminus\mc{S}_\mc{P'}} = X_PY_P^\top \odot ({\mc{S}_\mc{P'}} + {\mc{S}_P \setminus \mc{S}_\mc{P'}}) = X_PY_P^\top$. Invoking \cref{lem: continuous_function_of_S} with the singleton $\{P\}$, there exists a function $(X_f^2(t), Y_f^2(t))$ such that:
	\begin{itemize}
		\item[1)] $\supp(X_f^2(t)) \subseteq I_P, \supp(Y_f^2(t)) \subseteq J_P$.
		
		\item[2)] $f_2(0) = (X_{P}, Y_{P})$.
		
		\item[3)] $h_2(t) = X_f^2(t)Y_f^2(t)^\top, \forall t \in [0,1]$.
		
		\item[4)] $\|f_2(z) - f_2(t)\|^2 \leq \max\left(\vertiii{X_{R_{P}, P}^\dagger}^2, \vertiii{Y_{C_{P}, P}^\dagger}^2\right) \|h_2(z) - h_2(t)\|^2$.
	\end{itemize}
	We construct the functions $f(t) = (X_f(t), Y_f(t))$ as:
	\begin{equation*}
		X_f(t) = X_f^1(t) + X_f^2(t), \qquad Y_f(t) = Y_f^1(t) + Y_f^2(t)
	\end{equation*}
	We verify the validity of this construction. $f$ is clearly feasible due to the supports of $X_f^i(t), Y_f^i(t), i = 1, 2$. The remaining conditions are:
	\begin{itemize}[label={}]
		\item \ref{eq:lemS21}:	
		\begin{equation*}
			\begin{aligned}
				X_f(0) &= X_f^1(0) + X_f^2(0) = X_{T'} + X_P = X\\ 
				Y_f(0) &= Y_f^1(0) + Y_f^2(0) = Y_{T'} + Y_P = Y\\ 
			\end{aligned}
		\end{equation*}
		\item \ref{eq:lemS22}:
		\begin{equation*}
			\begin{split}
				X_f(t)Y_f(t)^\top &= X_f^1(t)Y_f^1(t)^\top + X_f^2(t)Y_f^2(t)^\top\\
				&= h_1(t) + h_2(t)\\
				&= (g(t) - X_PY_P^\top) \odot {\mc{S}_{\mc{P}'}} + X_PY_P^\top \odot {\mc{S}_{\mc{P}'}} + g(t) \odot {\mc{S}_{P} \setminus \mc{S}_{\mc{P}'}}\\
				&= g(t) \odot ({\mc{S}_{\mc{P}'}} + {\mc{S}_{P} \setminus \mc{S}_{\mc{P}'}}) = g(t)
			\end{split}
		\end{equation*}
		\item \ref{eq:lemS23}:
		\begin{equation*}
			\begin{split}
				&\|f(z) - f(t)\|^2\\
				&= \|f_1(z) - f_1(t)\|^2 + \|f_2(z) - f_2(t)\|^2\\
				&\leq \mc{C}'\|h_1(z) - h_1(t)\|^2 + \max\left(\vertiii{X_{R_{P}, P}^\dagger}^2, \vertiii{Y_{C_{P}, P}^\dagger}^2\right) \|h_2(z) - h_2(t)\|^2\\
				&\leq \mc{C} (\|h_1(z) - h_1(t)\|^2 + \|h_2(z) - h_2(t)\|^2)\\
				&= \mc{C} (\|(g(z) - g(t)) \odot {\mc{S}_{\mc{P}'}}\|^2 + \|(g(z) - g(t)) \odot {\mc{S}_{P} \setminus \mc{S}_{\mc{P}'}}\|^2)\\
				&= \mc{C} \|g(z) - g(t)\|^2
			\end{split}
		\end{equation*}
	\end{itemize}
\end{proof}

\section{Proofs for \cref{sec: spuriouslocal}}

\subsection{Proof of \cref{lem:fullrankhypothesis}}
\label{subapp:fullrankhypothesis}
The proof relies on two intermediate results that we state first:  \cref{lemma: to_inversible} and \cref{cor:to_fullrank}. The idea of \cref{lemma: to_inversible} can be found in \cite{venturi2020spurious}. Since it is not formally proved as a lemma or theorem, we reprove it here for self-containedness. In fact, \cref{lemma: to_inversible} and \cref{cor:to_fullrank} are special cases of \cref{lem:fullrankhypothesis} with no support contraints and $\mc{P}^\star = \mc{P} = \{P\}$ respectively.

\begin{lemma}
	\label{lemma: to_inversible}
	Let $X \in \mb{R}^{R \times p}, Y \in \mb{R}^{C \times p}, \min(R, C) \leq p$. There exists a continuous function $f(t) = (X_f(t), Y_f(t))$ on $[0,1]$ such that:
	\begin{itemize}
		\item $f(0 = (X,Y)$.
		\item $XY^\top = X_f(t)(Y_f(t))^\top, \forall t \in [0,1]$.
		\item $X_f(1)$ or $Y_f(1)$ has full row rank.
	\end{itemize}
\end{lemma}
\begin{proof}
	WLOG, we assume that $m \leq r$. If $X$ has full row rank, then one can choose constant function $f(t) = (X,Y)$ to satisfy the conditions of the lemma. Therefore, we can focus on the case where $\texttt{rank}(X) = q < m$. WLOG, we can assume that the first $q$ columns of $X$ ($X_{1}, \ldots, X_{q}$) are linearly independent. The remaining columns of $X$ can be expressed as:
	\begin{equation*}
		X_{k} = \sum_{i = 1}^q \alpha^k_i X_{i}, \forall q < k \leq r
	\end{equation*}
	We define a matrix $\te{Y}$ by their columns as follow:
	\begin{equation*}
		\te{Y}_{i} = \begin{cases}
			Y_{i} + \sum_{k = q + 1}^r \alpha^k_i Y_{k} & \text{ if } i \leq q\\
			0 & \text{ otherwise}
		\end{cases}
	\end{equation*}
	By construction, we have $XY^\top = X\te{Y}^\top$. We define the function $f_1: [0, 1] \to \mb{R}^{m \times r} \times \mb{R}^{n \times r}$ as:
	\begin{equation*}
		f_1(t) = (X, (1 - t)Y + t\te{Y})
	\end{equation*}
	This function will not change the value of $f$ since we have: 
	\begin{equation*}
		X((1 - t)Y^\top + t\te{Y}^\top) = (1-t)XY^\top + t X\te{Y}^\top = XY^\top.
	\end{equation*}
	Let $\te{X}$ be a matrix whose first $q$ columns are identical to that of $X$ and $\texttt{rank}(\te{X}) = m$. The second function $f_2$ defined as:
	\begin{equation*}
		f_2(t)  = ((1 - t)X + t\te{X}, \te{Y})
	\end{equation*}
	also has their product unchanged (since first $q$ columns of $(1 - t)X + t\te{X}$ are constant and last $r - q$ rows of $\te{Y}$ are zero). Moreover, $f_2(0) = (\te{X}, \te{Y})$ where $\te{X}$ has full row rank. Therefore, the concatenation of two functions $f_1$ and $f_2$ (and shrink $t$ by a factor of $2$) are the desired function $f$.
\end{proof}

\begin{corollary}\label{cor:to_fullrank}
	Consider $I,J$ support constraints of \cref{eq: fixed_supp_prob} with $\mc{P}^\star = \mc{P} = \{P\}$. There is a feasible continuous function $f: [0,1] \mapsto \mb{R}^{m \times r} \times \mb{R}^{n \times r}: f(t) = (X_f(t), Y_f(t))$ such that:
	\begin{enumerate}[]
		\item $f(0) = (X,Y)$;
		\item $X_f(t)(Y_f(t))^{\top} = XY^\top, \forall t \in [0,1]$;
		\item $(X_f(1))_{R_P, P}$ or $(Y_f(1))_{C_P,P}$ has full row rank.
	\end{enumerate}
\end{corollary}
\begin{proof}[Proof of \cref{cor:to_fullrank}]
	WLOG, up to permuting columns, we can assume $P = \llbracket |P| \rrbracket, R_P = \llbracket |R_P|\rrbracket$ and $C_P =\llbracket |C_P| \rrbracket$ ($R_P$ and $C_P$ are defined in Definition \cref{def:completeeqclass}). A feasible function $f = (X_f(t), Y_f(t))$ has the form: 
	$$X_f(t) = \begin{pmatrix}	\te{X}_f(t) & \mathbf{0}\\ \mathbf{0} & \mathbf{0} \end{pmatrix}, Y_f(t) = \begin{pmatrix}	\te{Y}_f(t) & \mathbf{0}\\ \mathbf{0} & \mathbf{0} \end{pmatrix}$$ 
	where $\te{X}_f: [0,1] \mapsto \mb{R}^{R_P \times P}$, $\te{Y}_f: [0,1] \mapsto \mb{R}^{C_P \times P}$. 
	
	Since $P$ is a CEC, we have $p  \geq \min(R_P,C_P)$. Hence we can use \cref{lemma: to_inversible} to build $(\te{X}_f(t),\te{Y}_f(t))$ satisfying all conditions of \cref{lemma: to_inversible}. Such $(\te{X}_f(t),\te{Y}_f(t))$ fully determines $f$ and make $f$ our desirable function.
\end{proof}

\begin{proof}[Proof of \cref{lem:fullrankhypothesis}]
	First, we decompose $X$ and $Y$ as:
	\begin{equation*}
		X = X_{\bar{T}} + \sum_{P \in \mc{P}^\star} X_P, \qquad \qquad Y = Y_{\bar{T}} + \sum_{P \in \mc{P}^\star} Y_P 
	\end{equation*}
	Since $\bar{T}$ and $P \in \mc{P}^\star$ form a partition of $\intset{r}$, the product $XY^\top$ can be written as:
	\begin{equation*}
		XY^\top = X_{\bar{T}} Y_{\bar{T}}^\top + \sum_{P \in \mc{P}^\star} X_PY_P^\top.
	\end{equation*}
	For each $P \in \mc{P}^{\star}$, $(I_P, J_P)$ contains one CEC. By applying \cref{cor:to_fullrank}, we can build continuous functions $(X_f^P(t), Y_f^P(t))$, $\supp(X_f^{P}(t)) \subseteq I_P, \supp(Y_f^{P}(t)) \subseteq J_P, \forall t \in [0,1]$ such that:
	\begin{enumerate}
		\item $(X_f^{P}(0),Y_f^{P}(0)) = (X_P, Y_P)$.
		\item $X_f^{P}(t)(Y_f^{P}(t))^\top = X_PY_P^\top, \forall t \in [0,1]$.
		\item $(X_f^{P}(1))_{R_P, P}$ or $(Y_f^{P}(1))_{C_P,P}$ has full row rank.
	\end{enumerate}
	Our desirable $f(t) = (X_f(t), Y_f(t))$ is defined as:
	\begin{equation*}
		X_f(t) = X_{\bar{T}} + \sum_{P \in \mc{P}^\star} X_f^P(t), \qquad \qquad Y(t) = Y_{\bar{T}} + \sum_{P \in \mc{P}^\star} Y_f^P(t)
	\end{equation*}
	To conclude, it is immediate to check that $f = (X_f(t), Y_f(t))$ is feasible, $f(0) = (X,Y)$, $f(1)$ is CEC-full-rank and $X_f(t)Y_f(t)^\top = XY^\top, \forall t \in [0,1]$.
\end{proof}

\subsection{Proof of \cref{lem:connecttozeroCEC}}
\label{subapp:connecttozeroCEC}
Denote $Z = XY^\top$, we  construct $f$ such that $X_f(t)Y_f(t)^\top = B(t)$, where $B(t) = Z \odot {\bar{\mS}_T} + (At + Z(1 - t)) \odot {\mS_T}$. Such function $f$ makes $L(X_f(t),Y_f(t))$ non-increasing since:
\begin{equation}
	\label{eq:non-increasing}
	\begin{split}
		\|A - X_f(t)Y_f(t)^\top\|^2 &= \|A - B(t)\|^2\\
		&= \|(A - Z) \odot {\bar{\mS}_T}\|^2 + (1 - t)^2 \|(A - Z) \odot {\mS_T}\|^2
	\end{split}
\end{equation}
Thus, the rest of the proof is devoted to show that such a function $f$ exists by using \cref{lem: continous_function_of_S2}. Consider the function $g(t) = B(t) - X_{\bar{T}}(Y_{\bar{T}})^\top$. We have that $g(t)$ is continuous, $g(0) = B(0) - X_{\bar{T}}(Y_{\bar{T}})^\top = Z - X_{\bar{T}}(Y_{\bar{T}})^\top = X_{T}(Y_{T})^\top$ and:
\begin{equation*}
	\begin{split}
		g(t) \odot {\bar{\mS}_T} &= (B(t) - X_{\bar{T}}(Y_{\bar{T}})^\top) \odot {\bar{\mS}_T} \\
		&= (Z - X_{\bar{T}}(Y_{\bar{T}})^\top) \odot {\bar{\mS}_T} \\
		&= (X_TY_T^\top) \odot {\bar{\mS}_T}= \mathbf{0}
	\end{split}
\end{equation*}
which shows $\supp(g(t)) \subseteq  \mS_T$. Since $(X_T,Y_T)$ is CEC-full-rank (by our assumption, $(X,Y)$ is CEC-full-rank), invoking \cref{lem: continous_function_of_S2} with $(I_T, J_T)$, there exists $f^T(t) = (X^T_f(t),Y^T_f(t))$ such that:
\begin{enumerate}[label=\textbf{D\arabic*}]
	\item \label{eq:lemCEC1} $\supp(X_f^T(t)) \subseteq I_T, \supp(Y_f^C(t)) \subseteq J_T$.
	
	\item \label{eq:lemCEC2}  $f^T(0) = (X_T, Y_T)$.
	
	\item \label{eq:lemCEC3}  $g(t) = X_f^T(t)(Y_f^T(t))^\top, \forall t \in [0,1]$.
\end{enumerate}
We can define our desired function $f(t) = (X_f(t), Y_f(t))$ as:
\begin{equation*}
	X_f(t) = X_{\bar{T}} + X_f^T(t), \qquad \qquad Y = Y_{\bar{T}} + Y_f^T(t)
\end{equation*}
$f$ is clearly feasible due to \eqref{eq:lemCEC1}. The remaining condition to be checked is:
\begin{itemize}
	\item First condition:
	\begin{equation*}
		\begin{aligned}
			X_f(0) = X_f^T(0) + X_{\bar{T}} = X_T + X_{\bar{T}} = X, \quad
			Y_f(0) = Y_f^T(0) + Y_{\bar{T}} = Y_T + Y_{\bar{T}} = Y\\
		\end{aligned}
	\end{equation*}
	\item Second condition: holds thanks to \Cref{eq:non-increasing} and:
	\begin{equation*}
		\begin{split}
			X_f(t)(Y_f(t))^\top = X_{\bar{T}}Y_{\bar{T}}^\top + X_f^C(t)(Y_f^C(t))^\top = X_{\bar{T}}Y_{\bar{T}}^\top + g(t) = B(t)
		\end{split}
	\end{equation*}
	\item Third condition:
	\begin{align*}
		(A - X_f(1)(Y_f(1))^\top) \odot {\mS_T} &= (A - B(1)) \odot {\mS_T}\\
		&= (A - Z \odot {\bar{\mS}_T} - A \odot {\mS_T}) \odot {\mS_T}= \mathbf{0}
	\end{align*}
\end{itemize}

\subsection{Proof of \Cref{lem:connecttooptimal}}
\label{subapp:connecttooptimal}
Consider $X_T, X_{\bar{T}}^i, Y_T, Y_{\bar{T}}^i, i = 1, 2$ as in \cref{def:taxonomy2}. We redefine $A' = A \odot {\bar{\mS}_T}, I' = I_{\bar{T}}^1, J' = J_{\bar{T}}^1$ as in \cref{theorem: reduction_disjoint_overlapping}.

In light of \cref{cor:optimalsol}, an optimal solution $(\te{X}, \te{Y})$ has the following form:
\begin{itemize}
	\item[1)] $\te{X}_{\bar{T}}^1 = \te{X} \odot I_{\bar{T}}^1, \te{Y}_{\bar{T}}^1 = \te{Y} \odot J_{\bar{T}}^1$ is an optimal solution of \eqref{eq: fixed_supp_prob} with $(A', I', J')$. 
	\item[2)] $\te{X}_{\bar{T}}^2 = \te{X} \odot I_{\bar{T}}^2, \te{Y}_{\bar{T}}^2 = \te{Y} \odot J_{\bar{T}}^2$ can be arbitrary.
	\item[3)] $\te{X}_T = \te{X} \odot I_T, \te{Y}_T = \te{Y} \odot J_T$ satisfy:
	\begin{equation*}
		\te{X}_T\te{Y}_T^\top = (A - \sum_{(i,j) \neq (1,1)}\te{X}_{\bar{T}}^i\te{Y}_{\bar{T}}^j)^\top \odot {\mS_T}
	\end{equation*}
\end{itemize}
Since $(I',J')$ has its support constraints satisfying \cref{theorem:disjoint_totally_overlapping} assumptions as shown in \cref{theorem: reduction_disjoint_overlapping}, by \cref{theorem:noSpuriousSimple}, there exists a function $(X^{\bar{T}}_f(t), Y^{\bar{T}}_f(t))$ such that:
\begin{itemize}
	\item[1)] $\supp(X^{\bar{T}}_f(t)) \subseteq I_{\bar{T}}^1, \supp(Y_f^{\bar{T}}(t)) \subseteq J_{\bar{T}}^1$.
	\item[2)] $X_f^{\bar{T}}(0) = X_{\bar{T}}^1, Y_f^{\bar{T}}(0) = Y_{\bar{T}}^1$.
	\item[3)] $L'(X_f^{\bar{T}}(t), Y_f^{\bar{T}}(t)) = \|A' - X_f^{\bar{T}}(t)Y_f^{\bar{T}}(t)^\top\|^2$ is non-increasing.
	\item[4)] $(X_f^{\bar{T}}(1), Y_f^{\bar{T}}(1))$ is an optimal solution of the instance of \eqref{eq: fixed_supp_prob} with $(A', I', J')$.
\end{itemize}
Consider the function $g(t) = \left(A - (X_f^{\bar{T}}(t) + X_{\bar{T}}^2)(Y_f^{\bar{T}}(t) + Y_{\bar{T}}^2)^\top\right) \odot {\mS_T}$. This construction makes $g(0) = X_TY_T^\top$. Indeed, 
\begin{equation*}
	\begin{split}
		g(0) &= \left(A - (X_f^{\bar{T}}(0) + X_{\bar{T}}^2)(Y_f^{\bar{T}}(0) + Y_{\bar{T}}^2)^\top\right) \odot {\mS_T}\\
		&= \left(A - (X_{\bar{T}}^1 + X_{\bar{T}}^2)(Y_{\bar{T}}^1 + Y_{\bar{T}}^2)^\top\right) \odot {\mS_T}\\
		&\overset{(1)}{=} \left(XY^\top - (X_{\bar{T}}^1 + X_{\bar{T}}^2)(Y_{\bar{T}}^1 + Y_{\bar{T}}^2)^\top\right) \odot {\mS_T}\\
		&\overset{(2)}{=} X_TY_T^\top 
	\end{split}
\end{equation*} 
where (1) holds by the hypothesis $(A - XY^\top) \odot {\mS_T} = \mathbf{0}$, and (2) holds by Equation \eqref{eq:decomposition_lemma} and $\supp(X_TY_T^\top) \subseteq \mS_T$. Due to our hypothesis $(X,Y)$ is CEC-full-rank, $(X_T,Y_T)$ is CEC-full-rank. In addition, $g(t)$ continuous, $\supp(g(t)) \subseteq \mS_T$ and $g(0) = X_TY_T^\top$. Invoking \cref{lem: continous_function_of_S2} with $(I_T, J_T)$, there exist functions $(X_f^C(t), Y_f^C(t))$ satisfying:
\begin{itemize}[leftmargin = 25pt]
	\item[1)] $\supp(X_f^T(t)) \subseteq I_T, \supp(Y_f^T(t)) \subseteq J_T$.
	
	\item[2)] $f^T(0) = (X_T, Y_T)$.
	
	\item[3)] $g(t) = X_f^T(t)Y_f^T(t)^\top, \forall t \in [0,1]$.
\end{itemize}
Finally, one can define the function $X_f(t), Y_f(t)$ satisfying \cref{lem:connecttooptimal} as:
\begin{equation*}
	X_f(t) = X_f^{\bar{T}}(t) + X_f^C(t) + X_{\bar{T}}^2, \qquad \qquad Y_f(t) = Y_f^{\bar{T}}(t) + Y_f^C(t) + Y_{\bar{T}}^2
\end{equation*}
$f$ is feasible due to the supports of $X_f^P(t), Y_f^P(t), P \in \{{\bar{T}}, C\}$ and $X_{\bar{T}}^2, Y_{\bar{T}}^2$. The remaining conditions are satisfied as:
\begin{itemize}
	\item First condition:
	\begin{equation*}
		\begin{aligned}
			X_f(0) &=  X_f^{\bar{T}}(0) + X_f^C(0) + X_{\bar{T}}^2 = X_{\bar{T}}^1 + X_T + X_{\bar{T}}^2 = X\\
			Y_f(0) &=  Y_f^{\bar{T}}(0) + Y_f^C(0) + Y_{\bar{T}}^2 = Y_{\bar{T}}^1 + Y_T + Y_{\bar{T}}^2 = Y\\
		\end{aligned}
	\end{equation*}
	\item Second condition:
	\begin{equation*}
		\begin{split}
			&\|A - X_f(t)Y_f(t)^\top\|^2 = \|A - X_f^T(t)(Y_f^T(t))^\top - (X_f^{\bar{T}}(t) + X_{\bar{T}}^2)(Y_f^{\bar{T}}(t) + Y_{\bar{T}}^2)^\top\|^2\\
			&= \|g(t) - X_f^T(t)Y_f^T(t)^\top\|^2 + \|(A - X_f^{\bar{T}}(t)(Y_f^{\bar{T}}(t))^\top) \odot { \mS_\mc{P} \setminus \mS_T}\|^2 + \|A \odot {\bar{\mS}_\mc{P}}\|^2\\
			&= \|(A' - X_f^{\bar{T}}(t)(Y_f^{\bar{T}}(t))^\top) \odot { \mS_\mc{P} \setminus \mS_T}\|^2 + \|A \odot {\bar{\mS}_\mc{P}}\|^2\\
			&\overset{\eqref{eq:equivalent1}}{=} \|A' - X_f^{\bar{T}}(t)(Y_f^{\bar{T}}(t))^\top\|^2
		\end{split}
	\end{equation*}
	Since $\|A' - X_f^{\bar{T}}(t)(Y_f^{\bar{T}}(t))^\top\|^2$ is non-increasing, so is $\|A - X_f(t)Y_f(t)^\top\|^2$.
	\item Third condition: By \cref{theorem: reduction_disjoint_overlapping}, $(X_f(1), Y_f(1))$ is a global minimizer since $\|A - X_f(1)Y_f(1)^\top\|^2 = \|A' - X_f^{\bar{T}}(1)(Y_f^{\bar{T}}(1))^\top\|^2$ where $(X_f^{\bar{T}}(1), Y_f^{\bar{T}}(1))$ is an optimal solution of the instance of \eqref{eq: fixed_supp_prob} with $(A', I', J')$.
\end{itemize}

\subsection{Proof of \cref{theorem:nospuriousminimaComplex}}
\label{sec:proof_4}
The following corollary is necessary for the proof of \cref{theorem:nospuriousminimaComplex}.
\begin{corollary}
	\label{cor: bounded_operator}
	Consider $I, J$ support constraints of \eqref{eq: fixed_supp_prob}, such that $\mc{P}^\star = \mc{P}$. Given any feasible CEC-full-rank point $(X,Y)$ and any $B$ satisfying $\supp(B) \subseteq \mS_{\mc{P}}$, there exists $(\te{X},\te{Y})$ such that:
	\begin{enumerate}[label=\textbf{E\arabic*}]
		\item \label{eq:corbo1}$\supp(\te{X}) \subseteq I, \supp(\te{Y}) \subseteq J$
		
		\item \label{eq:corbo2}$\te{X}\te{Y}^\top = B$.
		
		\item \label{eq:corbo3}$\|X - \te{X}\|^2 + \|Y - \te{Y}\|^2 \leq \mc{C} \|XY^\top - B\|^2$. 
	\end{enumerate}
	where $\mc{C} =  \underset{P \in \mc{P}^\star}{\max} \left(\max\left(\vertiii{X_{R_{P}, P}^\dagger}^2, \vertiii{Y_{C_{P}, P}^\dagger}^2\right)\right)$.
\end{corollary}
\begin{proof}
	\Cref{cor: bounded_operator} is an application of \cref{lem: continous_function_of_S2}. Consider the function $g(t) = (1 - t)XY^\top + tB$. By construction, $g(t)$ is continuous, $g(0) = XY^\top$ and $\supp(g(t)) \subseteq \supp(XY^\top) \cup \supp(B) = \mS_{\mc{P}}$. Since $(X,Y)$ is CEC-full-rank, there exists a feasible function $f(t) = (X_f(t), Y_f(t))$ satisfying \ref{eq:lemS21} - \ref{eq:lemS23} by using \cref{lem: continous_function_of_S2}. 
	
	We choose $(\te{X},\te{Y}) = (X_f(1), Y_f(1))$. The verification of constraints is as follow:
	\begin{itemize}[label={}]
		\item \ref{eq:corbo1}: $f$ is feasible.
		\item \ref{eq:corbo2}: $\te{X}\te{Y}^\top = X_f(1)Y_f(1)^\top \overset{\ref{eq:lemS22}}{=} g(1) = B$.
		\item \ref{eq:corbo3}: $\|X - \te{X}\|^2 + \|Y - \te{Y}\|^2 \overset{\ref{eq:lemS21}}{=} \|f(1) - f(0)\|^2 \overset{\ref{eq:lemS23}}{\leq} \mc{C} \|g(0) - g(1)\|^2 \leq \mc{C} \|XY^\top - B\|^2$. 
	\end{itemize}
\end{proof}
\begin{proof}[Proof of \cref{theorem:nospuriousminimaComplex}]
	As mentioned in the sketch of the proof, given any $(X,Y)$ not CEC-full-rank, \cref{lem:fullrankhypothesis} shows the existence of a path $f$ along which $L$ is constant and $f$ connects $(X,Y)$ to some CEC-full-rank $(\te{X},\te{Y})$. Therefore, this proof will be entirely devoted to show that a feasible CEC-full-rank solution $(X,Y)$ cannot be a spurious local minimum. This fact will be shown by the two following steps:
	\begin{itemize}[label={}, leftmargin = 0pt]
		\item \textbf{FIRST STEP}: Consider the function $L(X,Y)$, we have:
		\begin{equation*}
			L(X,Y) = \|A - XY^\top\|^2 = \|A - \sum_{P' \in \mc{P}^\star} X_{P'}Y_{P'}^\top - X_{\bar{T}}Y_{\bar{T}}^\top\|^2
		\end{equation*}
		
		If $(X,Y)$  is truly a local minimum, then $\forall P \in \mc{P}^\star$, $(X_P, Y_P)$ is also the local minimum of the following function:
		\begin{equation*}
			L'(X_P, Y_P) =  \|(A - \sum_{P' \neq P} X_{P'}Y_{P'}^\top - X_{\bar{T}}Y_{\bar{T}}^\top) - X_PY_P^\top\|^2
		\end{equation*}
		where $L'$ is equal to $L$ but we optimize only w.r.t $(X_P, Y_P)$ while fixing the other coefficients. In other words, $(X_P, Y_P)$ is a local minimum of the problem:
		\begin{equation*}
			\begin{aligned}
				& \underset{X' \in \mb{R}^{m \times r}, Y' \in \mb{R}^{n \times r}}{\text{Minimize}} 
				& &L'(X', Y') = \|B - X'Y'^\top\|^2&\\
				& \text{Subject to: }
				& &\texttt{supp}(X') \subseteq I_P \text{ and } \texttt{supp}(Y') \subseteq J_P&\\
			\end{aligned}
		\end{equation*}
		where $B = A - \sum_{P' \neq P} X_{P'}Y_{P'}^\top - X_{\bar{T}}Y_{\bar{T}}$. Since all columns of $I_{P}$ (resp. of $J_{P}$) are identical, all rank-one contribution supports are totally overlapping. Thus, all local minima are global minima (\cref{theorem:noSpuriousSimple}). Global minima are attained when
		$X_PY_P^\top = B \odot {\mS_{P}}$ due to the expressivity of a CEC (\cref{lem: expressibility_two}). Thus, for any $P \in \mc{P}^\star$, $\forall (i,j) \in \mS_P$, we have:
		\begin{equation*}
			0 = (B - X_PY_P^\top)_{i,j} = (A - \sum_{P' \in \mc{P}^\star} X_{P'}Y_{P'}^\top - X_{\bar{T}}Y_{\bar{T}}^\top)_{i,j} = (A - XY^\top)_{i,j}
		\end{equation*}
		which implies \Cref{eq: special_critical_point}.
		
		\item \textbf{SECOND STEP}: In this step, we assume that \Cref{eq: special_critical_point} holds. Consider $X_T, X_{\bar{T}}^i, Y_T, Y_{\bar{T}}^i, i = 1, 2$ as in \cref{def:taxonomy}. Let $A' = A \odot {\bar{\mc{S}}_T}, I' = I_{\bar{T}}^1, J' = J_{\bar{T}}^1$.
		
		We consider two possibilities. First, if $(X_{\bar{T}}^1, Y_{\bar{T}}^1)$ is an optimal solution of the instance of \eqref{eq: fixed_supp_prob} with $(A', I', J')$, by \cref{cor:optimalsol}, $(X,Y)$ is an optimal solution of \eqref{eq: fixed_supp_prob} with $(A, I, J)$ (since \Cref{eq: special_critical_point} holds). Hence it cannot be a spurious local minimum. We now focus on the second case, where $(X_{\bar{T}}^1, Y_{\bar{T}}^1)$ is \emph{not} the optimal solution of the instance of \eqref{eq: fixed_supp_prob} with $(A', I', J')$. We show that in this case, in any neighborhood of $(X,Y)$, there exists a point $(X', Y')$ such that {$\supp(X') \subseteq I$, $\supp(Y') \subseteq J'$ and} $L(X,Y) > L(X',Y')$. Thus $(X,Y)$ cannot be a local minimum.
		
		Since $(I_{\bar{T}}^1, J_{\bar{T}}^1)$ satisfies \cref{theorem:disjoint_totally_overlapping} assumptions, \eqref{eq: fixed_supp_prob} has no spurious local minima (\cref{theorem:noSpuriousSimple}). As $(X_{\bar{T}}^1, Y_{\bar{T}}^1)$ is not an optimal solution, it cannot be a local minimum either, i.e., in any neighborhood of $(X_{\bar{T}}^1, Y_{\bar{T}}^1)$, there exists
		$(\dot{X},\dot{Y})$ with $\supp(\te{X}_{\bar{T}}^1) \subseteq I', \supp(\te{Y}_{\bar{T}}^1) \subseteq J'$ and
		\begin{equation}
			\label{eq:better_neighbors}
			\|A' - X_{\bar{T}}^1(Y_{\bar{T}}^1)^\top\|^2 > \|A'  -
			\te{X}_{\bar{T}}^1(\te{Y}_{\bar{T}}^1)^\top\|^2
		\end{equation}
		By \Cref{eq:equivalent1}, we have:
		\begin{equation}
			\label{eq:expansion}
			\begin{aligned}
				\|A'-(X_{\bar{T}}^1)(Y_{\bar{T}}^1)^\top\|^2 &= \|(A-(X_{\bar{T}}^1)(Y_{\bar{T}}^1)^\top) \odot { \mS_\mc{P} \setminus \mS_T}\|^{2}+ \|A \odot {\bar{\mS}_\mc{P}}\|^{2}\\
				\|A'-(\te{X}_{\bar{T}}^1)(\te{Y}_{\bar{T}}^1)^\top\|^2 &= \|(A-(\te{X}_{\bar{T}}^1)(\te{Y}_{\bar{T}}^1)^\top) \odot { \mS_\mc{P} \setminus \mS_T}\|^{2}+ \|A \odot {\bar{\mS}_\mc{P}}\|^{2}\\
			\end{aligned}
		\end{equation}
		By \Cref{eq:better_neighbors} and \Cref{eq:expansion} we have:
		\begin{equation}
			\label{eq:eqforproof1}
			\|(A -(X_{\bar{T}}^1)(Y_{\bar{T}}^1)^\top) \odot { \mS_\mc{P} \setminus \mS_T}\|^{2} > \|(A -\te{X}_{\bar{T}}^1(\te{Y}_{\bar{T}}^1)^\top)
			\odot { \mS_\mc{P} \setminus \mS_T}\|^{2}
		\end{equation}
		Consider the matrix: $B := \left(A -  (\te{X}_{\bar{T}}^1+X_{\bar{T}}^2) (\te{Y}_{\bar{T}}^1 + Y_{\bar{T}}^2)^\top \right) \odot { \mS_T}$. Since $\supp(B) \subseteq  \mS_T$ and $(X_T,Y_T)$ is CEC-full-rank (we assume $(X,Y)$ is CEC-full-rank), by \cref{cor: bounded_operator}, there exists ($\te{X}_T,\te{Y}_T$) such that:
		\begin{itemize}
			\item[1)] $\supp(\te{X}_T) \subseteq I_T,  \supp(\te{Y}_T) \subseteq J_T$.
			\item[2)]  $\te{X}_T\te{Y}_T^\top = B$.
			
			\item[3)]  $\|X_T - \te{X}_T\|^2 + \|Y_T - \te{Y}_T\|^2 \leq \mc{C} \|X_TY_T^\top - B\|^2$. 
		\end{itemize}
		where $\mc{C} =  \underset{P \in \mc{P}^\star}{\max} \left(\max\left(\vertiii{X_{R_{P}, P}^\dagger}^2, \vertiii{Y_{C_{P}, P}^\dagger}^2\right)\right)$. We define the point$(\te{X},\te{Y})$ as:
		\begin{equation*}
			\te{X} = \te{X}_T + \te{X}_{\bar{T}}^1 + X_{\bar{T}}^2, \qquad \qquad \te{Y} = \te{Y}_T + \te{Y}_{\bar{T}}^1 + Y_{\bar{T}}^2
		\end{equation*}
		The point $(\te{X},\te{Y})$ still satisfies \Cref{eq: special_critical_point}. Indeed,
		\begin{equation}
			\label{eq:eqforproof2}
			\begin{split}
				(A - \te{X}\te{Y}^\top) \odot { \mS_T} &= \left(A - \te{X}_T\te{Y}_T^\top - (\te{X}_{\bar{T}}^1 + X_{\bar{T}}^2)(\te{Y}_{\bar{T}}^1 + Y_{\bar{T}}^2)^\top
				\right) \odot { \mS_T}\\
				&= (B - \te{X}_T\te{Y}_T^\top) \odot { \mS_T} = \mathbf{0}.
			\end{split}
		\end{equation}
		It is clear that $(\te{X}, \te{Y})$ satisfies $\supp(\te{X}) \subseteq I$, $\supp(\te{Y}) \subseteq J$ due to the support of its components $(\te{X}_T, \te{Y}_T), (\te{X}_{\bar{T}}^1, \te{Y}_{\bar{T}}^1), (X_{\bar{T}}^2, Y_{\bar{T}}^2)$. Moreover, we have:
		\begin{equation*}
			\begin{aligned}
				\|A-\te{X}\te{Y}^\top\|^{2} &= \|(A-\te{X}\te{Y}^\top) \odot {\mS_T}\|^{2} + \|(A-\te{X}\te{Y}^\top) \odot { \mS_\mc{P} \setminus \mS_T}\|^{2}+ \|A \odot {\bar{\mS}_\mc{P}}\|^{2}\\
				&\overset{\eqref{eq:eqforproof2}}{=} \|(A  - \te{X}_{\bar{T}}^1(\te{Y}_{\bar{T}}^1)^\top) \odot { \mS_\mc{P} \setminus \mS_T}\|^2+ \|A \odot {\bar{\mS}_\mc{P}}\|^{2}\\
				&\overset{ \eqref{eq:eqforproof1}}{<} \|(A - X_{\bar{T}}^1(Y_{\bar{T}}^1)^\top) \odot { \mS_\mc{P} \setminus \mS_T}\|^2 + \|A \odot {\bar{\mS}_\mc{P}}\|^{2}\\
				&= \|A-XY^\top\|^{2}.
			\end{aligned}
		\end{equation*}
		Lastly, we show that $(\te{X},\te{Y})$ can be chosen arbitrarily close to $(X,Y)$ by choosing $(\te{X}_{\bar{T}}^1,\te{Y}_{\bar{T}}^1)$ close enough to $(X_{\bar{T}}^1, Y_{\bar{T}}^1)$. For this, denoting $\epsilon := \|X_{\bar{T}}^1 - \tilde{X}\|^2 + \|Y_{\bar{T}}^1 - \tilde{Y}\|^2$, we first compute:
		\begin{equation*}
			\begin{split}
				\|X - \te{X}\|^2 + \|Y  - \te{Y}\|^2 &= \|X_T - \te{X}_T\|^2 + \|Y_T - \te{Y}_T\|^2 + \|X_{\bar{T}}^1 - \te{X}_{\bar{T}}^1\|^2 + \|Y_{\bar{T}}^1 - \te{Y}_{\bar{T}}^1\|^2\\
				&\leq \mc{C}\|X_TY_T^\top - B\|^2 + \epsilon\\
			\end{split}
		\end{equation*}
		We will bound the value $\|X_TY_T^\top - B\|^2$. By using \Cref{eq: special_critical_point}, we have:
		\begin{equation*}
			\begin{split}
				(A - \sum_{1 \leq i,j \leq 2} (X_{\bar{T}}^i)(Y_{\bar{T}}^j)^\top) \odot  {\mc{S}_T} - X_TY_T^\top &= (A - X_TY_T^\top - \sum_{1 \leq i,j \leq 2} (X_{\bar{T}}^i)(Y_{\bar{T}}^j)^\top) \odot  {\mc{S}_T}\\
				&= (A - XY^\top) \odot {\mc{S}_T} \overset{\eqref{eq: special_critical_point}}{=} \mbf{0}
			\end{split}
		\end{equation*}
		Therefore, $X_TY_T^\top = [A - (X_{\bar{T}}^1 + X_{\bar{T}}^2)(Y_{\bar{T}}^1 + Y_{\bar{T}}^2)^\top] \odot { \mS_T}$. We have:
		\begin{equation*}
			\begin{split}
				\|X_TY_T^\top - B\|^2 &= 
				\|[A - (X_{\bar{T}}^1 + X_{\bar{T}}^2)(Y_{\bar{T}}^1 + Y_{\bar{T}}^2)^\top] \odot { \mS_T} - B\|
				^{2}\\ 
				&= 
				\|[(\tilde{X}_{\bar{T}}^1 + X_{\bar{T}}^2)(\tilde{Y}_{\bar{T}}^1 + Y_{\bar{T}}^2)^\top - (X_{\bar{T}}^1 + X_{\bar{T}}^2)(Y_{\bar{T}}^1 + Y_{\bar{T}}^2)^\top] \odot { \mS_T}\|^{2}\\ 
				&\leq \|(\tilde{X}_{\bar{T}}^1 + X_{\bar{T}}^2)(\tilde{Y}_{\bar{T}}^1 + Y_{\bar{T}}^2)^\top - (X_{\bar{T}}^1 + X_{\bar{T}}^2)(Y_{\bar{T}}^1 + Y_{\bar{T}}^2)^\top\|^{2} 
			\end{split}
		\end{equation*}
		When $\epsilon \to 0$, we have $\|(\tilde{X}_{\bar{T}}^1 + X_{\bar{T}}^2)(\tilde{Y}_{\bar{T}}^1 + Y_{\bar{T}}^2)^\top - (X_{\bar{T}}^1 + X_{\bar{T}}^2)(Y_{\bar{T}}^1 + Y_{\bar{T}}^2)^\top\| \to 0$. Therefore, with $\epsilon$ small enough, one have $\|X - X'\|^2 + \|Y  - Y'\|^2$ can be arbitrarily small. This concludes the proof.
	\end{itemize}
\end{proof}
\subsection{Proof for \cref{ex:spuriousinstances}}
\label{sec:proof_ex}
Direct calculation of the Hessian of $L$ at point $(X_0,Y_0)$ is given by:
\begin{equation*}
	H(L)_{\mid (X_0, Y_0)} = \begin{pmatrix}
		0 & 0 & 0 & 0 & 0 & 0 & 0 \\
		0 & 100 & 0 & 0 & 0 & 10 & 0 \\
		0 & 0 & 100 & 0 & 0 & 0 & -1 \\
		0 & 0 & 0 & 0 & 0 & 0 & 0 \\
		0 & 0 & 0 & 0 & 0 & 0 & 0 \\
		0 & 10 & 0 & 0 & 0 & 1 & 0 \\
		0 & 0 & -1 & 0 & 0 & 0 & 1 \\
	\end{pmatrix}
\end{equation*}
which is indeed positive semi-definite.

\section{Expressing any hierarchically off-diagonal low-rank matrix (HODLR)  as a product of $2$ factors with fixed supports}
\label{app:hierarchical}
In the following, we report the definition of HODLR matrices. For convenience, we report the definition only for a \emph{square} matrix whose size is a power of two, i.e $n = 2^J, J \in \mb{N}$. 
\begin{definition}[HODLR matrices]
	A matrix $A \in \mb{R}^{2^N \times 2^N}$ is called an HODLR matrix if either of the following two holds:
	\begin{itemize}[leftmargin = *]
		\item $N = 0$, i.e., $A \in \mb{R}^{1 \times 1}$.
		\item $A$ has the form $A = \bigl[\begin{smallmatrix} A_{11} & A_{12} \\ A_{21} & A_{22} \end{smallmatrix}\bigr]$ for $A_{i,j} \in \mb{R}^{2^{N - 1} \times 2^{N-1}}, 1 \leq i,j \leq 2$ such that $A_{21}, A_{12}$  are of rank at most one and $A_{11}, A_{22} \in \mb{R}$ are HODLR matrices. 
	\end{itemize}
\end{definition}
We prove that any HODLR matrix is a product of two factors with fixed support. The result is proved when $A_{12},A_{21}$  are of rank at most one, but more generally, if we allow $A_{12}$ and $A_{21}$ to have rank $k \geq 1$, the general scheme of the proof of \Cref{lem: hierarchical} below still works (with the slight modification $|I| = |J| = O(kn\log n)$, $I, J \in \{0,1\}^{2^N \times k(3 \times 2^N -2)}$).
We prove that any HODLR matrix is a product of two factors with fixed support.
\begin{lemma}
	\label{lem: hierarchical}
	For each $N \geq 1$ there exists $I,J \in \{0,1\}^{2^N \times  (3 \times 2^N - 2)}$ support constraints such that for any HODLR matrix $A \in \mb{R}^{2^N \times 2^N}$, we have:
	\begin{itemize}[leftmargin=*]
		\item[1)] $A$ admits a factorization $XY^\top$ and $\supp(X) \subseteq I, \supp(Y) \subseteq J$.
		\item[2)] $|I| = |J| = O(n\log n)$ ($n = 2^N$). 
		\item[3)] $(I, J)$ satisfies the assumption of \Cref{theorem:disjoint_totally_overlapping}.
	\end{itemize}
\end{lemma}
\begin{proof}
	The proof is carried out by induction. 
	\begin{itemize}[leftmargin=*]
		\item[1)] For $N = 1$, one can consider $(I,J) \in \{0,1\}^{2 \times 4} \times  \{0,1\}^{2 \times 4}$ defined (in the binary matrix form) as follows:
		\begin{equation*}
			I = \begin{pmatrix}
				1 & 0 & 1 & 0\\
				0 & 1 & 0 & 1
			\end{pmatrix}, \;J = \begin{pmatrix}
				0 & 1 & 1 & 0\\
				1 & 0 & 0 & 1
			\end{pmatrix}.
		\end{equation*} 
		Any $(X,Y)$ constrained to $(I,J)$ will have the following form:
		\begin{equation*}
			X = \begin{pmatrix}
				x_1 & 0 & x_3 & 0\\
				0 & x_2 & 0 & x_4
			\end{pmatrix}, \;Y = \begin{pmatrix}
				0 & y_2 & y_3 & 0\\
				y_1 & 0 & 0 & y_4\\
			\end{pmatrix}, \;XY^\top = \begin{pmatrix}
				x_3y_3 & x_1y_1\\
				x_2y_2 & x_4y_4
			\end{pmatrix}.
		\end{equation*}
		Given any matrix $A \in \mb{R}^{2 \times 2}$ (and in particular, given any HODLR matrix in this dimension) it is easy to see that $A$ can be represented as $XY^\top$ such that $\supp(X) \subseteq I, \supp(Y) \subseteq J$ (take e.g. $x_{3}=a_{11}$, $x_{1}=a_{12}$, $x_{2}=a_{21}$, $x_{4}=a_{22}$ and all $y_{i}=1$). It is also easy to verify that this choice of $(I,J)$ makes all the supports of the rank-one contributions pairwise disjoint, so that the assumptions of \Cref{theorem:disjoint_totally_overlapping} are fulfilled. Finally, we observe that $|I_{N}|=|J_{N}|=4$.
		\item[2)] Suppose that our hypothesis is correct for $N - 1$, we need to prove its correctness for $N$. Let $(I_{N-1}, J_{N-1})$ be the pair of supports for $N-1$, we construct $(I_N, J_N)$ (still in binary matrix form) as follows:
		\begin{align*}
			I_N &= \begin{pmatrix}
				\mbf{1}_{n/2 \times 1} & \mbf{0}_{n/2 \times 1} & I_{N-1} & \mbf{0}_{n/2 \times {\color{red} (3}n/2{\color{red}-2)}} \\
				\mbf{0}_{n/2 \times 1} & \mbf{1}_{n/2 \times 1} & \mbf{0}_{n/2 \times {\color{red} (3}n/2{\color{red}-2)}} & I_{N-1}\\
			\end{pmatrix}\\
			J_N &= \begin{pmatrix}
				\mbf{0}_{n/2 \times 1} & \mbf{1}_{n/2 \times 1} & J_{N-1} & \mbf{0}_{n/2 \times {\color{red} (3}n/2{\color{red}-2)}} \\
				\mbf{1}_{n/2 \times 1} & \mbf{0}_{n/2 \times 1} & \mbf{0}_{n/2 \times {\color{red} (3}n/2{\color{red}-2)}} & J_{N-1}\\
			\end{pmatrix}
		\end{align*}
		where $n = 2^N$ and $\mbf{1}_{p \times q}$ (resp. $\mbf{0}_{p \times q}$) is the matrix of size $p \times q$ full of ones (resp. of zeros). Since $I_{N-1}$ and $J_{N-1}$ are both of dimension $2^{N-1} \times (3 \times 2^{N-1}-2) = (n/2) (3n/2-2)$, the dimensions of $I_N$ and $J_N$ are both equal to $(n, 2 \times (3n/2 - 2) + 2) = (n, 3n - 2)$. Moreover, the cardinalities of $I_N$ and $J_N$ satisfy the following recursive formula:
		\begin{equation*}
			|I_N| = n + 2|I_{N-1}|, \qquad |J_N| = n + 2|J_{N-1}|,
		\end{equation*}
		which justifies the fact that $|I_N| = |J_N| = O(n\log n)$. Finally, any factors $(X,Y)$ respecting the support constraints $(I_N, J_N)$  need to have the following form:
		\begin{align*}
			X &= \begin{pmatrix}
				X_{1} & \mbf{0}_{n/2 \times 1} & X_3 & \mbf{0}_{n/2 \times {\color{red} (3}n/2{\color{red}-2)}} \\
				\mbf{0}_{n/2 \times 1} & X_2 & \mbf{0}_{n/2 \times {\color{red} (3}n/2{\color{red}-2)}} & X_4\\
			\end{pmatrix}\\
			Y &= \begin{pmatrix}
				\mbf{0}_{n/2 \times 1} &Y_2 & Y_3 & \mbf{0}_{n/2 \times {\color{red} (3}n/2{\color{red}-2)}} \\
				Y_1 & \mbf{0}_{n/2 \times 1} & \mbf{0}_{n/2 \times {\color{red} (3}n/2{\color{red}-2)}} & Y_4\\
			\end{pmatrix}
		\end{align*}
		where $X_i, Y_i \in \mb{R}^{n/2}, 1 \leq i \leq 2$, and for $3 \leq j \leq 4$ we have $X_j, Y_j \in \mb{R}^{n/2 \times {\color{red} (3}n/2{\color{red}-2)}}$, $\supp(X_j) \subseteq I_{N-1}, \supp(Y_j) \subseteq J_{N-1}$. Their product yields:
		\begin{equation*}
			XY^\top = \begin{pmatrix}
				X_3Y_3^\top & X_1Y_1^\top\\
				X_2Y_2^\top & X_4Y_4^\top
			\end{pmatrix}.
		\end{equation*}
		Given an HODLR matrix $A \in \mb{R}^{n \times n}$, since $A_{12},A_{21} \in \mb{R}^{n/2 \times n/2}$ are of rank at most one, one can find $X_i, Y_i \in \mb{R}^{n/2}, 1 \leq i \leq 2$ such that $A_{12} = X_1Y_1^\top, A_{21} = X_2Y_2^\top$. Since $A_{11},A_{22} \in  \mb{R}^{n/2 \times n/2}$ are HODLR, by the induction hypothesis, one can also find $X_i, Y_i \in \mb{R}^{n/2 \times {\color{red} (3}n/2{\color{red}-2)}}$, $3 \leq i \leq 4$ such that $\supp(X_{i}) \subseteq I_{N-1}$, $\supp(Y_{i}) \subseteq I_{N-1}$ and $A_{11} = X_3Y_3^\top, A_{22} = X_4Y_4^\top$. Finally, this construction also makes all the supports of the rank-one contributions pairwise disjoint: the first two rank-one supports are $\mc{S}_1 = \{n/2+1, \ldots, n\} \times \intset{n/2}, \mc{S}_2 =  \intset{n/2} \times \{n/2+1, \ldots, n\}$, and the remaining ones are inside $\intset{n/2} \times \intset{n/2}$ and $\{n/2+1, \ldots, n\} \times \{n/2+1, \ldots, n\}$ which are disjoint by the induction hypothesis.
	\end{itemize}	
\end{proof}
\end{document}